\documentclass[a4paper]{article}
\pdfoutput=1
\usepackage[top=2.8cm, bottom=3cm, left=3.2cm , right=3.2cm]{geometry}

\usepackage[cmex10,fleqn]{amsmath}
\usepackage{amsfonts,amssymb,amstext,amsthm}
\usepackage{hyperref}
\usepackage{authblk}
\usepackage{blkarray}

\usepackage{xcolor}
\usepackage{multirow}
\usepackage[vlined,linesnumbered]{algorithm2e}
\usepackage[export]{adjustbox}

\DeclareMathOperator{\argmin}{argmin}
\DeclareMathOperator{\argmax}{argmax}

\newcommand{\comb}[2]{C_{#1}^{#2}}
\newcommand{\arrang}[2]{#2!\,C_{#1}^{#2}}

\definecolor{dark-gray}{rgb}{0.4,0.4,0.4}

\newtheorem{definition}{Definition}
\newtheorem{proposition}{Proposition}
\newtheorem{lemma}{Lemma}
\newtheorem{theorem}{Theorem}
\newtheorem{example}{Example}
\newtheorem{problem}{Problem}
\newtheorem{corollary}{Corollary}

\title{Linear Sum Assignment with Edition}

\author{S{\'e}bastien Bougleux}
\author{Luc Brun}
\affil{Normandie Universit{\'e}\\GREYC UMR 6072\\ CNRS - Universit{\'e} de Caen Normandie - ENSICAEN\\Caen, France}

\begin{document}
\maketitle
\begin{abstract}
  We consider the problem of transforming a set of elements into
  another by a sequence of elementary edit operations, namely
  substitutions, removals and insertions of elements. Each possible
  edit operation is penalized by a non-negative cost and the cost of a
  transformation is measured by summing the costs of its operations. A
  solution to this problem consists in defining  a transformation having a
  minimal cost, among all possible transformations. To compute such a
  solution, the classical approach consists in representing removal
  and insertion operations by augmenting the two sets so that they
  get the same size. This allows to express the problem as a linear
  sum assignment problem (LSAP), which thus finds an optimal bijection
  (or permutation, perfect matching) between the two augmented
  sets. While the LSAP is known to be efficiently solvable in
  polynomial time complexity, for instance with the Hungarian
  algorithm, useless time and memory are spent to treat the elements
  which have been added to the initial sets. In this report, we show
  that the problem can be formalized as an extension of the LSAP which
  considers only one additional element in each set to represent
  removal and insertion operations. A solution to the problem is no
  longer represented as a bijection between the two augmented
  sets. We show that the considered problem is a binary linear program
  (BLP) very close to the LSAP. While it can be solved by any BLP
  solver, we propose an adaptation of the Hungarian algorithm which
  improves the time and memory complexities previously obtained by the
  approach based on the LSAP. The importance of the improvement
  increases as the size of the two sets and their absolute difference
  increase. Based on the analysis of the problem presented in this
  report, other classical algorithms can be adapted.
\end{abstract}

\section{Introduction}\label{sec-intro}
Assigning the elements of some set to the elements of another, such that the resulting assignment satisfy some optimality conditions, is a fundamental combinatorial problem encountered in a wide variety of scientific fields \cite{law76,bur09}. This report focusses on linear sum assignment problems, also known as weighted bipartite graph matching problems. The simplest example assumes that the two sets have the same cardinality.
\begin{problem}[LSAP]
\label{pb:lsap}
  Provided two finite sets $\mathcal{U}$ and $\mathcal{V}$ having the
  same cardinality $n$, assigning the $n$ elements of $\mathcal{U}$ to
  the $n$ elements of $\mathcal{V}$ can be described by a bijective
  mapping $\mathcal{U}\,{\rightarrow}\,\mathcal{V}$. The task of
  assigning an element of $\mathcal{U}$ to an element of $\mathcal{V}$
  is penalized by a real non-negative cost function
  $c\,{:}\,\mathcal{U}\,{\times}\,\mathcal{V}\,{\rightarrow}\,[0,+\infty)$. The
  Linear Sum Assignment Problem (LSAP) consists in finding an optimal
  bijection
\begin{equation}\label{eq:lsap}
\hat{\varphi}\in\underset{\varphi\,{\in}\,\mathcal{B}(\mathcal{U},\mathcal{V})}{\argmin}~\left\{A(\varphi,c)\overset{\text{def.}}{=}\sum_{u\in\mathcal{U}}c(u,\varphi(u))\right\},
\end{equation}
where $\mathcal{B}(\mathcal{U},\mathcal{V})$ is the set of bijections from $\mathcal{U}$ to $\mathcal{V}$.
\end{problem}
\noindent
Note that several optimal assignments may exist, which depends on the cost function. 
Recall that any bijection is in one-to-one correspondence with a permutation of $\{1,\ldots,n\}$, and that any permutation $\varphi$ can be represented by a matrix $\mathbf{X}\,{\in}\,\{0,1\}^{n\times n}$ satisfying $x_{i,j}\,{=}\,1$ if $\varphi(i)\,{=}\,j$ and $x_{i,j}\,{=}\,0$ else. Such a $n\,{\times}\,n$ matrix, called permutation matrix, satisfies
\begin{equation}\label{def-permat}
 	\mathbf{X}\,{\in}\,\{0,1\}^{n\times n},~~\mathbf{X}\mathbf{1}_n\,{=}\,\mathbf{1}_n,~\,\mathbf{X}^T\mathbf{1}_n\,{=}\,\mathbf{1}_n.
 \end{equation}
Let $\mathcal{P}_n$ be the set of all $n\,{\times}\,n$ permutation matrices. Then the LSAP can be reformulated as
\begin{equation}\label{eq-lsapmtx}
 	\hat{\mathbf{X}}\in\underset{\mathbf{X}\in\mathcal{P}_n}{\argmin}\left\{A(\mathbf{X},\mathbf{C})\overset{\text{def.}}{=}\sum_{i=1}^n\sum_{j=1}^nx_{i,j}c_{i,j}\right\},
\end{equation}
where $\mathbf{C}\,{=}\,(c_{i,j})_{i,j=1,\ldots,n}$ is the cost matrix
representing the cost function, that is $c_{i,j}\,{=}\,c(u_i,v_j)$ with $u_i\,{\in}\,\mathcal{U}$ and $v_j\,{\in}\,\mathcal{V}$. From Eq.~\ref{eq-lsapmtx}, the LSAP can be easily reformulated as a binary linear program \cite{sier15}
\begin{equation}\label{eq-lsaplp}
 	\hat{\mathbf{x}}\in\argmin\left\{\mathbf{c}^T\mathbf{x}~|~\mathbf{L}\mathbf{x}\,{=}\,\mathbf{1}_{n^2},~\mathbf{x}\,{\in}\,\{0,1\}^{n^2}\right\},
\end{equation}
where $\mathbf{x}\,{=}\,\text{vec}(\mathbf{X})$ and
$\mathbf{c}\,{=}\,\text{vec}(\mathbf{C})$ are the vectorization of the
permutation matrix $\mathbf{X}$ and the cost matrix $\mathbf{C}$,
respectively. The right-hand side of Eq.~\ref{eq-lsaplp} is the matrix
version of the constraints defined by Eq.~\ref{def-permat}. The matrix
$\mathbf{L}\,{\in}\,\{0,1\}^{2n\times n^2}$ corresponds to the
node-edge incidence matrix of the complete bipartite graph $K_{n,n}$
with node sets $\mathcal{U}$ and $\mathcal{V}$, \textit{i.e.}
$(\mathbf{L})_{k,(i,j)}\,{=}\,1$ if $(k=i)\vee (k=j)$ and $0$ else. A
solution to the LSAP corresponds to a perfect bipartite matching,
\textit{i.e.} a subgraph of $K_{n,n}$ such that each element of
$\mathcal{U}\,{\cup}\,\mathcal{V}$ has degree one.

The LSAP is efficiently solvable. While there is $n!$ possible assignments, it can be solved in polynomial time
complexity (worst-case), for instance in $O(n^3)$ with the well-known
Hungarian algorithm \cite{kuhn55,kuhn56,munk57} or its improvements, see~\cite{law76,bur09} for more details.

When the two sets $\mathcal{U}$ and $\mathcal{V}$ have different
cardinalities, $|\mathcal{U}|\,{=}\,n$ and $|\mathcal{V}|\,{=}\,m$
with $n\,{\leq}\,m$, an assignment becomes an injection
$\varphi\,{:}\,\mathcal{U}\,{\rightarrow}\,\mathcal{V}$, or equivalently an
arrangement of $n$ elements among $m$. There is
$\arrang{m}{n}\,{=}\,m!/(m-n)!$ such assignments. Under this context, a solution to the LSAP
is an injection minimizing the objective functional
defined in Eq.~\ref{eq:lsap}.
The Hungarian algorithm has been modified to solve this LSAP in $O(n^2m)$ time complexity~\cite{bour71}.

\begin{figure}[!t]
 	\begin{tabular}{ccc}
 	\includegraphics[scale=0.45]{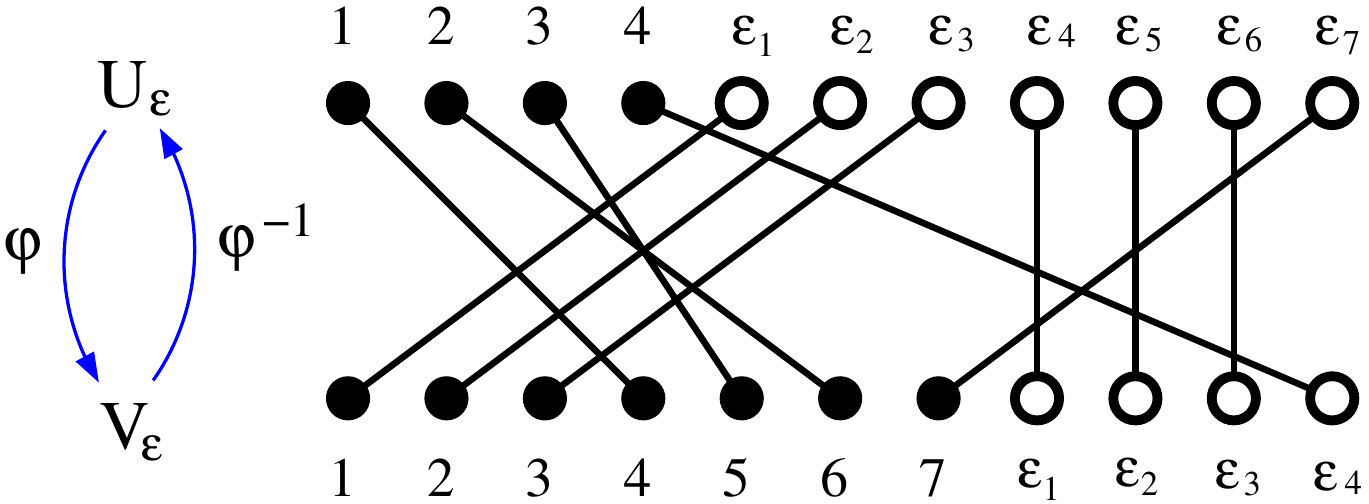}
 &~&\includegraphics[scale=0.45]{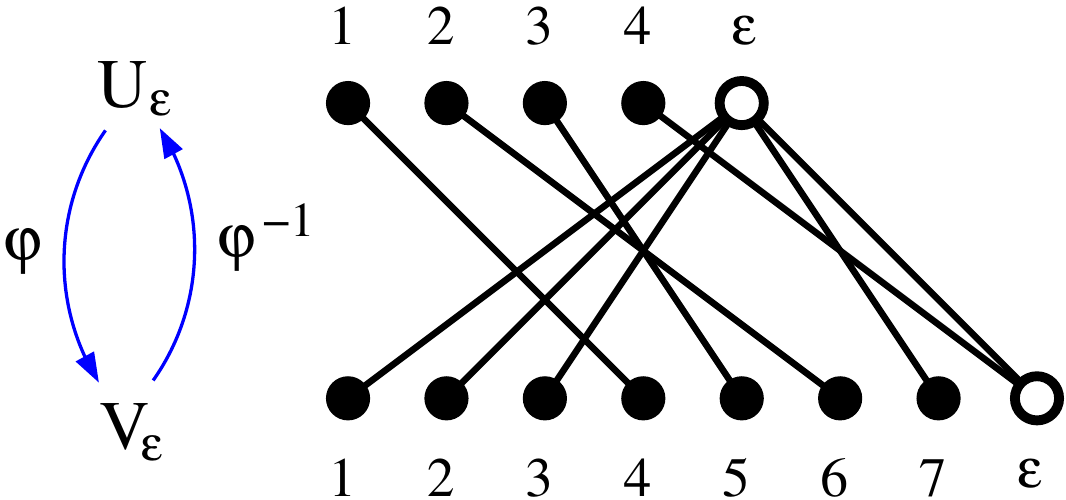}\\
 	(a)& &(b)\\
 	\multicolumn{3}{c}{\includegraphics[scale=0.45]{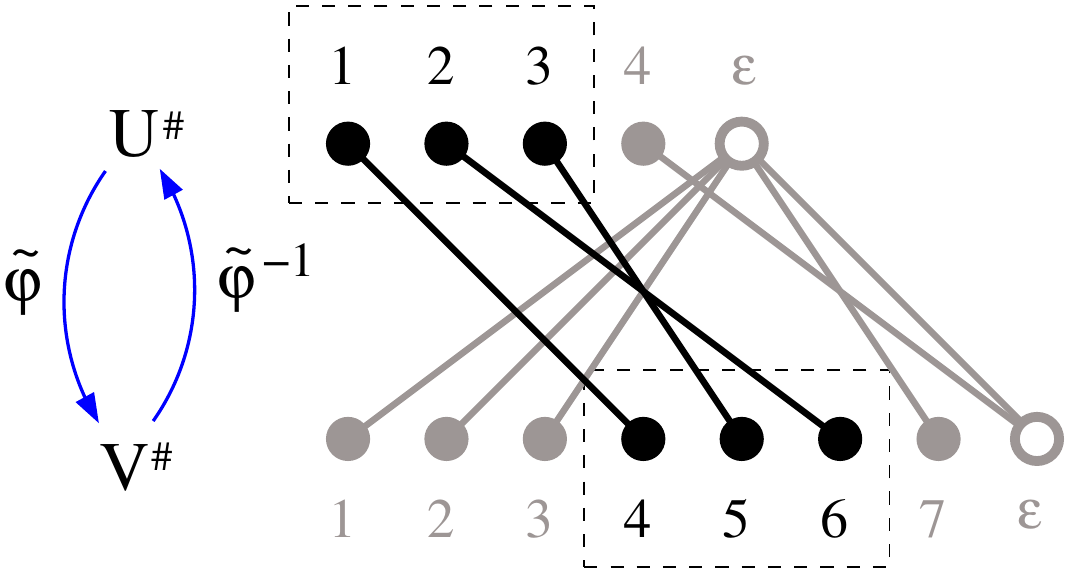}}\\
 	\multicolumn{3}{c}{(c)}
 	\end{tabular}
	\caption{(a) A squared assignment with edition (Problem~\ref{def-squaredlsap}), \textit{i.e.} a bijection $\varphi\,{:}\,\mathcal{U}\,{\cup}\,\mathcal{E}_\mathcal{U}\rightarrow\mathcal{V}\,{\cup}\,\mathcal{E}_\mathcal{V}$. (b) The equivalent assignment with edition (Problem~\ref{pb-assigneps}), \textit{i.e.} a mapping $\varphi\,{:}\,\mathcal{U}_\epsilon\,{\rightarrow}\,\mathcal{P}(\mathcal{V}_\epsilon)$ satisfying constraints given by Eq.~\ref{eq-cst-vphi}. (c) The bijective restriction $\widetilde{\varphi}$ of the assignment with edition $\varphi$ defined in (b). The sets $\mathcal{U}^\sharp$ and $\mathcal{V}^\sharp$ in (b) are surrounded by dashed boxes.}\label{fig-assignedit}
\end{figure}
In this paper, we study an extension of this problem which allows to
model the transformation of the first set $\mathcal{U}$ into the second one $\mathcal{V}$ by means of edit operations. Three simple operations are considered. In addition to the traditional mapping of elements, \textit{i.e.} each element of $\mathcal{U}$ can be substituted by an element of $\mathcal{V}$, the elements of $\mathcal{U}$ can be removed and the elements of $\mathcal{V}$ can be inserted into $\mathcal{U}$. This extension
is motivated by the definition of error-tolerant distance or similarity measures, which 
play an important role in pattern recognition. In particular, the graph edit distance \cite{tsai74,bun83,san83,Bunke1999} is one of the most flexible graph dissimilarity measure, but computing it exactly is a NP-complete problem. The main strategy developed in the last decade consists in approximating the graph edit distance by solving a special weighted bipartite graph matching of the nodes such that insertion and removal operations are represented \cite{Justice2006,riesen2007bipartite,riesen09,Zeng2009,rie10,fan11,gau14,rie15,ser15}. An edit path containing node and edge operations is then deduced from an initial node mapping provided by the LSAP. The approximate edit distance is then defined as the cost of this edit path. Initially designed as an upper bound of the graph edit distance by defining the cost of mapping two nodes as the cost of mapping their labels \cite{Justice2006}, it has been improved by also penalizing the mapping of their direct neighborhoods (incident edges and adjacent nodes) within this cost \cite{riesen2007bipartite,riesen09}. More global and representative patterns than direct node neighborhoods, such as bags of walks \cite{gau14}, have also been explored. 

The main difficulty within this bipartite framework consists in
defining properly the removal and insertion operations. Moreover, the
cost of mapping two direct neighborhoods or two bags is usually
defined as the cost of an optimal weighted bipartite graph matching of
the edges or patterns within the bags, which also consider removal and
insertion operations. Such weighted bipartite graph matchings may be
formulated as the following error-tolerant assignment problem between
two sets \cite{riesen2007bipartite,riesen09}.

\begin{problem}[sLSAPE]\label{def-squaredlsap}
Let $\mathcal{U}$ and $\mathcal{V}$ be two finite sets
  ($|\mathcal{U}|\,{=}\,n$ and $|\mathcal{V}|\,{=}\,m$). Any element of $\mathcal{U}$ can be substituted to any element of $\mathcal{V}$. The removal
  of the elements of $\mathcal{U}$ is represented by augmenting
  $\mathcal{V}$ by $n$ distinct null elements
  $\mathcal{E}_\mathcal{V}\,{=}\,\{\epsilon_i\}_{i=1,\ldots,n}$ such
  that each element $\epsilon_i\,{\in}\,\mathcal{E}_\mathcal{V}$
  corresponds to a unique element $u_i\,{\in}\,\mathcal{U}$ and
  reciprocally. The insertion of the elements of $\mathcal{V}$ in $\mathcal{U}$ is represented similarly by augmenting $\mathcal{U}$ by $m$
  elements
  $\mathcal{E}_\mathcal{U}\,{=}\,\{\epsilon_j\}_{j=1,\ldots,m}$. The
  two augmented sets have the same cardinality $n\,{+}\,m$ (Fig.\,\ref{fig-assignedit}(a)). 
  Each possible edit operation (substitution, removal and insertion) is penalized by a cost function $c\,{:}\,(\mathcal{U}\,{\cup}\,\mathcal{E}_\mathcal{U})\times(\mathcal{V}\,{\cup}\,\mathcal{E}_\mathcal{V})\,{\rightarrow}\,[0,\omega]$ such that:
	\begin{enumerate}
	\item[$\bullet$]the cost $c(u_i,v_j)$ penalizes the substitution of any $u_i\,{\in}\,\mathcal{U}$ by any $v_j\,{\in}\,\mathcal{V}$,
	\item[$\bullet$]the cost $c(u_i,\epsilon_i)$ penalizes the removal of any $u_i$,
	\item[$\bullet$]the cost $c(\epsilon_j,v_j)$ penalizes the insertion of any $v_j$ into $\mathcal{U}$.
	\end{enumerate}
	Assigning any $u_i$ to any $\epsilon_k$ ($k\,{\not=}\,i$) or any $\epsilon_l$ to any $v_j$ ($l\,{\not=}\,j$) is forbidden, so the associated costs are set to a large value $\omega\,{>}\,\max\{c(u_i,v_j),\,\forall i\,{\in}\,\mathcal{U},\,j\,{\in}\,\mathcal{V}\}$, \textit{i.e.} $c(u_i,\epsilon_k)\,{=}\,c(\epsilon_l,v_j)\,{=}\,\omega$. Finally, we have $c(\epsilon_j,\epsilon_i)\,{=}\,0$ for all $u_i\,{\in}\,\mathcal{U}$ and all $v_j\,{\in}\,\mathcal{V}$, since assigning two null elements should not influence the overall assignment.
	
	The squared linear sum assignment problem with edition (sLSAPE) consists in finding a bijection
  $\varphi\,{:}\,\mathcal{U}\,{\cup}\,\mathcal{E}_\mathcal{U}\,{\rightarrow}\,\mathcal{V}\,{\cup}\,\mathcal{E}_\mathcal{V}$
  that minimizes the objective functional:
	\begin{equation}\label{eq-Aepssq}
		\begin{aligned}	
                  A(\varphi,c)&=\hspace{-0.12cm}\sum_{u\in\mathcal{U}\cup\mathcal{E}_\mathcal{U}}\hspace{-0.18cm}c(u,\varphi(u))=\sum_{u_i\in\mathcal{U}}c(u_i,\varphi(u_i))+\sum_{\epsilon_j\in\mathcal{E}_\mathcal{U}}c(\epsilon_j,\varphi(\epsilon_j))\\
                  &=\underset{\text{substitutions}}{\underbrace{\sum_{\substack{u_i\in\mathcal{U}\\v_j=\varphi(u_i)}}\hspace{-0.18cm}c(u_i,v_j)}}~+\underset{\text{removals}}{\underbrace{\sum_{\substack{u_i\in\mathcal{U}\\\epsilon_i=\varphi(u_i)}}\hspace{-0.18cm}c(u_i,\epsilon_i)}}~+\underset{\text{insertions}}{\underbrace{\sum_{\substack{\epsilon_j\in\mathcal{E}_\mathcal{U}\\v_j=\varphi(\epsilon_j)}}\hspace{-0.18cm}c(\epsilon_j,v_j)}}
		\end{aligned}
	\end{equation}
The set of bijective mappings between $\mathcal{U}\,{\cup}\,\mathcal{E}_\mathcal{U}$ and $\mathcal{V}\,{\cup}\,\mathcal{E}_\mathcal{V}$ is denoted $\mathcal{S}_{\mathcal{E}}(\mathcal{U},\mathcal{V})$.
\end{problem}
\noindent
Since the sLSAPE is a LSAP (Problem~\ref{pb:lsap}) with a specific cost function, it
can be solved by any algorithm solving the LSAP, for instance in
$O((n\,{+}\,m)^3)$ time complexity by the Hungarian algorithm \cite{riesen09}. This algorithm is based on the formulation of the LSAP given by Eq.~\ref{eq-lsapmtx}. For the sLSAPE, it consists in finding an optimal $(n\,{+}\,m)\,{\times}\,(n\,{+}\,m)$ permutation matrix, provided a $(n\,{+}\,m)\,{\times}\,(n\,{+}\,m)$ edit cost matrix of the following form
\begin{equation}\label{eq-sC}\begin{small}
	\left(\begin{array}{cccc|cccc}
		& & && c(1,\epsilon_1) & \omega &  \cdots    & \omega \\
		\multicolumn{4}{c|}{(c(u_i,v_j))_{i,j}}& \omega &  c(2,\epsilon_2) & \ddots     & \vdots \\
		& & & &  \vdots         & \ddots &\ddots &\omega\\
		& &		&							&  \omega  &  \cdots &\omega&  c(n,\epsilon_n)\\\hline
		c(\epsilon_1,1) & \omega  & \cdots & \omega    & & && \\
			\omega	&   c(\epsilon_2,2) & \ddots & \vdots  & \multicolumn{4}{c}{(c(\epsilon_j,\epsilon_i))_{i,j}\,{=}\,\mathbf{0}_{m,n}}\\
			\vdots	& \ddots & \ddots & \omega  & &  &&\\
		\omega& \cdots & \omega    & c(\epsilon_m,m) & & &\\
	\end{array}\right)\in[0,\omega]^{(n+m)\times(m+n)}.
\end{small}\end{equation}
As we
can observe, many coefficients of such a cost matrix are not
relevant (equal to $\omega$ or to $0$) in the blocks representing
the cost of removal/insertion operations. These values have been
added in order to model the problem of transforming one set into another, by means of edit operations, as a LSAP. To improve the time and memory complexity, the computation of an optimal transformation should not depend on these useless values. To this, a variant of the sLSAPE
has been proposed in \cite{ser15}, but it restricts insertion (or removal)
operations to occur in only one set.

In this report, rather than modelling the initial transformation problem such that it can be
solved with known algorithms, we formalize it as an extension of the LSAP which considers only one additional element in each set to represent removal and insertion operations (Fig.~\ref{fig-assignedit}(b)). A solution to the problem is not any more represented as a bijection between the two augmented sets (Section~\ref{sec-assed}). We call this extension the linear sum assignment problem with edit
operations (LSAPE). Intuitively, it consists in reducing the square edit cost matrix (Eq.~\ref{eq-sC}) by removing its useless elements, in order to obtain a $(n\,{+}\,1)\,{\times}\,(m\,{+}\,1)$ edit cost matrix
\begin{equation}\begin{small}
	\left(\begin{array}{cccc|c}
		& & && c(1,\epsilon)  \\
		\multicolumn{4}{c|}{(c(u_i,v_j))_{i,j}}& c(2,\epsilon) \\
		& & & &  \vdots    \\
		& &		&&	 c(n,\epsilon)\\\hline
		c(\epsilon,1) & c(\epsilon,2)  & \cdots & c(\epsilon,m)    & 0 \\
	\end{array}\right)\in[0,{+\infty})^{(n+1)\times(m+1)}.
\end{small}\end{equation}
We show that the LSAPE is
a specific binary linear program where the set of
constraints defines a mixed bipartite graph (Section~\ref{sec-pbform}). According to this
analysis, we propose a modified Hungarian algorithm which solves the
problem in $O(\min\{n,m\}^2\max\{n,m\})$ time complexity and in $O(nm)$ memory complexity (Section~\ref{sec-algo}). This is faster than the similar
Hungarian algorithm used to solve the sLSAPE, as observed
experimentally with comparable implementations.

\section{Assignments with Edition}\label{sec-assed}
In this section we define the notion of assignment with edition (or $\epsilon$-assignment),
and we analyze some associated algebraic and numerical properties.
\subsection{Definition and combinatorial properties}
Consider two finite sets $\mathcal{U}$ and $\mathcal{V}$. Let $n\,{=}\,|\mathcal{U}|$ and $m\,{=}\,|\mathcal{V}|$ be their respective cardinalities. Transforming $\mathcal{U}$ into $\mathcal{V}$, by means of edit operations, can be realized:
\begin{enumerate}
\item[$\bullet$]by substituting each element of $\mathcal{U}$ by a unique element of $\mathcal{V}$, or by removing it from $\mathcal{U}$, and then
\item[$\bullet$]by inserting into $\mathcal{U}$ each element of $\mathcal{V}$ which has not been used for a substitution.
\end{enumerate}
Let $\epsilon$ be an element introduced to represent removal and insertion operations. The removal of an element $u\,{\in}\,\mathcal{U}$ is denoted by $u\,{\rightarrow}\,\epsilon$, and the insertion of an element $v\,{\in}\,\mathcal{V}$ into $\mathcal{U}$ is denoted by $\epsilon\,{\rightarrow}\,v$.  Consider the two augmented sets $\mathcal{U}_\epsilon\,{=}\,\mathcal{U}\,{\cup}\,\{\epsilon\}$ and $\mathcal{V}_\epsilon\,{=}\,\mathcal{V}\,{\cup}\,\{\epsilon\}$. We propose to define a transformation from $\mathcal{U}$ into $\mathcal{V}$ as follows (Fig.\ref{fig-assignedit}(b)).
\begin{definition}[$\epsilon$-assignment]\label{def-assigneps}
An assignment with edition (or $\epsilon$-assignment) from $\mathcal{U}$ to $\mathcal{V}$ is a mapping $\varphi\,{:}\,\mathcal{U}_\epsilon\,{\rightarrow}\,\mathcal{P}(\mathcal{V}_\epsilon)$ satisfying the following constraints:
\begin{equation}\label{eq-cst-vphi}
	\left\{\begin{aligned}
		\forall u\,{\in}\,\mathcal{U},&~~\,\left|\varphi(u)\right| = 1\\
		\forall v\,{\in}\,\mathcal{V},&~~\,\left|\varphi^{-1}[v]\right| = 1\\
		\epsilon\in\varphi(\epsilon)&~
	\end{aligned}\right.
\end{equation}
where $\mathcal{P}(\cdot)$ is the powerset, and $\varphi^{-1}[v]\,{\subset}\,\mathcal{U}_\epsilon$ denotes the preimage of any singleton $\{v\}\,{\in}\,\mathcal{P}(\mathcal{V}_\epsilon)$ by $\varphi$. Note that braces are omitted for singletons ($\varphi^{-1}[\{v\}]$ is written $\varphi^{-1}[v]$). Let $\mathcal{A}_\epsilon(\mathcal{U},\mathcal{V})$ be the set of all $\epsilon$-assignments from  $\mathcal{U}$ to $\mathcal{V}$. 
\end{definition}
\noindent
Note that the edit operation $\epsilon\,{\rightarrow}\,\epsilon$ is always selected, but we could equivalently consider this operation as being never selected. Compared to the definition of bijections involved in the sLSPAE (Problem \ref{def-squaredlsap}), an $\epsilon$-assignment does not constrain the $\epsilon$-elements directly. Indeed, by definition of $\varphi$ above, we have
\begin{equation}\label{eq-unconst}
1\leq|\varphi(\epsilon)|\leq m\,{+}\,1~~~\text{and}~~~
1\leq|\varphi^{-1}[\epsilon]|\,{\leq}\,n\,{+}\,1.
\end{equation}
For example the $\epsilon$-assignment in Fig.~\ref{fig-assignedit}(b) is given by
$$
	1\,{\rightarrow}\,\{4\},~2\,{\rightarrow}\,\{6\},~3\,{\rightarrow}\,\{5\},~4\,{\rightarrow}\,\{\epsilon\},~\epsilon\,{\rightarrow}\,\{1,2,3,7,\epsilon\}.
$$
Definition~\ref{def-assigneps} also implies the existence of the sets:
\begin{equation}
\begin{aligned}
	&\mathcal{U}^\sharp=\left\{u\,{\in}\,\mathcal{U}\,|\,\varphi(u)\,{\subset}\,\mathcal{V} \right\}=\varphi^{-1}[\mathcal{V}]\cap\mathcal{U}\\
	&\mathcal{V}^\sharp=\left\{v\,{\in}\,\mathcal{V}\,|\,\varphi^{-1}[v]\,{\subset}\,\mathcal{U} \right\}=\varphi[\mathcal{U}]\cap\mathcal{V}
\end{aligned}
\end{equation}
where $\varphi^{-1}[\mathcal{V}]$ denotes the preimage of the set composed of all singletons of $\mathcal{P}(\mathcal{V})$ by $\varphi$. Similarly $\varphi[\mathcal{U}]$ denotes the image of $\mathcal{U}$ by $\varphi$. Then the function (Fig. \ref{fig-assignedit}(c))
\begin{equation}
	\widetilde{\varphi}\,{:}\,
        \left(\begin{array}{ccl}
                \mathcal{U}^\sharp&\rightarrow&\mathcal{V}^\sharp\\
                u&\mapsto& v~~\mbox{with}~\{v\}\,{=}\,\varphi(u)\\
              \end{array}\right)
\end{equation}
is bijective. Indeed, given $u\,{\in}\,\mathcal{U}^\sharp$, by
Eq.~\ref{eq-cst-vphi} we have $|\varphi(u)|\,{=}\,1$. Since $u\,{\in}\,\mathcal{U}^\sharp$,
there is $v\,{\in}\,\mathcal{V}$ such that $\varphi(u)\,{=}\,\{v\}$. Moreover, by Eq.~\ref{eq-cst-vphi}
we still have $|\varphi^{-1}[v]|\,{=}\,1$, and since
$\varphi(u)\,{=}\,\{v\}$, then $\varphi^{-1}[v]\,{=}\,\{u\}\,{\subset}\,\mathcal{U}$. Hence
$v\,{\in}\,\mathcal{V}^\sharp$. Moreover, by Eq.~\ref{eq-cst-vphi}, for any
$v\,{\in}\,\mathcal{V}^\sharp$ there is a single
$u\,{\in}\,\mathcal{U}$ such that $\{v\}\,{=}\,\varphi(u)$. We have
$u\,{\in}\,\mathcal{U}$ and $\varphi(u)\,{=}\,\{v\}\,{\subset}\,\mathcal{V}$, thus
$u\,{\in}\,\mathcal{U}^\sharp$.  Hence, for any
$v\,{\in}\,\mathcal{V}^\sharp$ there is a single
$u\,{\in}\,\mathcal{U}^\sharp$ such that
$\widetilde{\varphi}(u)\,{=}\,v$.

The function $\widetilde{\varphi}$
is thus bijective. It may be understood as the restriction of
$\varphi$ to the set of elements of $\mathcal{U}$ which are not mapped
onto $\epsilon$. An $\epsilon$-assignment may thus be understood as
a bijection on which the bijectivity constraint is relaxed for
$\epsilon$.

Let $\mathcal{I}(\mathcal{U},\mathcal{V})$ be the set of all injections
from a subset of $\mathcal{U}$ onto $\mathcal{V}$. Let
$\varphi\,{\in}\,\mathcal{A}_\epsilon(\mathcal{U},\mathcal{V})$ be an
$\epsilon$-assignment. Note that the bijective restriction
$\widetilde{\varphi}$ of $\varphi$ is by definition injective from
$\mathcal{U}^\sharp$ onto $\mathcal{V}$, and
hence belongs to $\mathcal{I}(\mathcal{U},\mathcal{V})$. Consider a mapping
  $\widetilde{\varphi}\,{:}\,\mathcal{U}_s\,{\subseteq}\,\mathcal{U}\,{\rightarrow}\,\mathcal{V}$
  of $\mathcal{I}(\mathcal{U},\mathcal{V})$. Then the 
  mapping $\varphi\,{:}\,\mathcal{U}_\epsilon\,{\rightarrow}\,\mathcal{P}(\mathcal{V}_\epsilon)$ defined by:
\begin{equation}\label{eq-phisharp}
	\left\lbrace\begin{aligned}
		\forall u\,{\in}\,\mathcal{U}_s,&~~\varphi(u)=\{\widetilde{\varphi}(u)\}\\
		\forall u\,{\in}\,\mathcal{U}\,{\setminus}\,\mathcal{U}_s,&~~\varphi(u)=\{\epsilon\}\\
		~&~~\varphi(\epsilon)=\mathcal{V}_\epsilon\setminus\widetilde{\varphi}[\mathcal{U}_s]
	\end{aligned}\right.
\end{equation}
belongs to $\mathcal{A}_\epsilon(\mathcal{U},\mathcal{V})$
by construction. More generally $\varphi$ and $\widetilde{\varphi}$ are linked as follows.
\begin{proposition}\label{prop:bijAESA}
The following mapping is bijective:
\begin{equation}
	\psi:\left(\begin{array}{ccc}
		\mathcal{A}_\epsilon(\mathcal{U},\mathcal{V})&\rightarrow &\mathcal{I}(\mathcal{U},\mathcal{V})\\
		\varphi & \mapsto & \widetilde{\varphi}
	\end{array}\right)
\end{equation}
\end{proposition}
\begin{proof}
  Consider a mapping
  $\widetilde{\varphi}\,{:}\,\mathcal{U}_s\,{\subseteq}\,\mathcal{U}\,{\rightarrow}\,\mathcal{V}$
  of $\mathcal{I}(\mathcal{U},\mathcal{V})$. The associated
  mapping $\varphi\,{\in}\,\mathcal{A}_\epsilon(\mathcal{U},\mathcal{V})$ can
  be reconstructed  according to Eq.~\ref{eq-phisharp}, and one can easily show that $\varphi$
belongs to $\mathcal{A}_\epsilon(\mathcal{U},\mathcal{V})$
by construction. The mapping $\psi$
is thus surjective. We show that it is bijective. Consider two elements
$\varphi_1,\varphi_2\in\mathcal{A}_\epsilon(\mathcal{U},\mathcal{V})$
such that $\varphi_1\,{\not=}\,\varphi_2$.  There is necessarily $u\,{\in}\,\mathcal{U}_\epsilon$ such that
$\varphi_1(u)\,{\not=}\,\varphi_2(u)$:
\begin{itemize}
\item If $u\,{\in}\,\mathcal{U}$, then: 
  \begin{itemize}
  \item If $\varphi_1(u)\,{=}\,\{v\}$ and $\varphi_2(u)\,{=}\,\{v'\}$ with
    $\epsilon\,{\not\in}\,\{v,v'\}$ then
    $u\,{\in}\, \mathcal{U}_1^\sharp\cap \mathcal{U}_2^\sharp$,
    but in this case we have $\widetilde{\varphi}_1(u)\,{=}\,v\,{\not=}\,\widetilde{\varphi}_2(u)\,{=}\,v'$, which implies that
    $\widetilde{\varphi}_1\,{\not=}\,\widetilde{\varphi}_2$.  
  \item Else, suppose w.l.o.g. that
    $\varphi_1(u)\,{=}\,\{\epsilon\}$ and
    $\varphi_2(u)\,{\not=}\,\{\epsilon\}$, then
    $u\,{\not\in}\,\mathcal{U}_1^\sharp$ and
    $u\,{\in}\,\mathcal{U}_2^\sharp$ which implies that
    $\widetilde{\varphi}_1\,{\not=}\,\widetilde{\varphi}_2$ (the $\widetilde{\varphi}_i$ are not defined on the same subset of $\mathcal{U}$).
  \end{itemize}
\item If $u\,{=}\,\epsilon$, consider w.l.o.g. that
  $v\,{\in}\,\varphi_1(\epsilon)\,{\setminus}\,\varphi_2(\epsilon)$.
  By Eq.~\ref{eq-unconst}, we have $\epsilon\,{\in}\,\varphi_1(\epsilon)\,{\cap}\,\varphi_2(\epsilon)$, consequently 
  $\epsilon\not=v\,{\in}\,\mathcal{V}$. Since $v\,{\not\in}\,\varphi_2(\epsilon)$
  it exists $u^\prime\,{\in}\,\mathcal{U}$ such that
  $\varphi_2(u^\prime)\,{=}\,\{v\}$.  Then we have
  $u^\prime\,{\in}\,\mathcal{U}_2^\sharp$.
  \begin{itemize}
  \item If
    $u^\prime\,{\not\in}\,\mathcal{U}_1^\sharp$,
    $\widetilde{\varphi}_1$ and $\widetilde{\varphi}_2$ are not
    defined on the same subset of $\mathcal{U}$ and are consequently
    not equal.
  \item If
    $u^\prime\,{\in}\,\mathcal{U}_1^\sharp$,
    since $v\,{\in}\,\varphi_1(\epsilon)$ then
    $\varphi_1(u^\prime)\,{\not=}\,\{v\}\,{=}\,\varphi_2(u^\prime)$ from Eq.~\ref{eq-cst-vphi},
    and so $\widetilde{\varphi}_1\,{\not=}\,\widetilde{\varphi}_2$.
  \end{itemize}
\end{itemize}
In all cases we get $\widetilde{\varphi}_1\,{\not=}\,\widetilde{\varphi}_2$. The mapping $\psi$ is thus injective and hence bijective.
\end{proof}
\noindent
A simple consequence is given by the following equality.
\begin{corollary}\label{coro-eqAI}
$|\mathcal{A}_\epsilon(\mathcal{U},\mathcal{V})|\,{=}\,|\mathcal{I}(\mathcal{U},\mathcal{V})|$.
\end{corollary}
\noindent
This allows to easily enumerate the number of $\epsilon$-assignments.
\begin{theorem}\label{th-nbass}
	There is 
	\begin{equation*}
	|\mathcal{A}_\epsilon(\mathcal{U},\mathcal{V})|=\hspace{-0.16cm}\sum_{p=0}^{\min\{n,m\}}\hspace{-0.16cm}\comb{n}{p}\comb{m}{p}p!
	\end{equation*}
	possible $\epsilon$-assignments between two sets
        $\mathcal{U}$ and $\mathcal{V}$ with $n\,{=}\,|\mathcal{U}|$,
        $m\,{=}\,|\mathcal{V}|$. $\comb{n}{p}$ is the number of permutations (bijections), and $\arrang{m}{p}$ the number of partial permutations (injections).
\end{theorem}
\begin{proof}
  Consider the set $\mathcal{I}(\mathcal{U},\mathcal{V})$ and the
  notations of Eq.~\ref{eq-phisharp}. We have $\comb{n}{p}$ subsets
  $\mathcal{U}_s$ of $\mathcal{U}$ with
  $|\mathcal{U}_s|\,{=}\,p\,{\leq}\,\min\{n,m\}$. For each subset there exists
  $p!\,\comb{m}{p}$ injective mappings from $\mathcal{U}_s$ to
  $\mathcal{V}$. So for each $p\,{=}\,0,\ldots,\min\{n,m\}$, the
  number of injective mappings is $p!\,\comb{n}{p}\comb{m}{p}$. Note
  that $p\,{=}\,0$ corresponds to the case where all elements of
  $\mathcal{U}$ and $\mathcal{V}$ are assigned to $\epsilon$.
\end{proof}
\noindent
Remark that there exists much more $\epsilon$-assignments than injections. Indeed, suppose that $n\leq m$, we have
\begin{equation*}
	|\mathcal{A}_\epsilon(\mathcal{U},\mathcal{V})|=\arrang{m}{n}+\sum_{p=0}^{n-1}\arrang{m}{p}\comb{n}{p}>\arrang{m}{n}
\end{equation*}
However, there is much less $\epsilon$-assignments than bijections enumerated over $n\,{+}\,m$ elements, as considered by the sLSAPE.
\begin{corollary}
  If w.l.o.g we suppose that $n\leq m$ we have:
  \[
  \comb{n+m}{n}\leq |\mathcal{A}_\epsilon(\mathcal{U},\mathcal{V})|\leq \frac{(n+m)!}{m!}<(n+m)!
  \]
\end{corollary}
\begin{proof}
  We have by a classical result on binomial coefficient:
  \[
  \sum_{p=0}^n\comb{n}{p}\comb{m}{p}=\comb{n+m}{n}=\frac{(n+m)!}{n!m!}
  \]
  Hence from Theorem~\ref{th-nbass} we have:
  \[
  \begin{aligned}
  &|\mathcal{A}_\epsilon(\mathcal{U},\mathcal{V})|\leq n!\sum_{p=0}^n\comb{n}{p}\comb{m}{p}
  =\frac{(n+m)!}{m!}<(n+m)!\\
  &|\mathcal{A}_\epsilon(\mathcal{U},\mathcal{V})|\geq \sum_{p=0}^n\comb{n}{p}\comb{m}{p}=\comb{n+m}{n}
  \end{aligned}
  \]
\end{proof}
\noindent
As for bijections, $\epsilon$-assignments can be equivalently represented by matrices, which are useful for analyzing the associated linear sum assignment problem and for deriving a solution.
\subsection{Matrix form}\label{sec-mtxform}
Given two sets $\mathcal{U}\,{=}\,\{u_i\}_{i=1,\ldots,n}$ and $\mathcal{V}\,{=}\,\{v_j\}_{j=1,\ldots,m}$, any $\epsilon$-assignment $\varphi\,{\in}\,\mathcal{A}_\epsilon(\mathcal{U},\mathcal{V})$ is in one-to-one correspondence with a matrix $\mathbf{X}\,{\in}\,\{0,1\}^{(n+1)\times(m+1)}$ such that its $n$ first rows represent $\mathcal{U}$ (row $i$ corresponds to $u_i$), its $m$ first columns represent $\mathcal{V}$ (column $j$ corresponds to $v_j$), its row $(n\,{+}\,1)$ and its column $(m\,{+}\,1)$ represent the null element $\epsilon$, and
\begin{equation*}
	\left\lbrace\begin{aligned}
		&x_{i,j}=\delta_{\varphi(u_i)=\{v_j\}},&&\forall i\,{=}\,1,\ldots,n,~\forall j\,{=}\,1,\ldots,m&&\text{(substitutions)}\\
		&x_{i,m+1}=\delta_{\varphi(u_i)=\{\epsilon\}},&&\forall i\,{=}\,1,\ldots,n&&\text{(removals)}\\
		&x_{n+1,j}=\delta_{\varphi^{-1}[v_j]=\{\epsilon\}},&&\forall j\,{=}\,1,\ldots,m&&\text{(insertions)}\\
		&x_{n+1,m+1}=1&&~&&~
	\end{aligned}\right.
\end{equation*}
where $\delta_{r}\,{=}\,1$ if the relation $r$ is true, or $\delta_{r}\,{=}\,0$ else.

The matrix $\mathbf{X}$ has the general form:
\begin{equation}\label{eq-mtx-assigneps}
	\begin{blockarray}{cc|c}
	&v_1\cdots v_m&\epsilon \\
	\begin{block}{c(c|c)}
		u_1&&\\
		\vdots&\mathbf{X}^\text{sub}& \mathbf{x}^\text{rem} \\
		u_n&  &\\\cline{1-3}
		\epsilon&\mathbf{x}^\text{ins}&1\\
	\end{block}
	\end{blockarray}\in\{0,1\}^{(n+1)\times(m+1)}
\end{equation}
where $\mathbf{X}^{\text{sub}}\,{\in}\,\{0,1\}^{n\times m}$ represents substitutions, $\mathbf{x}^{\text{rem}}\,{\in}\,\{0,1\}^{n\times 1}$ removals, and $\mathbf{x}^{\text{ins}}\,{\in}\,\{0,1\}^{1\times m}$ insertions. For example, the matrix which represents the $\epsilon$-assignment shown in Fig.~\ref{fig-assignedit}(b) is given by:
\begin{equation*}\begin{small}
	\begin{blockarray}{cccccccc|c}
	&v_1&v_2&v_3&v_4&v_5&v_6&v_7&\epsilon &&\\
	\begin{block}{c(ccccccc|c)}
		u_1&0 & 0 & 0 & 1 & 0 & 0 & 0 & 0\\
		u_2&0 & 0 & 0 & 0 & 0 & 1 & 0 & 0\\
		u_3&0 & 0 & 0 & 0 & 1 & 0 & 0 & 0\\
		u_4&0 & 0 & 0 & 0 & 0 & 0 & 0 & 1\\\cline{1-9}
		\epsilon &1 & 1 & 1 & 0 & 0 & 0 & 1 & 1\\
	\end{block}
	\end{blockarray}
\end{small}\end{equation*}
Due to the constraints on $\varphi$ (Definition \ref{def-assigneps}), such a matrix has a $1$ on each of its $n$ first rows, and a $1$ on each of its $m$ first columns:
\begin{equation}\label{eq-assignedX}
	\left\lbrace\begin{aligned}
		&\sum_{j=1}^{m+1}x_{i,j}=1,&\forall i\,{=}\,1,\ldots,n&~~~~~~(\textit{i.e.}~|\varphi(u_i)|=1)\\
		&\sum_{i=1}^{n+1}x_{i,j}=1,&\forall j\,{=}\,1,\ldots,m&~~~~~~(\textit{i.e.}~|\varphi^{-1}[v_j]|=1)\\
		&x_{n+1,m+1}=1&&&
	\end{aligned}\right.
\end{equation}

Reciprocally, any matrix $\mathbf{X}\,{\in}\,\{0,1\}^{(n+1)\times(m+1)}$ satisfying Eq.~\ref{eq-assignedX} represents an $\epsilon$-assignment $\varphi$ such that:
\begin{equation}
	\left\{\begin{aligned}
	\forall u_i\,{\in}\,\mathcal{U},&&\varphi(u_i)&=\left\lbrace\begin{aligned}
		&\{v_j\}&&\text{if}~\,\exists j\,{\in}\,\{1,\ldots,m\}~|~x_{i,j}\,{=}\,1\\
		&\{\epsilon\}&&\text{else}~(\textit{i.e.}~x_{i,m+1}\,{=}\,1)
	\end{aligned}\right.\\
	~&&\varphi(\epsilon)&=\left\{v_j\,{\in}\,\mathcal{V}_\epsilon~|~x_{n+1,j}\,{=}\,1\right\}
	\end{aligned}\right.
\end{equation}
Let $\mathcal{S}_{n,m,\epsilon}$ be the set of all matrices in $\{0,1\}^{(n+1)\times(m+1)}$ satisfying Eq.~\ref{eq-assignedX}. It is the matrix form of $\mathcal{A}_\epsilon(\mathcal{U},\mathcal{V})$. Compared to a permutation matrix, where each row and each column sum to $1$, the last row and the last column of a matrix satisfying Eq.~\ref{eq-assignedX} satisfies the following relaxed constraints (Eq.~\ref{eq-unconst}):
\begin{equation*}
	1\leq\sum_{i=1}^{n+1}x_{i,m+1}\leq n+1~~~\text{and}~~~1\leq\sum_{j=1}^{m+1}x_{n+1,j}\leq m+1.
\end{equation*}

Remark that any $\epsilon$-assignment $\varphi\,{\in}\,\mathcal{A}_\epsilon(\mathcal{U},\mathcal{V})$ can be equivalently represented by a bipartite graph $G\,{=}\,((\mathcal{U}_\epsilon,\mathcal{V}_\epsilon),\mathcal{E})$ such that each edit operation performed according to $\varphi$ corresponds to an edge of $\mathcal{E}$. This graph is a subgraph of the complete bipartite graph $K_{n+1,m+1}$ generated from the two sets $\mathcal{U}_\epsilon$ and $\mathcal{V}_\epsilon$. Note that 
the matrix $\mathbf{X}$ associated with $\varphi$ corresponds to the node adjacency matrix of $G$. Reciprocally, any subgraph of $K_{n+1,m+1}$ represented by a node adjacency matrix $\mathbf{X}$ satisfying Eq.~\ref{eq-assignedX} defines an $\epsilon$-assignment.
\subsection{Partial assignments with edition}\label{sec-partial}
A partial assignment from $\mathcal{U}$ onto $\mathcal{V}$ is an assignment wherein all elements are not assigned, \textit{i.e.} a bijection from a subset of $\mathcal{U}$ to a subset of $\mathcal{V}$, or equivalently a partial permutation matrix. This can be defined similarly in the context of $\epsilon$-assignments.
\begin{definition}[partial $\epsilon$-assignment]\label{def-partialeps}
A partial $\epsilon$-assignment from $\mathcal{U}$ to $\mathcal{V}$ is a mapping $\varphi:\mathcal{U}_\epsilon\rightarrow\mathcal{P}(\mathcal{V}_\epsilon)$ satisfying the following set of constraints:
\begin{equation}\label{eq-cst-vphi-partial}
	\left\{\begin{aligned}
		\forall u\,{\in}\,\mathcal{U},&~~\,\left|\varphi(u)\right| \in \{0,1\}\\
		\forall v\,{\in}\,\mathcal{V},&~~\,\left|\varphi^{-1}[v]\right| \in \{0,1\}\\
		\epsilon\in\varphi(\epsilon)&~
	\end{aligned}\right.
\end{equation}
which relaxes the one defining $\epsilon$-assignments (Eq.~\ref{eq-cst-vphi})
\end{definition}
\noindent
A partial $\epsilon$-assignment from $\mathcal{U}$ to $\mathcal{V}$ is
an $\epsilon$-assignment from a subset
$\mathcal{U}^\prime\,{\subseteq}\,\mathcal{U}$ onto a subset
$\mathcal{V}^\prime\,{\subseteq}\,\mathcal{V}$,
\textit{i.e.} an $\epsilon$-assignment in
$\mathcal{A}_\epsilon(\mathcal{U}^\prime,\mathcal{V}^{\,\prime})$.

A partial $\epsilon$-assignment can be equivalently represented by a matrix $\mathbf{X}\,{\in}\,\{0,1\}^{(n+1)\times(m+1)}$, having the structure of matrices of $S_{n,m,\epsilon}$ (Eq.~\ref{eq-mtx-assigneps}), and satisfying the set of constraints
\begin{equation}
\left\lbrace\begin{aligned}
	&\sum_{j=1}^{m+1}x_{i,j}\leq 1,&&\forall i\,{=}\,1,\ldots,n\\
	&\sum_{i=1}^{n+1}x_{i,j}\leq 1,&&\forall j\,{=}\,1,\ldots,m\\
	&x_{n+1,m+1}\,{=}\,1&&
\end{aligned}\right.
\end{equation}
This relaxes the ones defined by Eq.~\ref{eq-assignedX}.

For example, the following matrix represents a partial $\epsilon$-assignment with two unassigned elements: $u_2$ and $v_2$.
\begin{equation*}\begin{small}
	\begin{blockarray}{cccccccc|c}
	&v_1&v_2&v_3&v_4&v_5&v_6&v_7&\epsilon &&\\
	\begin{block}{c(ccccccc|c)}
		u_1&0 & 0 & 0 & 1 & 0 & 0 & 0 & 0\\
		u_2&0 & 0 & 0 & 0 & 0 & 0 & 0 & 0\\
		u_3&0 & 0 & 0 & 0 & 1 & 0 & 0 & 0\\
		u_4&0 & 0 & 0 & 0 & 0 & 0 & 0 & 1\\\cline{1-9}
		\epsilon &1 & 0 & 1 & 0 & 0 & 0 & 1 & 1\\
	\end{block}
	\end{blockarray}
\end{small}\end{equation*}
Note that the matrix $\mathbf{X}$ defines the node adjacency matrix of a bipartite graph having some isolated nodes.

\section{Minimal Linear Sum Assignment with Edition}\label{sec-pbform}
In order to simplify the forthcoming expressions, we assume w.l.o.g. that $\mathcal{U}\,{=}\,\{1,\ldots,n\}$, $\mathcal{V}\,{=}\,\{1,\ldots,m\}$, $\mathcal{U}_\epsilon\,{=}\,\mathcal{U}\,{\cup}\,\{n\,{+}\,1\}$ and $\mathcal{V}_\epsilon\,{=}\,\mathcal{V}\,{\cup}\,\{m\,{+}\,1\}$, \textit{i.e.} the element $\epsilon$ corresponds to $n\,{+}\,1$ in $\mathcal{U}_\epsilon$ and to $m\,{+}\,1$ in $\mathcal{V}_\epsilon$.
\subsection{Edit cost and problem formulation}
The definition of an $\epsilon$-assignment from $\mathcal{U}$ to $\mathcal{V}$ does not rely on the nature of the sets, \textit{i.e.} on the nature of the underlying data. In order to select a relevant $\epsilon$-assignment, among all assignments from $\mathcal{U}$ to $\mathcal{V}$, each possible edit operation $i\,{\rightarrow}\,j$ is penalized by a non-negative cost $c_{i,j}$. All possible costs can be represented by an edit cost matrix $\mathbf{C}\,{\in}\,[0,{+\infty})^{(n+1)\times(m+1)}$ having the same structure as the one of matrices in $S_{n,m,\epsilon}$ (Eq.~\ref{eq-mtx-assigneps}):
\begin{equation}\label{eq:matrixC}
	\mathbf{C} = \begin{blockarray}{@{\hspace{2pt}}c|c@{\hspace{2pt}}cc}
	1\cdots m&\epsilon &&\\
	\begin{block}{(@{\hspace{2pt}}c|c)@{\hspace{2pt}}cc}\vspace{-0.15cm}
		&&&1\\\vspace{-0.15cm}
		\mathbf{C}^\text{sub}& \mathbf{c}^\text{rem} &&\vdots\\
		&  &&n\\\cline{1-4}
		\mathbf{c}^\text{ins}&0&&\epsilon\\
	\end{block}
	\end{blockarray}
\end{equation}
where $c^{\text{sub}}_{i,j}$ penalizes the substitution
$i\,{\rightarrow}\,j$ for all pair
$(i,j)\,{\in}\,\mathcal{U}\,{\times}\mathcal{V}$, $c_{i,m+1}$
penalizes the removal $i\,{\rightarrow}\,\epsilon$ for all
$i\,{\in}\,\mathcal{U}$, $c_{n+1,j}$ penalizes the insertion
$\epsilon\,{\rightarrow}\,j$ for all $j\,{\in}\,\mathcal{V}$, and
$c_{n+1,m+1}\,{=}\,0$ associates a zero cost to the mapping 
$\epsilon\,{\rightarrow}\,\epsilon$. Note that the edit cost matrix
associates a cost to each edge of the complete bipartite graph
$K_{n+1,m+1}$.

Let $\varphi\,{\in}\,\mathcal{A}_\epsilon(\mathcal{U},\mathcal{V})$ be an $\epsilon$-assignment. Its cost is defined as
\begin{equation}\label{eq-Aeps}
\begin{aligned}
	A_\epsilon(\varphi,\mathbf{C})&=\sum_{i\in\mathcal{U}_\epsilon}\,\sum_{j\in\varphi(i)}c_{i,j}\\
	&=\underset{\text{substitutions}}{\underbrace{\sum_{\substack{i\in\mathcal{U}\\\varphi(i)=\{j\}}}\hspace{-0.15cm} c_{i,j}}}~+\underset{\text{removals}}{\underbrace{\sum_{\substack{i\in\mathcal{U}\\\varphi(i)=\{\epsilon\}}}\hspace{-0.1cm}c_{i,m+1}}}~+\underset{\text{insertions}}{\underbrace{\sum_{\substack{j\in\mathcal{V}\\\varphi^{-1}[j]=\{\epsilon\}}}\hspace{-0.18cm}c_{n+1,j}}}\\\\
        &=\underset{\text{substitutions}}{\underbrace{\sum_{\substack{i\in\mathcal{U}^\sharp}} c_{i,\widetilde{\varphi}(i)}}}~+\underset{\text{removals}}{\underbrace{\sum_{\substack{i\in\mathcal{U}\setminus\mathcal{U}^\sharp}}\hspace{-0.1cm}c_{i,m+1}}}~+\underset{\text{insertions}}{\underbrace{\sum_{\substack{j\in\mathcal{V}\setminus\mathcal{V}^\sharp}}\hspace{-0.1cm}c_{n+1,j}}}.
      \end{aligned}
\end{equation}
\begin{problem}[LSAPE]\label{pb-assigneps}
Given to sets $\mathcal{U}$ and $\mathcal{V}$, the linear sum assignment problem with edition (LSAPE) consists in finding an $\epsilon$-assignment having a minimal cost among all $\epsilon$-assignments transforming $\mathcal{U}$ into $\mathcal{V}$, \textit{i.e.} satisfying 
\begin{equation}\label{eq-lsapevphi}
		\underset{\varphi\in\mathcal{A}_\epsilon(\mathcal{U},\mathcal{V})}{\argmin}\,A_\epsilon(\varphi,\mathbf{C}).
\end{equation}
\end{problem}
\noindent
This problem is equivalent to the sLSAPE, as demonstrated in the following section.
\subsection{Link to the sLSAPE}
Recall that the sLAPE finds a squared $\epsilon$-assignment from
$\mathcal{U}$ to $\mathcal{V}$, \textit{i.e.} a bijection from
$\mathcal{U}\,{\cup}\,\mathcal{E}_\mathcal{U}$ to
$\mathcal{V}\,{\cup}\,\mathcal{E}_\mathcal{V}$, with
$\mathcal{E}_\mathcal{U}\,{=}\,\{\epsilon_j,~ j\,{\in}\,\mathcal{V}\}$
and $\mathcal{E}_\mathcal{V}\,{=}\,\{\epsilon_i,~
i\,{\in}\,\mathcal{U}\}$, and such that $i\,{\in}\,\mathcal{U}$
(resp. $\epsilon_j\,{\in}\,\mathcal{E}_\mathcal{U}$) cannot be
assigned to $\epsilon_k\,{\in}\,\mathcal{E}_\mathcal{V}$
(resp. $l\,{\in}\,\mathcal{V}$) with $k\,{\not=}\,i$
(resp. $l\,{\not=}\,j$). This is penalized through the cost values by
setting a large value to such forbidden couplings. Also, the cost of
assigning any element of $\mathcal{E}_\mathcal{U}$ to any element of
$\mathcal{E}_\mathcal{V}$ is null (equivalent to
$\epsilon\,{\rightarrow}\,\epsilon$, which has also a zero cost). We
assume that any squared $\epsilon$-assignment has a non-infinite cost
$A_\epsilon$.
\begin{lemma}\label{prop:edtosquared}
  Each $\epsilon$-assignment $\varphi\,{\in}\,\mathcal{A}_\epsilon(\mathcal{U},\mathcal{V})$ is associated to a bijection $\psi\,{\in}\,\mathcal{S}_{\mathcal{E}}(\mathcal{U},\mathcal{V})$ (Problem~\ref{def-squaredlsap}) such that:
\begin{equation}\label{eq-varphipsi1}
	\left\{\begin{array}{ll}
      \forall i\,{\in}\,\mathcal{U}^\sharp,&\psi(i)=\varphi(i)\\
	  \forall i\,{\in}\,\mathcal{U}\,{\setminus}\,\mathcal{U}^\sharp,&\psi(i)=\epsilon_i\\
	  \forall j\,{\in}\,\mathcal{V}\,{\setminus}\,\mathcal{V}^\sharp,&\psi(\epsilon_j)=j
  	\end{array}\right.
\end{equation}
and the restriction of $\psi$ to
\begin{equation}\label{eq-varphipsi2}
  	\left(\begin{array}{rcl}
  		\mathcal{E}_\mathcal{U}^\sharp=\{\epsilon_j\in \mathcal{E}_\mathcal{U} \,|\, 
  j\in \mathcal{V}^\sharp\}&\rightarrow&
  \mathcal{E}_\mathcal{V}^\sharp=\{\epsilon_i\in \mathcal{E}_\mathcal{V}\,|\,
 i\in \mathcal{U}^\sharp\}\\
 		\epsilon_j&\mapsto&\epsilon_i
  	\end{array}\right)
\end{equation}
is bijective. Moreover, mappings $\varphi$ and $\psi$ have the same
cost (Eq.~\ref{eq-Aepssq} and Eq.~\ref{eq-Aeps}).
\end{lemma}
\begin{proof}
  Let $\varphi\,{\in}\,\mathcal{A}_\epsilon(\mathcal{U},\mathcal{V})$ be an $\epsilon$-assignment, and let $\widetilde{\varphi}$ be its associated injection. Consider the mapping $\psi$ defined by Eq.~\ref{eq-varphipsi1} and Eq.~\ref{eq-varphipsi2}. We show that $\psi\,{\in}\,\mathcal{S}_{\mathcal{E}}(\mathcal{U},\mathcal{V})$.
  \begin{table}
      \[\begin{small}
      \begin{array}{|l|l|l|}
        \hline
        \text{Mapping}&\text{Initial set}&\text{Arrival set}\\
        \hline
        \psi^{\text{sub}}\,{=}\,\widetilde{\varphi}&\mathcal{U}^\sharp&\mathcal{V}^\sharp\\
        \psi^{\text{rem}}&      \mathcal{U}\setminus\mathcal{U}^\sharp&\mathcal{E}_\mathcal{V}\setminus\mathcal{E}_\mathcal{V}^\sharp\\
        \psi^{\text{ins}}&\mathcal{E}_\mathcal{U}\setminus\mathcal{E}_\mathcal{U}^\sharp   &  \mathcal{V}\setminus\mathcal{V}^\sharp\\
        \psi^{\text{comp}}&\mathcal{E}_\mathcal{U}^\sharp&\mathcal{E}_\mathcal{V}^\sharp\\
        \hline
      \end{array}
      \end{small}\]
      \caption{Arrival and terminal sets involved in the construction of a squared $\epsilon$-assignment from an $\epsilon$-assignment.}
      \label{tab:build_lsape_map}
    \end{table}
  The mapping $\psi\,{:}\,\mathcal{U}\,{\cup}\,\mathcal{E}_\mathcal{U}\,{\rightarrow}\,\mathcal{V}\,{\cup}\,\mathcal{E}_\mathcal{V}$ is built as described in Table~\ref{tab:build_lsape_map}:
  \begin{enumerate}
  \item[$\bullet$]\textit{Substitutions}. When $i\,{\in}\,\mathcal{U}^\sharp$, we have $\psi(i)\,{=}\,\widetilde{\varphi}(i)$. Since $\widetilde{\varphi}$ is a bijection from $\mathcal{U}^\sharp$ onto $\mathcal{V}^\sharp=\widetilde{\varphi}[\mathcal{U}^\sharp]$, the restriction $\psi^{\text{sub}}\,{:}\,\mathcal{U}^\sharp\,{\rightarrow}\,\psi[\mathcal{U}^\sharp]\,{=}\,\mathcal{V}^\sharp$ of $\psi$ is also a bijection.
  \item[$\bullet$]\textit{Removals/insertions}. When $i\,{\in}\,\mathcal{U}\,{\setminus}\,\mathcal{U}^\sharp$ (reps. $j\,{\in}\,\mathcal{V}\,{\setminus}\,\mathcal{V}^\sharp$), we have $\psi(i)\,{=}\,\epsilon_i$ (resp. $\psi(\epsilon_j)\,{=}\,j$). By definition, there is 
  a natural bijective mapping from any $i\,{\in}\,\mathcal{U}$ (resp. $\epsilon_j\,{\in}\,\mathcal{E}_\mathcal{U}$) to $\epsilon_i\,{\in}\,\mathcal{E}_\mathcal{V}$ (reps. $j\,{\in}\,\mathcal{V}$). So the restriction $\psi^{\text{rem}}\,{:}\,\mathcal{U}\,{\setminus}\,\mathcal{U}^\sharp\,{\rightarrow}\,\mathcal{E}_\mathcal{V}\,{\setminus}\,\mathcal{E}^\sharp_\mathcal{V}$ of $\psi$, and the restriction $\psi^{\text{ins}}\,{:}\,\mathcal{E}_\mathcal{U}\,{\setminus}\,\mathcal{E}^\sharp_\mathcal{U}\,{\rightarrow}\,\mathcal{V}\,{\setminus}\,\mathcal{V}^\sharp$, are bijective.
  \item[$\bullet$]\textit{Completion.} Meanwhile all edit operations have been defined through the two last items, the mapping $\psi$ is not yet complete. Indeed, it is not yet defined from $\mathcal{E}^\sharp_\mathcal{U}$ to $\mathcal{E}^\sharp_\mathcal{V}$. Since $\psi^{\text{sub}}$ is a bijection, we have $|\mathcal{U}^\sharp|\,{=}\,|\mathcal{V}^\sharp|$ and by consequence $|\mathcal{E}^\sharp_\mathcal{U}|\,{=}\,|\mathcal{E}^\sharp_\mathcal{U}|$. So the restriction $\psi^{\text{comp}}\,{:}\,\mathcal{E}^\sharp_\mathcal{U}\,{\rightarrow}\,\mathcal{E}^\sharp_\mathcal{V}$ of $\psi$ can be any bijection from $\mathcal{E}^\sharp_\mathcal{U}$ to $\mathcal{E}^\sharp_\mathcal{V}$.
\end{enumerate}
    All
    initial sets and arrival sets are disjoints and the mapping defined
    between each couple of sets is bijective. The mapping $\psi$ is
    thus bijective. Moreover, considering that any mapping of
    $\epsilon_j\in \mathcal{E}_\mathcal{U }^\sharp$ onto
    $\psi^{\text{comp}}(\epsilon_j)\in \mathcal{E}_\mathcal{V}^\sharp$ is associated to a $0$ cost, both mappings have the same cost.
\end{proof}
\noindent
Note that one $\epsilon$-assignment corresponds to several bijections of $\mathcal{S}_\mathcal{E}(\mathcal{U},\mathcal{V})$ due to the bijective mapping $\psi^{\text{comp}}$. The proof also show that all these bijections have the same cost.
\begin{lemma}\label{prop:edtosquared2}
  Any bijection $\psi\,{\in}\,\mathcal{S}_\mathcal{E}(\mathcal{U},\mathcal{V})$ is associated to an $\epsilon$-assignment $\varphi\,{\in}\,\mathcal{A}_\epsilon(\mathcal{U},\mathcal{V})$ 
such that:
	\begin{equation}\label{eq-psivarphi}
		\left\{\begin{array}{ll}
			\varphi(i)=\{\psi(i)\}&\text{if}~\,\psi(i)\,{\in}\,\mathcal{V}\\
			\varphi(i)=\{\epsilon\}&\text{if}~\,\psi(i)\,{\in}\,\mathcal{E}_\mathcal{V}\\
			\varphi(\epsilon)=\psi[\mathcal{E}_\mathcal{U}^\sharp]\cup\{\epsilon\}&\text{with}~\,\mathcal{E}_\mathcal{U}^\sharp\,{=}\,\{\epsilon_j\,{\in}\,\mathcal{E}_\mathcal{U}~|~\psi(\epsilon_j)\,{\in}\,\mathcal{V}\}
		\end{array}\right.
	\end{equation}	  
  Mappings $\psi$ and $\varphi$ have the same cost.
\end{lemma}
\begin{proof}
  Let $\psi\,{\in}\,\mathcal{S}_\mathcal{E}(\mathcal{U},\mathcal{V})$ be a squared $\epsilon$-assignment. Since it defines a bijection from
  $\mathcal{U}\cup\mathcal{E}_{\mathcal{U}}$ to
  $\mathcal{V}\cup\mathcal{E}_{\mathcal{V}}$, its restriction
  $\widetilde{\varphi}$ to $\mathcal{U}^\sharp =\{ u\in\mathcal{U}\,|\, \psi(u)\in \mathcal{V}\}$ is injective from
  $\mathcal{U}^\sharp$ to $\mathcal{V}$. Such an injective mapping
  corresponds to a unique $\epsilon$-assignment from
  $\mathcal{U}$ to $\mathcal{V}$ (Proposition~\ref{prop:bijAESA})
  whose cost is provided by Eq.~\ref{eq-Aeps} (with
  $\mathcal{V}^\sharp=\widetilde{\varphi}[\mathcal{U}^\sharp]$). By construction (Eq.~\ref{eq-psivarphi}), $\varphi$ has the same cost as $\psi$ (Eq.~\ref{eq-Aepssq}).
\end{proof}
\noindent
A simple consequence of Lemma~\ref{prop:edtosquared} and Lemma~\ref{prop:edtosquared2} is given by the following property.
\begin{proposition}
   The LSAPE and the sLSAPE are equivalent, their respective solutions provide the same minimal cost $A_\epsilon$.
\end{proposition}
\noindent
Note that even if the minimal cost is the same for the LSAPE and the sLSAPE,
the proof of this proposition shows that solving a sLSAPE
requires to define additionally a
bijective mapping between two sets of $\epsilon$-values. This mapping
being useless in terms of edit costs, computing an optimal $\epsilon$-assignment should be more efficient than its squared version.
\subsection{Matrix form and linear programming}\label{sec-linprog}
Let $\mathbf{C}\,{\in}\,[0,{+\infty})^{(n+1)\times(m+1)}$ be an edit cost matrix. From Eq.~\ref{eq-Aeps}, the cost associated to an $\epsilon$-assignment given by a matrix $\mathbf{X}\,{\in}\,\mathcal{S}_{n,m,\epsilon}$ (Eq.~\ref{eq-assignedX} and Eq.~\ref{eq-mtx-assigneps}), is measured by
\begin{equation}
	A_\epsilon(\mathbf{X},\mathbf{C})=\sum_{i=1}^{n+1}\sum_{j=1}^{m+1}c_{i,j}x_{i,j}.
\end{equation}
The LSAPE (Problem~\ref{pb-assigneps}) can then be rewritten as finding a matrix of $\mathcal{S}_{n,m,\epsilon}$ having a minimal cost, \textit{i.e.} satisfying
\begin{equation}\label{eq-lsapemat}
	\underset{\mathbf{X}\in \mathcal{S}_{n,m,\epsilon}}{\argmin}~A_\epsilon(\mathbf{X},\mathbf{C})
\end{equation}
The LSAPE can be rewritten as a linear programming problem, as described in the following.
\begin{figure*}[!t]
\begin{equation*}
\begin{small}
	\begin{blockarray}{@{\extracolsep{-0.5cm}}rcccccccccccccccc}
		~ &_{(1,1)}&\ldots &_{(1,m)}&_{(1,\epsilon)}&_{(2,1)}&\ldots &_{(2,m)}&_{(2,\epsilon)}&   \cdots    &_{(n,1)}&\ldots &_{(n,m)}&_{(n,\epsilon)}&_{(\epsilon,1)}&\ldots &_{(\epsilon,m)}\vspace{0.1cm}\\
		\begin{block}{@{\extracolsep{-0.5cm}}r(cccc|cccc|c|cccc|ccc)}
		1 &~\,1& \cdots & 1 & 1 & ~ &   ~   & ~ & ~ & ~      & ~ &   ~   & ~ & ~ & ~ & ~      & ~\\
		2 & ~ & ~      & ~ & ~ &~\,1&\cdots & 1 & 1 &        & ~ &   ~   & ~ & ~ &   &        & ~\\
		~  & ~ & ~      & ~ & ~ &   &       &   &   & \cdots & ~ &   ~   & ~ & ~ &   &        & ~\\
		n  & ~ & ~      & ~ & ~ &   &       &   &   &        &~\,1&\cdots & 1 & 1 &   &        &  \\\cline{2-17}
1   & ~\,1 & ~      & ~ & ~ & ~\,1 &       &   &   &        & ~\,1 &   ~   & ~ & ~ & ~\,1 &        & \\
		~  &   & \ddots &   & ~ &   &\ddots &   &   & \cdots & ~ &\ddots & ~ & ~ &   & \ddots & ~\\
 m  &   &        & 1 & 0 &   &       & 1 & 0 &        & ~ &   ~   & 1 &0  &   &        & 1\\
 	\end{block}
	\end{blockarray}
\end{small}
\end{equation*}
\caption{Constraint matrix $\mathbf{L}$. Its $n$ first rows represent $\mathcal{U}$ and its $m$ last rows represent $\mathcal{V}$. Missing values are equal to $0$.\label{fig-mat}}
\end{figure*}

Let $\mathbf{x}\,{=}\,\text{vec}(\mathbf{X})\,{\in}\,\{0,1\}^{(n+1)(m+1)-1}$ be the
vectorization of $\mathbf{X}$ obtained by concatenating its
rows and by removing its last element (since $x_{n+1,m+1}c_{n+1,m+1}\,{=}\,0$). Then the
constraints can be rewritten as the linear system of equations
$\mathbf{L}\mathbf{x}\,{=}\,\mathbf{1}$, where the matrix
$\mathbf{L}\,{\in}\,\{0,1\}^{(n+m)\times[(n+1)(m+1)-1]}$ is given by Fig.~\ref{fig-mat}, \textit{i.e.} 
\begin{equation}
\forall (i,j),~~\left\lbrace
\begin{aligned}
	&l_{k,(i,j)}=\delta_{k=i},&&\forall k\,{=}\,1,\ldots,n\\
	&l_{n+k,(i,j)}=\delta_{k=j},&&\forall k\,{=}\,1,\ldots,m
\end{aligned}\right.
\end{equation}
where $\delta_{a=b}\,{=}\,1$ if $a\,{=}\,b$ and $0$ else. The $n$ first rows of $\mathbf{L}$ represent the constraints on the
rows of $\mathbf{X}$ ($|\varphi(i)|\,{=}\,1, \forall i\,{\in}\,\mathcal{U}$)
and its $m$ last rows represent the constraints on the columns of
$\mathbf{X}$ ($|\varphi^{-1}[j]|\,{=}\,1, \forall j\,{\in}\,\mathcal{V}$). Each element $l_{k,(i,j)}\,{=}\,1$ corresponds to a possible edit operation:
\begin{enumerate}
\item[$\bullet$] a substitution $i\rightarrow j$, when $k\,{=}\,i\,{\in}\,\mathcal{U}$ or $n\,{+}\,j$ with $j\,{\in}\,\mathcal{V}$,
\item[$\bullet$] a removal $i\rightarrow\epsilon$ when $k\,{=}\,i\,{\in}\,\mathcal{U}$ and $j\,{=}\,m\,{+}\,1$, or
\item[$\bullet$] an insertion $\epsilon\rightarrow j$ when $k\,{=}\,n\,{+}\,j$ with $j\,{\in}\,\mathcal{V}$, and $i\,{=}\,n\,{+}\,1$.
\end{enumerate}
Note that $\epsilon$ is not represented as a row of $\mathbf{L}$ since it is unconstrained. The matrix $\mathbf{L}$ represents the discrete domain whereon solutions to the LSAPE are defined.

Contrary to the constraint matrix involved in the LASP, which
corresponds to the node-edge incidence matrix of the complete
bipartite graph $K_{n,m}$ with node sets $\mathcal{U}$ and
$\mathcal{V}$, the constraint matrix $\mathbf{L}$ can be viewed as a
node-edge incidence matrix of the mixed bipartite graph
$K_{n,m,\epsilon}$ composed of $K_{n,m}$ and the two bipartite
digraphs
$((\mathcal{V},\{\epsilon\}),\mathcal{V}\,{\times}\,\{\epsilon\})$ and
$((\mathcal{U},\{\epsilon\}),\mathcal{U}\,{\times}\,\{\epsilon\})$
representing insertion and removal operations, respectively. It is
important to remark that $K_{n,m,\epsilon}$ does not have any arc from
$\epsilon$ to $\mathcal{U}\,{\cup}\,\mathcal{V}$. Due to the
definition of $\mathbf{x}$, $\mathbf{L}\mathbf{x}$ selects exactly one
edge or arc in the neighborhood of each node in
$\mathcal{U}\,{\cup}\,\mathcal{V}$. An $\epsilon$-assignment can thus
be assimilated to a subgraph of $K_{n,m,\epsilon}$ having each node in
$\mathcal{U}\,{\cup}\,\mathcal{V}$ connected to only one other
node. Solving the LSAPE consists in computing such a subgraph that
minimizes the objective functional $A_\epsilon$.

Similarly, let $\mathbf{c}\,{=}\,\text{vec}(\mathbf{C})\,{\in}\,[0,+\infty)^{(n+1)(m+1)-1}$ be the vectorization of the edit cost matrix $\mathbf{C}$. With these notations, the LSAPE given by Eq.~\ref{eq-lsapemat} consists in selecting a vector satisfying
\begin{equation}\label{eq-lsapecx}
		\argmin\left\{\mathbf{c}^T\mathbf{x}~|~\mathbf{L}\mathbf{x}\,{=}\,\mathbf{1},~\mathbf{x}\in\{0,1\}^{(n+1)(m+1)-1}\right\}
\end{equation}
which is a binary linear program.

By definition, each column of $\mathbf{L}$ sums to no more than $2$,
and its rows can be partitioned into two sets, $\mathcal{U}$ and
$\mathcal{V}$, such that a $1$ appears in each column at most once in
each set (Fig.~\ref{fig-mat}). Then the matrix $\mathbf{L}$ is totally unimodular. By standard tools in linear programming, a binary linear programming
problem with totally unimodular constraint matrix has always a binary
optimal solution, see \cite{sier15} for more details on linear programming problems. So the LSAPE has a binary optimal solution. As for the LSAP (and the sLSAPE), this shows that the LSAPE can be solved with linear programming tools.
%
\subsection{Primal-dual problem and admissible transformations}
Consider the LSAPE expressed as a binary linear program (Eq.~\ref{eq-lsapecx}). By standard tools in duality theory and linear programming \cite{sier15}, the problem dual to the LSAPE consists in finding two vectors $\mathbf{u}\,{\in}\,[0,+\infty)^{n}$ and $\mathbf{v}\,{\in}\,[0,+\infty)^{m}$, which associate a real value to each element of the sets $\mathcal{U}$ and $\mathcal{V}$ respectively, and satisfying
\begin{equation*}
	\argmax\left\{\mathbf{1}^T_n\mathbf{u}+\mathbf{1}^T_m\mathbf{v}~|~\mathbf{L}^T\begin{pmatrix}
	\mathbf{u}\\
	\mathbf{v}
	\end{pmatrix}\leq\mathbf{c},\,\begin{pmatrix}
	\mathbf{u}\\
	\mathbf{v}
	\end{pmatrix}\in[0,+\infty)^{n+m}\right\}
\end{equation*}
or equivalently
\begin{equation}\label{eq-duallsape}
	\begin{aligned}
		&\underset{(\mathbf{u},\mathbf{v})}{\argmax}\left\{E(\mathbf{u},\mathbf{v})\overset{\text{def.}}{=}\mathbf{1}^T_n\mathbf{u}+\mathbf{1}^T_m\mathbf{v}\right\}\\
		&~~~~~~~\text{s.t.}~\mathbf{u}\,{\in}\,[0,+\infty)^n,\,\mathbf{v}\,{\in}\,[0,+\infty)^{m}\\
		&~~~~~~~~~~~~~u_i+v_j\leq c_{i,j},~\forall i=1,\ldots,n,~\forall j=1,\ldots,m\\
		&~~~~~~~~~~~~~u_i\leq c_{i,m+1},~\forall i=1,\ldots,n\\
		&~~~~~~~~~~~~~v_j\leq c_{n+1,j},~\forall j=1,\ldots,m
	\end{aligned}
\end{equation}
One can remark that the objective functional $E(\mathbf{u},\mathbf{v})$ is the same as the one involved in the problem dual to the LSAP, see \cite{bur09}, but the two last constraints are different.

By the strong duality theorem, if $\mathbf{X}$ is a solution to the LSAPE (primal problem), then its dual problem has an optimal solution $(\mathbf{u},\mathbf{v})$ and the solutions $(\mathbf{X},(\mathbf{u},\mathbf{v}))$ satisfy
\begin{equation}\label{eq-primdual}
	A_\epsilon(\mathbf{X},\mathbf{C})\,{=}\,E(\mathbf{u},\mathbf{v}).
\end{equation}
Such a primal-dual problem can be solved by finding transformations of the edit cost matrix, called admissible. We adapt the notion of admissible transformations of cost matrices (or equivalent cost matrices), developed in the context of the classical assignment problem to edit cost matrices, see for instance \cite{bur09}.

Note that the dual variables can be expressed as a pair $(\mathbf{u},\mathbf{v})$ such that $\mathbf{u}\,{\in}\,[0,+\infty)^{n+1}$ with $u_{n+1}\,{=}\,0$, and $\mathbf{v}\,{\in}\,[0,+\infty)^{m+1}$ with $v_{m+1}\,{=}\,0$, without altering objective functional $E(\mathbf{u},\mathbf{v})$ and leading to a rewriting of the three last constraints in Eq.~\ref{eq-duallsape} as
\begin{equation}\label{eq-cstdual}
	u_i+v_j\leq c_{i,j},~~\forall i=1,\ldots,n\,{+}\,1,~j=1,\ldots,m\,{+}\,1.
\end{equation}
We use this trick to simplify the forthcoming expressions.
\begin{definition}[$\epsilon$-admissible transformation]
  Let $T(\mathbf{C},\overline{\mathbf{C}})$ be the transformation of an edit cost matrix
  $\mathbf{C}$ into a matrix $\overline{\mathbf{C}}$, with $\mathbf{C},\overline{\mathbf{C}}\,{\in}\,[0,{+\infty})^{(n+1)\times(m+1)}$. $T$ is called $\epsilon$-admissible with index
  $A(T)\,{\in}\,\mathbb{R}$ if $\overline{\mathbf{C}}$ is an edit cost matrix and
   \[
   A_\epsilon(\mathbf{X},\mathbf{C})\,{=}\,A_\epsilon(\mathbf{X},\overline{\mathbf{C}})+A(T)
   \]
  for any $\epsilon$-assignment $\mathbf{X}$.
\end{definition}
\noindent
By consequence, an $\epsilon$-assignment $\mathbf{X}$ which minimizes
$A_\epsilon(\mathbf{X},\mathbf{C})$ also minimizes
$A_\epsilon(\mathbf{X},\overline{\mathbf{C}})+A(T)$, and so the edit
cost matrices $\mathbf{C}$ and $\overline{\mathbf{C}}$ are said to be
equivalent (w.r.t. the LSAPE). Remark that transformations are assumed to produce edit cost matrices. In particular, transformed matrices must satisfy $\overline{\mathbf{C}}\,{\geq}\,\mathbf{0}$.
\begin{lemma}\label{lem-opt}
  Let $T(\mathbf{C},\overline{\mathbf{C}})$ be an $\epsilon$-admissible transformation. Let
  $\hat{\mathbf{X}}$ be an $\epsilon$-assignment satisfying
  $A_\epsilon(\hat{\mathbf{X}},\overline{\mathbf{C}})\,{=}\,0$. Then
  $\hat{\mathbf{X}}$ represents an optimal $\epsilon$-assignment, with total cost
  $A_\epsilon(\hat{\mathbf{X}},\mathbf{C})\,{=}\,A(T)$.
\end{lemma}
\begin{proof}
Let $\mathbf{X}$ be an arbitrary $\epsilon$-assignment. Under the conditions of the lemma, we have $\overline{\mathbf{C}}\geq\mathbf{0}$ and then $A_\epsilon(\mathbf{X},\overline{\mathbf{C}})\geq 0$. So $A_\epsilon(\mathbf{X},\mathbf{C})=A_\epsilon(\mathbf{X},\overline{\mathbf{C}})+A(T)\geq A(T)=A_\epsilon(\hat{\mathbf{X}},\overline{\mathbf{C}})+A(T)=\,A_\epsilon(\hat{\mathbf{X}},\mathbf{C})$, and $A_\epsilon(\mathbf{X},\mathbf{C})$ reaches its minimum value for $\mathbf{X}\,{=}\,\hat{\mathbf{X}}$.
\end{proof}
\noindent
Admissible transformations of edit cost matrices are linked to the primal-dual equation \ref{eq-primdual} as follows.
\begin{proposition}\label{theo-admuv}
  Let $\mathbf{C}$ be an edit
  cost matrix. Let $\mathbf{a}\in\mathbb{R}^{n+1}$ and
  $\mathbf{b}\in\mathbb{R}^{m+1}$ be two vectors such that
  $a_{n+1}=b_{m+1}=0$ and~\,$\mathbf{a}\mathbf{1}^T_{m+1}+\mathbf{1}_{n+1}\mathbf{b}^T\leq\mathbf{C}$. The
  transformation $T(\mathbf{C},\overline{\mathbf{C}})$ such that
  $\overline{\mathbf{C}}=\mathbf{C}\,{-}\,\mathbf{a}\mathbf{1}^T_{m+1}\,{-}\,\mathbf{1}_{n+1}\mathbf{b}^T$ is $\epsilon$-admissible with index $A(T)\,{=}\,E(\mathbf{a},\mathbf{b})$.
\end{proposition}
\begin{proof}Let $\mathbf{X}$ be any $\epsilon$-assignment. We have:
	\begin{equation*}
	\begin{aligned}
	A_\epsilon(\mathbf{X},\overline{\mathbf{C}})&=\sum_{i=1}^{n+1}\sum_{j=1}^{m+1}(c_{i,j}-a_i-b_j)x_{i,j}=A_\epsilon(\mathbf{X},\mathbf{C})-\sum_{i=1}^{n}a_i\underset{=1}{\underbrace{\sum_{j=1}^{m+1}x_{i,j}}}-\sum_{j=1}^{m}b_j\underset{=1}{\underbrace{\sum_{i=1}^{n+1}x_{i,j}}}\\
	~&=A_\epsilon(\mathbf{X},\mathbf{C})-\left(\mathbf{1}^T_{n+1}\mathbf{a}+\mathbf{1}^T_{m+1}\mathbf{b}\right)=A_\epsilon(\mathbf{X},\mathbf{C})-E(\mathbf{a},\mathbf{b}).
	\end{aligned}
	\end{equation*}
	So $A_\epsilon(\mathbf{X},\mathbf{C})=A_\epsilon(\mathbf{X},\overline{\mathbf{C}})+E(\mathbf{a},\mathbf{b})$. Remark that $\overline{\mathbf{C}}\geq\mathbf{0}$ with $c_{n+1,m+1}\,{=}\,0$, so $\overline{\mathbf{C}}$ is an edit cost matrix and the transformation is $\epsilon$-admissible.
\end{proof}
\noindent
Consider a dual solution $(\mathbf{u},\mathbf{v})$. Recall that $u_{n+1}\,{=}\,v_{m+1}\,{=}\,0$. By definition the dual solution fulfils the conditions of Proposition~\ref{theo-admuv}. This implies the following complementary slackness condition.
\begin{proposition}\label{th-compclack}
$(\mathbf{X},(\mathbf{u},\mathbf{v}))$ solves the LSAPE and its dual problem (primal-dual equation Eq.~\ref{eq-primdual}) iff $A_\epsilon(\mathbf{X},\overline{\mathbf{C}})\,{=}\,0$ with $$\overline{\mathbf{C}}\,{=}\,\mathbf{C}\,{-}\,\mathbf{u}\mathbf{1}^T_{m+1}\,{-}\,\mathbf{1}_{n+1}\mathbf{v}^T$$
\textit{i.e.} for all $(i,j)\,{\in}\,\{1,\ldots,n\,{+}\,1\}\,{\times}\,\{1,\ldots,m\,{+}\,1\}$ 
\begin{equation}\label{eq-compslack}
	\begin{aligned}
		&x_{i,j}\overline{c}_{i,j}\,{=}\,0~\Leftrightarrow~\left(\left(x_{i,j}\,{=}\,1\right)\wedge\left(\overline{c}_{i,j}\,{=}\,0\right)\right)\vee\left(\left(x_{i,j}\,{=}\,0\right)\wedge\left(\overline{c}_{i,j}\,{\geq}\,0\right)\right)\\
		&\Leftrightarrow~\left(\left(x_{i,j}\,{=}\,1\right)\wedge\left(c_{i,j}\,{=}\,u_i\,{+}\,v_j\right)\right)\vee\left(\left(x_{i,j}\,{=}\,0\right)\wedge\left(c_{i,j}\,{\geq}\,u_i\,{+}\,v_j\right)\right)
	\end{aligned}
\end{equation}
The optimal cost is $A_\epsilon(\mathbf{X},\mathbf{C})\,{=}\,E(\mathbf{u},\mathbf{v})$.
\end{proposition}
\begin{proof}
Suppose that $(\mathbf{X},(\mathbf{u},\mathbf{v}))$ is a solution to the primal-dual equation \ref{eq-primdual}. By Eq.~\ref{eq-cstdual} we have $\mathbf{u}\mathbf{1}^T\,{+}\,\mathbf{1}\mathbf{v}^T\leq\mathbf{C}$. Then by Proposition \ref{theo-admuv}, $\overline{\mathbf{C}}\,{=}\,\mathbf{C}-(\mathbf{u}\mathbf{1}^T\,{+}\,\mathbf{1}\mathbf{v}^T)$ is an $\epsilon$-admissible transformation: $A_\epsilon(\mathbf{X},\mathbf{C})\,{=}\,A_\epsilon(\mathbf{X},\overline{\mathbf{C}})+E(\mathbf{u},\mathbf{v})$. Since $A_\epsilon(\mathbf{X},\mathbf{C})\,{=}\,E(\mathbf{u},\mathbf{v})$ by Eq.~\ref{eq-primdual}, so $A_\epsilon(\mathbf{X},\overline{\mathbf{C}})\,{=}\,0$, which demonstrates the first part ($\Rightarrow$). Now suppose that $\mathbf{X}$ is an $\epsilon$-assignment, and $\overline{\mathbf{C}}$ is an edit cost matrix such that $\overline{\mathbf{C}}\,{=}\,\mathbf{C}-(\mathbf{u}\mathbf{1}^T\,{+}\,\mathbf{1}\mathbf{v}^T)$ and $A_\epsilon(\mathbf{X},\overline{\mathbf{C}})\,{=}\,0$. By Lemma \ref{lem-opt} $\mathbf{X}$ is optimal with cost $A_\epsilon(\mathbf{X},\mathbf{C})\,{=}\,E(\mathbf{u},\mathbf{v})$.
\end{proof}
\noindent
This condition is the same for the LSAP, but $\mathbf{X}$ and $\mathbf{C}$ do not satisfy the same constraints. Solving the LSAP is equivalent to find a transformed cost matrix having $k$ independent zero elements ($k\,{=}\,\min\{n,m\}$). Elements of a matrix are independent if none of them occupies the same row or column. According to K{\"o}nig's theorem and Eq.~\ref{eq-compslack}, this independent set of zero elements exactly corresponds to an optimal assignment \cite{kuhn55,munk57,bour71}. Hungarian-type algorithms are based on this property. In order to derive similar algorithms for solving the LSAPE, the notion of independent elements needs to be adapted.
\begin{definition}[$\epsilon$-independent]\label{def-epsindep}
Let $\mathcal{S}$ be a set of elements of a $(n\,{+}\,1)\,{\times}\,(m\,{+}\,1)$ edit cost matrix $\mathbf{C}$. Elements of $\mathcal{S}$ are $\epsilon$-independent if there is at most one element of $\mathcal{S}$ on each row $i\,{=}\,1,\ldots,n$ of $\mathbf{C}$, and at most one element of $\mathcal{S}$ on each column $j\,{=}\,1,\ldots,m$ of $\mathbf{C}$. The set $\mathcal{S}$ is maximal when there is exactly one element of $\mathcal{S}$ on each row $i\,{=}\,1,\ldots,n$, and exactly one element of $\mathcal{S}$ on each column $i\,{=}\,1,\ldots,n$.
\end{definition}
\noindent
Obviously, by definition, any $\epsilon$-assignment contains a maximal set of $\epsilon$-independent elements equals to $1$. This implies the following property.
\begin{corollary}\label{cor-compslack}
Let $\mathcal{S}$ be a maximal set of $\epsilon$-independent elements of an edit cost matrix $\mathbf{C}$. Let $\mathbf{u}$ and $\mathbf{v}$ be two vectors such that 
\begin{equation}
\left\lbrace\begin{aligned}
&c_{i,j}=u_i+v_j,&&\forall (i,j)\,{\in}\,\mathcal{S},\\
&c_{i,j}\geq u_i+v_j,&&\forall (i,j)\,{\not\in}\,\mathcal{S}.
\end{aligned}\right.
\end{equation}
Then the matrix $\mathbf{X}$ constructed such that 
$x_{i,j}=\delta_{(i,j)\in\mathcal{S}}$
is an optimal $\epsilon$-assignment for the LSAPE, and $(\mathbf{u},\mathbf{v})$ solves its dual problem.
\end{corollary}
\begin{proof}
By construction and from Definition \ref{def-epsindep}, $\mathbf{X}$ defines an $\epsilon$-assignment. Also, $\mathbf{X}$ and $\mathbf{C}$ satisfy the complementary slackness condition given by Eq.~\ref{eq-compslack}. So by Proposition~\ref{th-compclack} $(\mathbf{X},(\mathbf{u},\mathbf{v}))$ solves the LSAPE and its dual problem.
\end{proof}
\noindent
The Hungarian-type algorithm proposed in the following section finds such a maximal set of $\epsilon$-independent elements.


\section{A Hungarian algorithm for solving the LSAPE}\label{sec-algo}
We propose to compute a solution to the LSAPE (Eq.~\ref{pb-assigneps} or Eq.~\ref{eq-lsapemat}), and its dual problem (Eq.~\ref{eq-duallsape}), by adapting the Hungarian algorithm that initially solves the LSAP and it dual problem \cite{kuhn55,munk57,bour71,law76,bur09}. There is several versions of the Hungarian algorithm. We adapt the one presented in~\cite{law76,bur09}, which is known to be more efficient both in time and space.

We consider an edit cost matrix $\mathbf{C}\,{\in}\,[0,+\infty)^{(n+1)\times(m+1)}$  (Eq.~\ref{eq:matrixC}) between two sets $\mathcal{U}\,{=}\,\{1,\ldots,n\}$ and $\mathcal{V}\,{=}\,\{1,\ldots,m\}$. Recall that row $n\,{+}\,1$ of $\mathbf{C}$ encodes the cost of inserting 
each element of $\mathcal{V}$ into $\mathcal{U}$. Similarly, column $m\,{+}\,1$ encodes the cost of removing each element of $\mathcal{U}$. Consider the sets $\mathcal{U}_\epsilon\,{=}\,\mathcal{U}\,{\cup}\,\{n\,{+}\,1\}$ and $\mathcal{V}_\epsilon\,{=}\,\mathcal{V}\,{\cup}\,\{m\,{+}\,1\}$. The proposed adaptation of the Hungarian algorithm takes into account the relaxation of the constraints associated to removal and insertion operations.

Intuitively, the algorithm proceeds like the classical Hungarian algorithm. First, it computes an initial partial $\epsilon$-assignment (Section~\ref{sec-partial}), together with a pair of dual variables, such that the complementary slackness condition is satisfied (Eq.~\ref{eq-compslack}). Then, while maintaining this condition, the set of assigned elements of $\mathcal{V}$ (substituted or removed) is iteratively augmented by solving a restricted primal-dual problem. When all elements of $\mathcal{V}$ are assigned, some elements of $\mathcal{U}$ may not be assigned. So the algorithm proceeds until they are, as before. Since the complementary slackness condition is preserved through the iterations, the algorithm ends with an optimal solution (Prop.~\ref{th-compclack}). This is detailed in the following sections.
\subsection{Encoding and pre-processing}
In the following algorithms, a partial $\epsilon$-assignment is not
encoded by its associated binary matrix $\mathbf{X}\,{\in}\,\{0,1\}^{(n+1)\times(m+1)}$ (Section~\ref{sec-partial}) but by a pair of vectors $(\rho,\varrho)$ with $\rho\,{\in}\,\{0,1,\ldots,m\,{+}\,1\}^n$, $\varrho\,{\in}\,\{0,1,\ldots,n\,{+}\,1\}^m$, and such that
\begin{equation}
	\begin{array}{l}
		\rho_i=\left\{\begin{array}{ll}
			j&\text{if}~\exists j\,{\in}\,\{1,\ldots,m\,{+}\,1\}~|~x_{i,j}=1\\
			0&\text{else~~(unassigned)}
		\end{array}\right.\\
		\varrho_j=\left\{\begin{array}{ll}
			i&\text{if}~\exists i\,{\in}\,\{1,\ldots,n\,{+}\,1\}~|~x_{i,j}=1\\
			0&\text{else~~(unassigned)}
		\end{array}\right.
	\end{array}
\end{equation}
Note that unassigned elements of the partial $\epsilon$-assignment are assigned to $0$. The matrix $\mathbf{X}$ can be easily reconstructed from $\rho$ and $\varrho$.

We want to compute an initial partial $\epsilon$-assignment $\mathbf{X}$, and a pair of dual variables $(\mathbf{u},\mathbf{v})$, such that Eq.~\ref{eq-compslack} is satisfied. To this we adapt the basic pre-processing step involved in the Hungarian algorithm \cite{bur09}. Algorithm~\ref{algo-preproc} details this step. The minimum cost on each row $i\,{\in}\,\mathcal{U}$ of $\mathbf{C}$ allows to construct $\mathbf{u}$. Then, the minimum cost on each column $j\,{\in}\,\mathcal{V}$ of $\mathbf{C}\,{-}\,\mathbf{u}\mathbf{1}^T_{m+1}$ allows to construct $\mathbf{v}$. By construction we have $\mathbf{C}\geq\mathbf{u}\mathbf{1}^T_{m+1}\,{+}\,\mathbf{1}_{n+1}\mathbf{v}^T$. Let $\overline{\mathbf{C}}\,{=}\,\mathbf{C}\,{-}\,\mathbf{u}\mathbf{1}^T_{m+1}\,{-}\,\mathbf{1}_{n+1}\mathbf{v}^T$ be the transformed edit cost matrix (Proposition~\ref{theo-admuv}). This transformation guaranties the existence of at least one element $j\,{\in}\,\mathcal{V}$ satisfying $c_{i,j}\,{=}\,u_i\,{+}\,v_j$, or equivalently $\overline{c}_{i,j}\,{=}\,0$, for each $i\,{\in}\,\mathcal{U}$. In other terms, each row $i\,{\in}\,\mathcal{U}$ of $\overline{C}$ contains at least one $0$. Similarly, each column $j\,{\in}\,\mathcal{V}$ of $\overline{C}$ contains at least one $0$. So it is possible to deduce a partial $\epsilon$-assignment, eventually reduced to one pair of assigned elements only, satisfying Eq.~\ref{eq-compslack}. Algorithm~\ref{algo-preproc} finds such a set by a scanline approach. Note that the whole algorithm has a $O(nm)$ time complexity.

\begin{algorithm}[!t]
\Begin(PreProcessing${(\mathbf{C})}$){
	$u_{n+1},v_{m+1}\leftarrow 0$,~\,$\rho\leftarrow\mathbf{0}_n$,~\,$\varrho\leftarrow\mathbf{0}_m$\\
	\textbf{for}~\,$i\,{=}\,1,\ldots,n$\,~\textbf{do}~\,$u_i\leftarrow\min\,\{c_{i,j}\,|\,j\,{=}\,1,\ldots,m\,{+}\,1\}$\\
	\textbf{for}~\,$j\,{=}\,1,\ldots,m$\,~\textbf{do}~\,$v_j\leftarrow\min\,\{c_{i,j}\,{-}\,u_i\,|\,i\,{=}\,1,\ldots,n\,{+}\,1\}$\\
	\For{$i\,{=}\,1,\ldots,n$}{
		\For{$j\,{=}\,1,\ldots,m$}{
			\If{$(\varrho_j\,{=}\,0)\wedge(c_{i,j}\,{=}\,u_i\,{+}\,v_j)$}{
				$\rho_i\,{\leftarrow}\,j$,\,~$\varrho_j\,{\leftarrow}\,i$\quad{\color{dark-gray}// $i$ assigned to $j$}\\
				\textbf{break}
			}
		}
		\textbf{if}~\,$(\rho_i\,{=}\,0)\wedge(c_{i,m+1}\,{=}\,u_i)$\,~\textbf{then}~\,$\rho_i\,{\leftarrow}\,m\,{+}\,1$\quad{\color{dark-gray}// $i\,{\rightarrow}\,\epsilon$}
	}
	\For{$j\,{=}\,1,\ldots,m$}{
		\textbf{if}~\,$(\varrho_j\,{=}\,0)\wedge(c_{n+1,j}\,{=}\,v_j)$~\,\textbf{then}~\,$\varrho_j\,{\leftarrow}\,n\,{+}\,1$\quad{\color{dark-gray}// $\epsilon\,{\rightarrow}\,j$}
	}
	\Return{$((\rho,\varrho),(\mathbf{u},\mathbf{v}))$}
}
\caption{Basic pre-processing\label{algo-preproc}}
\end{algorithm}
Algorithm~\ref{algo-preproc} has a $O(nm)$ time complexity. More sophisticated pre-processing methods can be adapted to find an initial partial $\epsilon$-assignment, see \cite{bur09} for more details on these methods.
\begin{example}\label{ex-red}
Consider the edit cost matrix
\begin{equation}\begin{small}
	\mathbf{C}=\left(\begin{array}{ccccc|c}
		7 & 11 & 9 & 8 & 9 & 10\\
		2 & 8 & 8 & 5 & 7 & 3\\
		1 & 7 & 6 & 6 & 9 & 5\\
		3 & 7 & 6 & 2 & 2 & 3\\\cline{1-6}
		4 & 2 & 2 & 7 & 8 & 0
	\end{array}\right)
\end{small}\end{equation}
First part of Algorithm \ref{algo-preproc} computes
\begin{equation*}\begin{small}
\mathbf{u}=\left(\begin{array}{c}7\\2\\1\\2\\\cline{1-1} 0\end{array}\right)~\text{and}~\mathbf{v}=\left(\begin{array}{c}0\\2\\2\\0\\0\\\cline{1-1} 0\end{array}\right),~\text{so}~~\overline{\mathbf{C}}=\mathbf{C}-\mathbf{u}\mathbf{1}^T\,{-}\,\mathbf{1}\mathbf{v}^T=\left(\begin{array}{ccccc|c}
		0 & 2 & 0 & 1 & 2 & 3\\
		0 & 4 & 4 & 3 & 5 & 1\\
		0 & 4 & 3 & 5 & 8 & 4\\
		1 & 3 & 2 & 0 & 0 & 1\\\cline{1-6}
		4 & 0 & 0 & 7 & 8 & 0
\end{array}\right).
\end{small}\end{equation*}
Due to the distribution of the zeros in $\overline{\mathbf{C}}$, which correspond to the elements of $\mathbf{C}$ satisfying $c_{i,j}\,{=}\,u_i\,{+}\,v_j$, only partial assignments with edition can be constructed. Second part of Algorithm \ref{algo-preproc} leads to $\rho\,{=}\,(1,0,0,4)$ and $\varrho\,{=}\,(1,5,5,4,0)$, equivalently represented by
\begin{equation*}\begin{small}
	\overline{\mathbf{C}}=\left(\begin{array}{lllll|c}
		0^{*} & 2 & 0 & 1 & 2 & 3\\
		0 & 4 & 4 & 3 & 5 & 1\\
		0 & 4 & 3 & 5 & 8 & 4\\
		1 & 3 & 2 & 0^* & 0 & 1\\\cline{1-6}
		4 & 0^* & 0^* & 7 & 8 & 0
	\end{array}\right),
\end{small}\end{equation*}
where $\overline{c}_{i,j}\,{=}\,0^*$ corresponds to $i\,{\rightarrow}\,j$ in the partial $\epsilon$-assignment.
\end{example}
\subsection{Augmenting paths}
Consider a partial $\epsilon$-assignment $(\rho,\varrho)$ (or equivalently $\mathbf{X}$) and a pair of dual variables $(\mathbf{u},\mathbf{v})$ satisfying Eq.~\ref{eq-compslack}. The partial $\epsilon$-assignment may be empty (no pre-processing), \textit{i.e.} $\rho\,{=}\,\mathbf{0}_n$, $\varrho\,{=}\,\mathbf{0}_m$, $\mathbf{u}\,{=}\,\mathbf{0}_{n+1}$ and $\mathbf{v}\,{=}\,\mathbf{0}_{m+1}$. 
Assume that at least one element $k\,{\in}\,\mathcal{V}$ is not assigned to an element of $\mathcal{U}_\epsilon$, for instance $k\,{=}\,5\,{\in}\,\mathcal{V}$ in Example~\ref{ex-red}. We want to find a new partial $\epsilon$-assignment such that:
\begin{enumerate}
\item $k$ becomes assigned to an element of $\mathcal{U}_\epsilon$,
\item all previously assigned elements are still assigned, and
\item there is a pair of dual variables such that Eq.~\ref{eq-compslack} is satisfied.
\end{enumerate}
Let $\mathcal{S}_{n,m,\epsilon,k}$ be the set of all matrices representing a partial $\epsilon$-assignment which satisfies constraints 1 and 2 above.

To solve this sub-problem of the LSAPE, we consider a subgraph of the mixed bipartite graph $K_{n,m,\epsilon}$ (Section~\ref{sec-linprog}) whose edges and arcs satisfy $c_{i,j}\,{=}\,u_i\,{+}\,v_j$, \textit{i.e.} $\overline{c}_{i,j}\,{=}\,0$. Let $G^0$ be this subgraph. By definition, the bipartite graph which represents the current partial $\epsilon$-assignment $\mathbf{X}$ is a subgraph of $G^0$ (its node adjacency matrix is $\mathbf{X}$). It defines a set of assigned nodes and edges on $G^0$, an edge $(i,j)$ being assigned iff $x_{i,j}\,{=}\,1$ in the current partial $\epsilon$-assignment.

The new partial $\epsilon$-assignment is constructed by computing minimal paths in $K_{n,m,\epsilon}$, \textit{i.e.} paths in $G^0$. By definition, a path $P$ in $K_{n,m,\epsilon}$ alternate an element of a set and an element of the other. Also, since there is no arc from $\epsilon$ to any other element of $\mathcal{U}\,{\cup}\,\mathcal{V}$ in $K_{n,m,\epsilon}$ (Section~\ref{sec-linprog}), $\epsilon$ can only be a terminal node of a path. The length of a path, penalized by the edit costs, is defined by
\begin{equation}\label{eq-lngth}
	\gamma_\mathbf{C}(P)=\sum_{(i,j)\subset P}c_{i,j}.
\end{equation}
\begin{definition}[alternating and augmenting paths]
	A path in $K_{n,m,\epsilon}$ is alternating if its edges/arcs are alternatively unassigned and assigned. An alternating path, which starts with an unassigned element $k$ of $\mathcal{V}$ (resp. $\mathcal{U}$) and ends with
\begin{enumerate}
\item an unassigned element of $\mathcal{U}$ (resp. $\mathcal{V}$), or
\item the null element $\epsilon$ of $\mathcal{U}_\epsilon$ or $\mathcal{V}_\epsilon$,
\end{enumerate}
is an augmenting path if its length defined by Eq.~\ref{eq-lngth} is minimal among all the alternating paths connecting $k$ to an element satisfying one of the two cases above.
\end{definition}
\noindent
Contrary to the classical definition of augmenting paths, see \cite{bur09}, an augmenting path as defined above may end with an assigned arc. Assume that the starting node of the path is $k\,{\in}\,\mathcal{V}$. Then, these two configurations can be encountered:
\begin{figure}[!t]
\begin{center}
\begin{tabular}{ccc}
\includegraphics[scale=0.5]{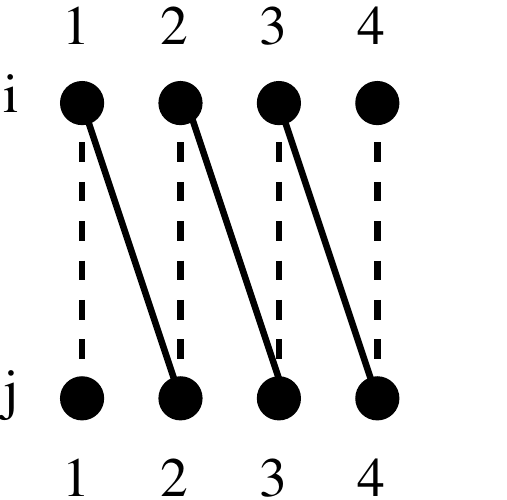}&\includegraphics[scale=0.5]{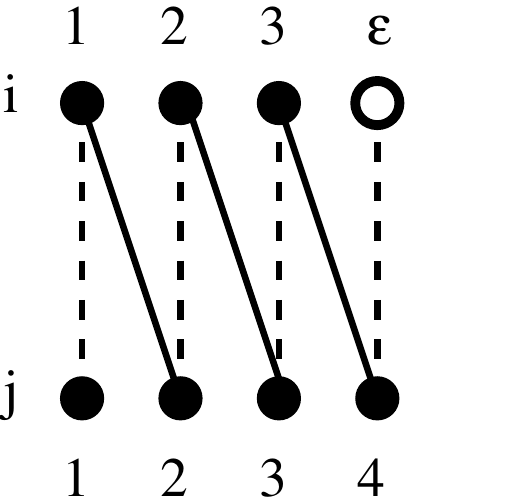}
&\includegraphics[scale=0.5]{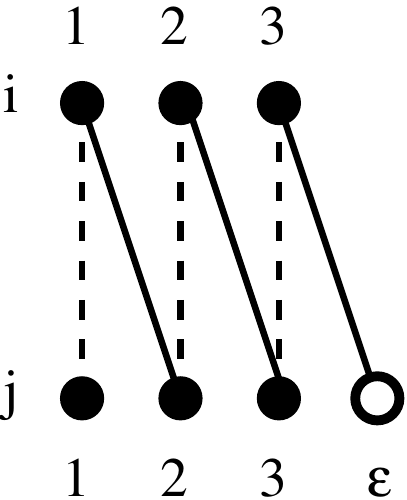}\\
(a) last edge is unassigned&(b) last edge is unassigned&(c) last edge is assigned\\
only substitutions&ending with a removal&ending with an insertion
\end{tabular}
\end{center}
\caption{\label{fig-augpath}Three cases of augmenting paths starting with an unassigned edge $(i_1,j_1)$. Unassigned edges are dotted.}
\end{figure}
\begin{enumerate}
\item[$\bullet$]If $i_l\,{\in}\,\mathcal{U}_\epsilon$ is the last node, the path can be represented as a sequence $(k\,{=}\,j_1,i_1,j_2,i_2,\ldots,j_l,i_l)$ with $j_s\,{\in}\,\mathcal{V}$ and $i_r\,{\in}\,\mathcal{U}$ for $r\,{=}\,1,\ldots,l\,{-}\,1$. Since the first edge $(j_1,i_1)$ is unassigned and since the path is alternating, the last edge $(j_l,i_l)$ is also unassigned (Fig.~\ref{fig-augpath}(a,b)), as in the classical case.
\item[$\bullet$]If $\epsilon\,{\in}\,\mathcal{V}_\epsilon$ is the last node, the path can be represented as a sequence $(k\,{=}\,j_1,i_1,j_2,i_2,\ldots,j_l,i_l,\epsilon)$ with $j_s\,{\in}\,\mathcal{V}$ and $i_r\,{\in}\,\mathcal{U}$. In this case the last arc $(i_l,\epsilon)$ is assigned, \textit{i.e.} $i_l$ is assigned to $\epsilon$ (Fig.~\ref{fig-augpath}(c)).
\end{enumerate}
Once an augmenting path has been found, a new partial $\epsilon$-assignment is constructed by assigning unassigned edges/arcs of the path, and reciprocally, as follows.
\begin{proposition}\label{prop-subprob}
	Let $\mathbf{X}$ be the matrix representation of a partial $\epsilon$-assignment, and $(\mathbf{u},\mathbf{v})$ be a pair of dual variables ($u_{n+1}\,{=}\,v_{m+1}\,{=}\,0$), such that the complementary slackness condition (Eq.~\ref{eq-compslack}) is satisfied. Let $P$ be an augmenting path starting with an unassigned element $k\,{\in}\,\mathcal{V}$ (resp. $\mathcal{U}$). The matrix $\mathbf{X}^\prime$, defined by
	$$
		x^\prime_{i,j}=\begin{cases}
			1-x_{i,j} & \text{if}~(i,j)\,{\subset}\,P\\
			x_{i,j}   & \text{else}
		\end{cases},
		$$
represents an optimal partial $\epsilon$-assignment, \textit{i.e.} satisfying
\begin{equation}\label{eq-subprob}
  \mathbf{X}^\prime\in\argmin\left\{A_\epsilon(\mathbf{Y},\mathbf{C})\,|\,\mathbf{Y}\,{\in}\,\mathcal{S}_{n,m,\epsilon,k}\right\},
\end{equation}
Its total cost is equal to $A_\epsilon(\mathbf{X}^\prime,\mathbf{C})=\gamma_{\overline{\mathbf{C}}}(P)+E(\mathbf{u},\mathbf{v})$, with $\overline{\mathbf{C}}\,{=}\,\mathbf{C}\,{-}\,\mathbf{u}\mathbf{1}_{m+1}^T\,{-}\,\mathbf{1}_{n+1}\mathbf{v}^T$.
\end{proposition}
\begin{proof}
	Let $P$ be an augmenting path starting with an element $k\,{\in}\,\mathcal{V}$. We have $\mathbf{X}^\prime\,{\in}\,\{0,1\}^{(n+1)\times(m+1)}$ by definition. We show that any column $j\,{\in}\,\mathcal{V}$ of $\mathbf{X}^\prime$ sums to no more than $1$. First remark that for any $j$ included in $P$ we have:
$$
	\sum_{i=1}^{n+1}x^\prime_{i,j}=\sum_{(i,j)\subset P}x^\prime_{i,j}+\sum_{(i,j)\not\subset P}x^\prime_{i,j}=\sum_{(i,j)\subset P}(1-x_{i,j})+\sum_{(i,j)\not\subset P}x_{i,j}=\sum_{(i,j)\subset P}(1-x_{i,j})
$$
since there is at most one $i\,{\in}\,\mathcal{U}_\epsilon$ such that $x_{i,j}\,{=}\,1$, and this element is included in $P$. Moreover we have:
\begin{enumerate}
	\item[$\bullet$]For $j\,{=}\,k$, there is no $i\,{\in}\,\mathcal{U}_\epsilon$ such that $x_{i,k}\,{=}\,1$, then the neighbor $i_1$ of $k$ in $P$ satisfies $x_{i_1,k}\,{=}\,0$. So $x_{i_1,k}^\prime\,{=}\,1$ and $x_{i,k}^\prime\,{=}\,0$ for all $i\,{\in}\,\mathcal{U}_\epsilon\,{\setminus}\,\{i_1\}$.
	\item[$\bullet$]For any $j\,{\in}\,\mathcal{V}\,{\setminus}\,\{k\}$ included in $P$, $j$ has two neighbors $i_a$ and $i_b$ of $\mathcal{U}_\epsilon$ in $P$ satisfying for instance $x_{i_a,j}\,{=}\,1$ and $x_{i_b,j}\,{=}\,0$. So we have $x^\prime_{i_a,j}\,{=}\,0$, $x^\prime_{i_b,j}\,{=}\,1$, and $x^\prime_{i,j}\,{=}\,0$ for all $i\,{\in}\,\mathcal{U}_\epsilon\,{\setminus}\,\{i_a,i_b\}$.
\end{enumerate}
So $\sum_{i=1}^{n+1}x^\prime_{i,j}\in\{0,1\}$ for any $j$ included in $P$. When $j$ is not a node of $P$, we have $\sum_{i=1}^{n+1}x^\prime_{i,j}=\sum_{i=1}^{n+1}x_{i,j}\in\{0,1\}$ by definition.

Similarly, we have $\sum_{j=1}^{m+1}x^\prime_{i,j}\in\{0,1\}$ and so $\mathbf{X}^\prime$ is a partial $\epsilon$-assignment. This shows that any augmenting path starting with an unassigned element $k\,{\in}\,\mathcal{V}$ leads to a new partial $\epsilon$-assignment that assigns this element. Now we show that this partial $\epsilon$-assignment is optimal among all partial $\epsilon$-assignments in $\mathcal{S}_{n,m,\epsilon,k}$.

Since $\mathbf{X}$ and $\overline{\mathbf{C}}$ satisfy the complementary slackness condition, the total cost $A_\epsilon(\mathbf{X},\overline{\mathbf{C}})$ can be decomposed into $4$ terms:
\begin{equation*}
	A_\epsilon(\mathbf{X},\overline{\mathbf{C}})=\underset{=0~~(\overline{c}_{i,j}=0)}{\underbrace{\sum_{\stackrel{(i,j)\subset P}{x_{i,j}=1}}\overline{c}_{i,j}x_{i,j}}}+\underset{=0~~(\overline{c}_{i,j}\geq 0)}{\underbrace{\sum_{\stackrel{(i,j)\subset P}{x_{i,j}=0}}\overline{c}_{i,j}x_{i,j}}}+\underset{=0~~(\overline{c}_{i,j}=0)}{\underbrace{\sum_{\stackrel{(i,j)\not\subset P}{x_{i,j}=1}}\overline{c}_{i,j}x_{i,j}}}+\underset{=0~~(\overline{c}_{i,j}\geq 0)}{\underbrace{\sum_{\stackrel{(i,j)\not\subset P}{x_{i,j}=0}}\overline{c}_{i,j}x_{i,j}}}
\end{equation*}
Similarly we have for $\mathbf{X}^\prime$:
\begin{equation*}
	\begin{aligned}
			A_\epsilon(\mathbf{X}^\prime,\overline{\mathbf{C}})&=\underset{=0~~(\overline{c}_{i,j}=0)}{\underbrace{\sum_{\stackrel{(i,j)\subset P}{x_{i,j}=1}}\overline{c}_{i,j}(1-x_{i,j})}}+\underset{\geq 0~~(\overline{c}_{i,j}\geq 0)}{\underbrace{\sum_{\stackrel{(i,j)\subset P}{x_{i,j}=0}}\overline{c}_{i,j}(1-x_{i,j})}}+\underset{=0~~(\overline{c}_{i,j}=0)}{\underbrace{\sum_{\stackrel{(i,j)\not\subset P}{x_{i,j}=1}}\overline{c}_{i,j}x_{i,j}}}+\underset{=0~~(\overline{c}_{i,j}\geq 0)}{\underbrace{\sum_{\stackrel{(i,j)\not\subset P}{x_{i,j}=0}}\overline{c}_{i,j}x_{i,j}}}\\
			&=0+\underset{\geq 0~(\overline{c}_{i,j}\geq 0)}{\underbrace{\sum_{\stackrel{(i,j)\subset P}{x_{i,j}=0}}\overline{c}_{i,j}}}+0+0=\gamma_{\overline{\mathbf{C}}}(P)
	\end{aligned}
\end{equation*}
Since $P$ is a minimal path among all augmenting paths connecting $k$ to an unassigned element of $\mathcal{U}$ or to $\epsilon$, $\mathbf{X}^\prime$ is a partial $\epsilon$-assignment with a minimal cost.
\end{proof}
\noindent
If each edge/arc of $P$ satisfies $\overline{c}_{i,j}\,{=}\,0$ ($P$ is thus contained in $G^0$), then $\gamma_{\overline{\mathbf{C}}}(P)\,{=}\,0$ and by consequence $A_\epsilon(\mathbf{X}^\prime,\mathbf{C})\,{=}\,E(\mathbf{u},\mathbf{v})$. The sub-problem (Eq.~\ref{eq-subprob}) of the LSAPE can thus be solved by computing an augmenting path $P$, together with a transformed edit cost matrix $\overline{\mathbf{C}}$, such that $\gamma_{\overline{\mathbf{C}}}(P)\,{=}\,0$. This is detailed in Section~\ref{sec-augalg}. The new partial $\epsilon$-assignment is then obtained by assigning unassigned edges and by deassigning assigned edges along the path (Proposition~\ref{prop-subprob}), as detailed in Section~\ref{sec-mainalg}.
\begin{example}\label{ex-augpath1}
Consider the partial $\epsilon$-assignment defined by $\rho\,{=}\,(1,0,0,4)$ and $\varrho\,{=}\,(1,5,5,4,0)$ for the edit cost matrix
\begin{equation*}\begin{small}
	\left(\begin{array}{lllll|c}
		\mathbf{0^{*}} & {2} & {0} & 0 & 1 & {3}\\
		0 & 4 & 4 & 2 & 4 & 1\\
		0 & 4 & 3 & 4 & 7 & 4\\
		2 & 4 & 3 & \mathbf{0^*} & 0 & 2\\\cline{1-6}
		4 & \mathbf{0^*} & \mathbf{0^*} & 6 & 7 & 0
	\end{array}\right),~~~~~G^0=\vcenter{\hbox{\includegraphics[scale=0.44]{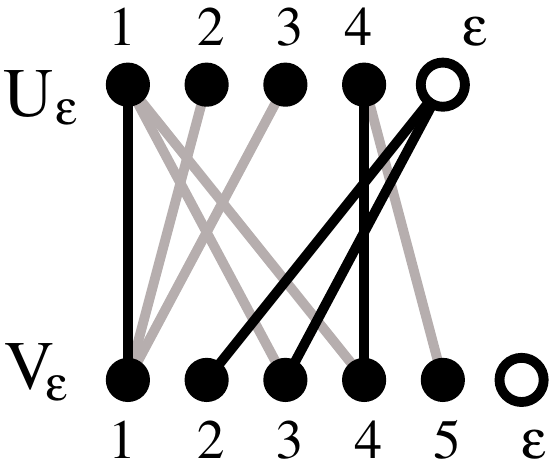}}}\end{small}
\end{equation*}
We want to assign $k\,{=}\,5\,{\in}\,\mathcal{V}$. In the corresponding bipartite graph $G^0$ above (gray and black edges), $k\,{=}\,5$ can be linked to an unassigned element of $\mathcal{U}$ ($2$ or $3$), by a path which alternates unassigned (gray) and assigned (black) edges, \textit{i.e.} an augmenting path. Consider the one ending with $2\,{\in}\,\mathcal{U}$:
\begin{equation*}
\vcenter{\hbox{\includegraphics[scale=0.44]{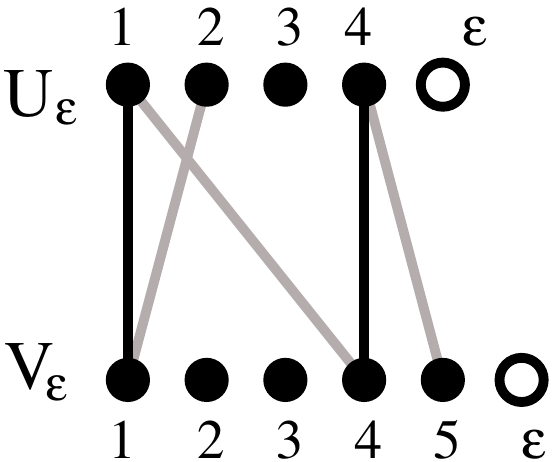}}}~\overset{\text{swap}}{\Longrightarrow}~\vcenter{\hbox{\includegraphics[scale=0.44]{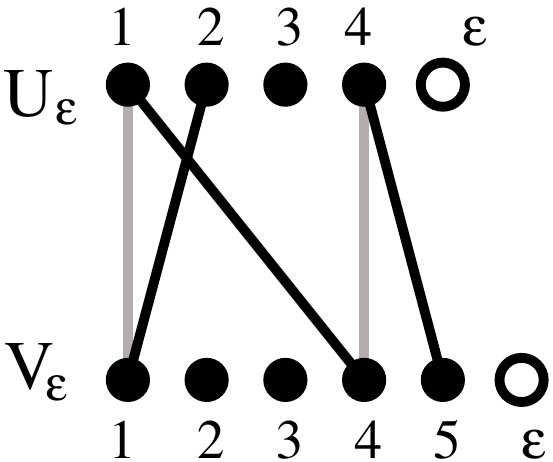}}}
\end{equation*}
By swapping assigned and unassigned edges along the path:
\begin{equation*}\begin{small}
	G^0=\vcenter{\hbox{\includegraphics[scale=0.44]{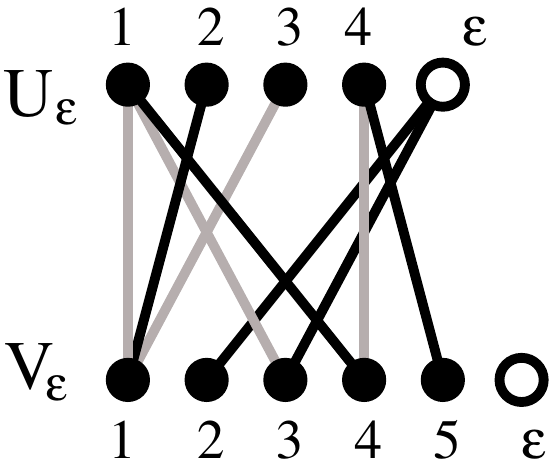}}}~~~~~\left(\begin{array}{lllll|c}
		\mathbf{0} & 2 & 0 & \mathbf{0^*} & 1 & 3\\
		\mathbf{0^*} & 4 & 4 & 2 & 4 & 1\\
		0 & 4 & 3 & 4 & 7 & 4\\
		2 & 4 & 3 & \mathbf{0} & \mathbf{0^*} & 2\\\cline{1-6}
		4 & 0^* & 0^* & 6 & 7 & 0
	\end{array}\right)
\end{small}\end{equation*}
we obtain the partial $\epsilon$-assignment $\rho\,{=}\,(4,1,0,5)$ and $\varrho\,{=}\,(2,5,5,1,4)$.
\end{example}
\noindent
The augmenting path can also ends with an element $\epsilon$, either $(n\,{+}\,1)\,{\in}\,\mathcal{U}_\epsilon$ or $(m\,{+}\,1)\,{\in}\,\mathcal{V}_\epsilon$. In one of these cases and contrary to the classical definition of augmenting paths \cite{bur09}, an augmenting path can end with an assigned arc, as illustrated in the following example. 
\begin{example}
Consider the partial $\epsilon$-assignment defined by $\rho\,{=}\,(4,1,2)$ and $\varrho\,{=}\,(2,3,0,1)$ for the edit cost matrix:
\begin{equation*}\begin{small}
	\left(\begin{array}{llll|c}
		0 & 0 & 3 & \mathbf{0^*} & 4\\
		\mathbf{0^*} & 2 & 0 & 2 & 7\\
		0 & \mathbf{0^*} & 2 & 3 & 6\\\cline{1-5}
		4 & 6 & 8 & 0 & 0
	\end{array}\right),~~~~G^0=\hspace{-0.18cm}\vcenter{\hbox{\includegraphics[scale=0.44]{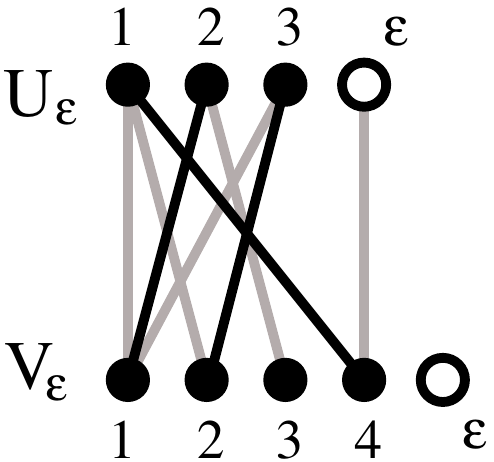}}}
	\end{small}\end{equation*}
The path starting with $3\,{\in}\,\mathcal{V}$ and ending with $(n\,{+}\,1)\,{\in}\,\mathcal{U}_\epsilon$ leads to $\rho=(1,2,3)$ and $\varrho=(1,2,3,4)$: 
\begin{equation*}\begin{small}
	\vcenter{\hbox{\includegraphics[scale=0.44]{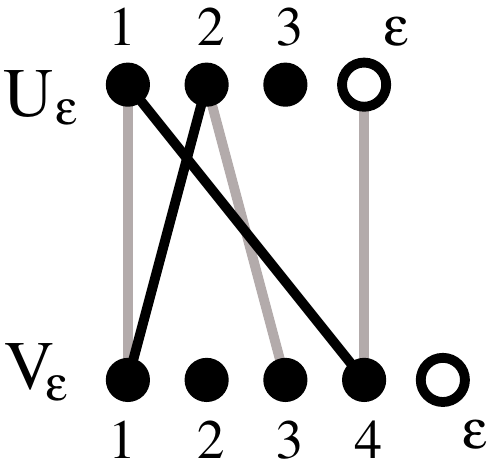}}}\overset{\text{swap}}{\Rightarrow}\vcenter{\hbox{\includegraphics[scale=0.44]{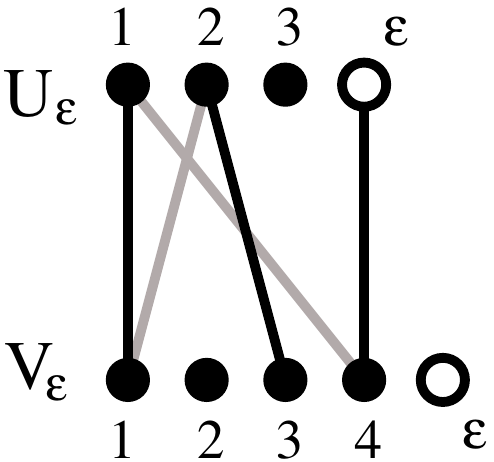}}} \Rightarrow \vcenter{\hbox{\includegraphics[scale=0.44]{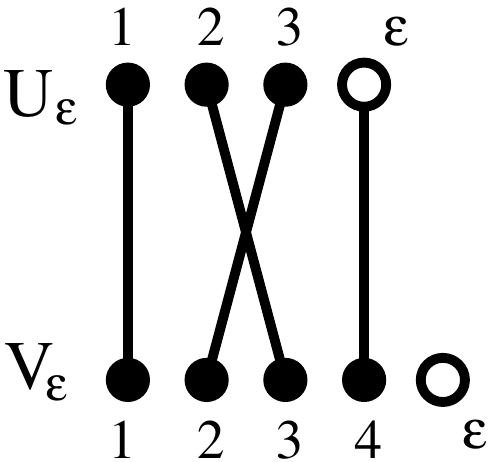}}}~~~~\left(\begin{array}{llll|c}
		\mathbf{0^*} & 0 & 3 & 0 & 4\\
		0 & 2 & \mathbf{0^*} & 2 & 7\\
		0 & \mathbf{0^*} & 2 & 3 & 6\\\cline{1-5}
		4 & 6 & 8 & \mathbf{0^*} & 0
	\end{array}\right)
\end{small}\end{equation*}
For the last case, consider the partial $\epsilon$-assignment defined by $\rho\,{=}\,(3,4,2)$ and $\varrho\,{=}\,(0,3,1)$. We want to assign $1\,{\in}\,\mathcal{V}$:
\begin{equation*}\begin{small}
	\left(\begin{array}{lll|l}
		2 & 3 & \mathbf{0^*}  & 4\\
		7 & 0 & 5  & \mathbf{0^*}\\
		0 & \mathbf{0^*} & 4  & 6\\\cline{1-4}
		4 & 6 & 0 & 0
	\end{array}\right),~~~~G^0=\hspace{-0.18cm}\vcenter{\hbox{\includegraphics[scale=0.44]{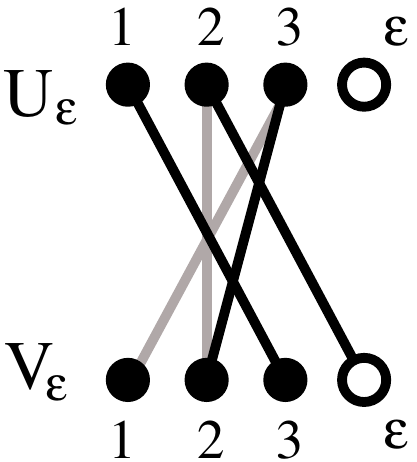}}}
\end{small}\end{equation*}
This leads to $\rho=(3,2,1)$ and $\varrho=(3,2,1)$:
\begin{equation*}\begin{small}
	\text{path:}\hspace{-0.18cm}\vcenter{\hbox{\includegraphics[scale=0.44]{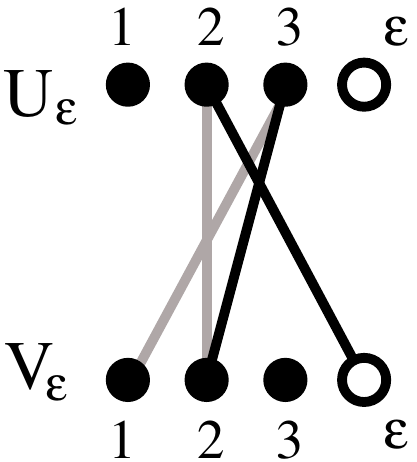}}}\overset{\text{swap}}{\Rightarrow}\vcenter{\hbox{\includegraphics[scale=0.44]{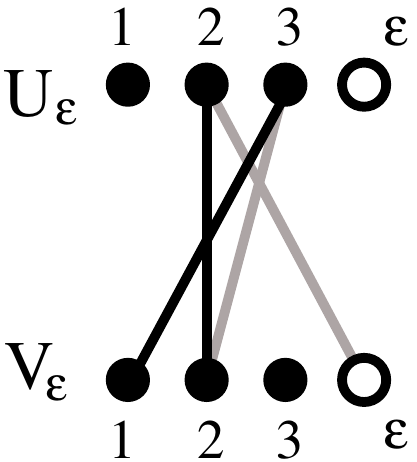}}}\Rightarrow\vcenter{\hbox{\includegraphics[scale=0.44]{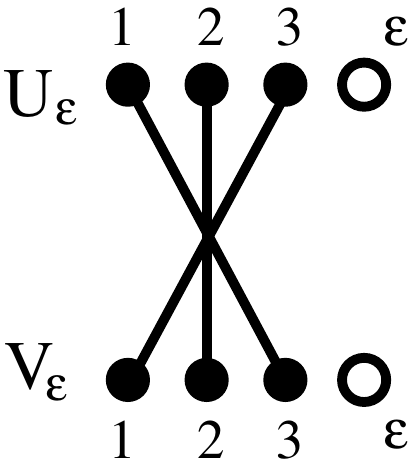}}}~~~~\left(\begin{array}{lll|l}
		2 & 3 & \mathbf{0^*}  & 4\\
		7 & \mathbf{0^*} & 5  & 0\\
		\mathbf{0^*} & 0 & 4  & 6\\\cline{1-4}
		4 & 6 & 0 & 0
	\end{array}\right)
\end{small}\end{equation*}
Remark that the last edge of the path switches from
assigned to unassigned.
\end{example}
\subsection{Construction of an augmenting path}\label{sec-augalg}
\begin{algorithm}[!t]
\Begin(\,Augment${(k,\mathbf{C},\rho,\varrho,\mathbf{u},\mathbf{v},U)}$){
$n\,{\leftarrow}\,|U|$,~\,$\pi\,{\in}\,[0,{+\infty}]^n$,~\,${SV},\, {SU},\, {LU}\leftarrow\emptyset$,~\,$j\leftarrow k$\\
\textbf{foreach}~$i\,{\in}\,U$~\textbf{do}~$\pi_i\leftarrow{+\infty}$\\
\While{true}{
	${SV}\leftarrow{SV}\,{\cup}\,\{j\}$\\
	\If{$(\varrho_j\,{\leq}\,n)\wedge(c_{n+1,j}\,{-}\,v_j\,{=}\,0)$}{
		\textbf{return}~$(n\,{+}\,1,j,\mathbf{u},\mathbf{v},\textbf{pred})$
	}
	{\color{dark-gray}// \textit{find candidate nodes and update the tree}}\\
	\ForEach{$i\,{\in}\,U\,{\setminus}\,{LU}$}{
		\If{$c_{i,j}-u_i-v_j<\pi_i$}{
			${pred}_i\leftarrow j$\\
			$\pi_i\leftarrow c_{i,j}-u_i-v_j$\\
			\If{$\pi_i=0$}{
				\textbf{if}~$\rho_i\,{\in}\,\{0,m\,{+}\,1\}$~\textbf{return}~$(i,0,\mathbf{u},\mathbf{v},\textbf{pred})$\\
				${LU}\leftarrow{LU}\,{\cup}\,\{i\}$
			}
		}
	}
	{\color{dark-gray}// \textit{update dual variables and find candidate nodes}}\\
	\If{${LU}\,{\setminus}\,{SU}\,{=}\,\emptyset$}{
		$\delta_s\leftarrow\min\{\pi_i\,|\,i\,{\in}\,U\,{\setminus}\,{LU}\}$\\
		$(l,\delta_\epsilon)\leftarrow(\argmin,\min)\{c_{n+1,j}\,{-}\,v_j\,|\,j\,{\in}\,{SV}\}$\\
		$\delta\leftarrow\min\{\delta_s,\delta_\epsilon\}$\\
		\textbf{foreach}~$j\,{\in}\,{SV}$~\textbf{do}~$v_j\leftarrow v_j+\delta$\\
		\textbf{foreach}~$i\,{\in}\,{LU}$~\textbf{do}~$u_i\leftarrow u_i-\delta$\\
		\textbf{if}~$\delta_\epsilon\,{\leq}\,\delta_s$~\textbf{return}~$(n\,{+}\,1,l,\mathbf{u},\mathbf{v},\textbf{pred})$\\
		\ForEach{$i\,{\in}\,U\,{\setminus}\,{LU}$}{
			$\pi_i\leftarrow\pi_i-\delta$\\
			\If{$\pi_i=0$}{
				\textbf{if}~$\rho_i\,{\in}\,\{0,m\,{+}\,1\}$~\textbf{return}~$(i,0,\mathbf{u},\mathbf{v},\textbf{pred})$\\
				${LU}\leftarrow{LU}\,{\cup}\,\{i\}$
			}
		}
	}
	{\color{dark-gray}// \textit{extend the tree}}\\	
	$i\leftarrow$~\text{any element in}~${LU}\,{\setminus}\,{SU}$\\
	${SU}\leftarrow{SU}\,{\cup}\,\{i\}$\\
	$j\leftarrow\rho_i$
}}
\caption{Construct an augmenting path starting at $k\,{\in}\,\mathcal{V}$.\label{algo-augment}}
\end{algorithm}
Algorithm~\ref{algo-augment} details the construction of an augmenting
path from an unassigned element $k\,{\in}\,\mathcal{V}$. It consists
in growing a tree $T^0$ of minimal alternating paths in $G^0$, rooted
in $k$, until an augmenting path has been found. This is similar to
Dijkstra's algorithm to compute minimal paths in bipartite graphs. The
tree is constructed by growing two sets, \textit{SV} (initialized to
$k$) and \textit{SU}, representing the current nodes of the
tree. Candidate nodes of $\mathcal{U}$ to the extension are
represented by the set $\textit{LU}\,{\setminus}\,{SU}$, where the set
\textit{LU} defines the nodes of $\mathcal{U}$ within the tree, as well
as all candidate nodes. For computational purposes, the more general
tree $T$ of minimal alternating paths rooted in $k$ (not necessary in
$G^0$) is constructed. This tree is encoded by a predecessor vector
$\mathbf{pred}\,{\in}\,\mathcal{V}^n$ (for unassigned edges) and the
pair $(\rho,\varrho)$ encoding the partial $\epsilon$-assignment (for
assigned edges). The vector $\pi\,{\in}\,[0,{+\infty})^n$ encodes, for
each $i$, the minimal transformed cost $\overline{c}_{i,j}$ among all
nodes $j\,{\in}\,\textit{SV}$. The tree $T^0\,{\subset}\,G^0$ is
obviously included in the tree $T$. Note that $\pi$ and \textbf{pred} do not consider the element $\epsilon$. Indeed, since it can only be a sink of an augmenting path, the algorithm stops when it is encountered (see below).

At each iteration, the last node $j\,{\in}\,\textit{SV}$ inserted to the tree $T^0$ is considered (line 5). Each node $i\,{\in}\,\mathcal{U}\,{\setminus}\,\textit{LU}$ (line 9) is added to \textit{LU} as a candidate node to the extension of $T^0$ (line~15) if it is connected to $j$ in $G^0$ (line 13), \textit{i.e.} $\overline{c}_{i,j}\,{=}\,0$. 
Remark that $T$ is constructed simultaneously (lines 10--12). When there is at least one candidate node ($\textit{LU}\,{\setminus}\,\textit{SU}\,{\not=}\,\emptyset$), one is selected and added to the set \textit{SU} of nodes of $T^0$ (lines 30 and 31), \textit{i.e.} the unassigned edge $(i,\text{pred}_i)$ is candidate for an augmenting path, together with the assigned node $\rho_i$ and edge $(i,\rho_i)$. Then $\rho_i$ is considered for the next iteration (lines 32 and 5). The tree stops to grow when an unassigned node of $\mathcal{U}$ is reached (line 14), when $\epsilon$ is reached (line 14 or 6), or when there is no candidate node in ${LU}$, \textit{i.e.} ${LU}\,{\setminus}\,{SU}\,{=}\,\emptyset$ (line 16). In this case, dual variables $\mathbf{u}$ and $\mathbf{v}$ are updated such that at least one new unassigned edge is inserted in $G^0$ (lines 18--22), \textit{i.e.} $\overline{c}_{i,j}\,{=}\,0$ is satisfied for at least one pair $(i,j)\,{\in}\,((\mathcal{U}\,{\setminus}\,{LU})\,{\cup}\,\{n\,{+}\,1\})\,{\times}\,{SV}$. Such nodes $i\,{\in}\,\mathcal{U}\,{\setminus}\,{LU}$ become candidates (lines 26 and 28), as before. By construction, $((\mathcal{U}\,{\setminus}\,\textit{LU})\,{\cup}\,\{n\,{+}\,1\})\times\textit{SV}$ does not contain any edge of $G^0$ (including assigned edges).
\begin{proposition}\label{theo-IJ}
	Let $(\mathbf{X},(\mathbf{u},\mathbf{v}))$ and $\mathbf{C}$ be a partial $\epsilon$-assignment, a pair of dual variables and an edit cost matrix such that Eq.~\ref{eq-compslack} is satisfied. Consider the minimum transformed cost $$\delta\,{=}\,\min\{\overline{c}_{i,j}~|~i\,{\not\in}\,\textit{LU},~j\,{\in}\,\textit{SV}\}$$
	with $\overline{\mathbf{C}}\,{=}\,\mathbf{C}\,{-}\,\mathbf{u}\mathbf{1}^T_{m+1}\,{-}\,\mathbf{1}_{n+1}\mathbf{v}^T$ and $u_{n+1}\,{=}\,v_{m+1}\,{=}\,0$.
	Eq.~\ref{eq-compslack} is still satisfied if the dual variables $(\mathbf{u},\mathbf{v})$ are updated by 
	\begin{equation}\label{eq-updual}
	\left\{\begin{array}{ll}
	u_i\,{\leftarrow}\,u_i\,{-}\,\delta,&\forall i\,{\in}\,\textit{LU}\\
	v_j\,{\leftarrow}\,v_j\,{+}\,\delta,&\forall j\,{\in}\,\textit{SV}
	\end{array}\right.
	\end{equation}
or equivalently if the transformed edit cost matrix $\overline{\mathbf{C}}$ is updated by
\begin{equation}\label{eq-assadm2}
	\overline{c}_{i,j}\leftarrow\left\lbrace
	\begin{array}{ll}
		\overline{c}_{i,j}-\delta &\text{if}~i\not\in\textit{LU},~j\in\textit{SV}\\
		\overline{c}_{i,j}+\delta &\text{if}~i\in\textit{LU},~j\not\in\textit{SV}\\
		\overline{c}_{i,j}&\text{else}
	\end{array}\right.
\end{equation}
This updating augments the graph $G^0$ with at least one new edge in $((\mathcal{U}\,{\setminus}\,\textit{LU})\,{\cup}\,\{n\,{+}\,1\})\times\textit{SV}$.
\end{proposition}
\begin{proof}
The updating of the dual variables given by Eq.~\ref{eq-updual} leads to the following cases:
\begin{enumerate}
\item When $i\,{\not\in}\,\textit{LU}$ and
  $j\,{\not\in}\,\textit{SV}$, no updating occurs, and so the
  complementary slackness condition(CSC) is maintained in this case.
\item When $i\,{\in}\,\textit{LU}$ and $j\,{\in}\,\textit{SV}$, we have $c_{i,j}-(u_i-\delta)-(v_j+\delta)=c_{i,j}-u_i-v_j=\overline{c}_{i,j}$. So no updating occurs and as before the CSC is preserved in this case.
\item When $i\,{\in}\,\textit{LU}$ and $j\,{\not\in}\,\textit{SV}$, $\overline{c}_{i,j}$ is updated by $c_{i,j}-(u_i-\delta)-v_j=\overline{c}_{i,j}+\delta$. So the transformed cost does not become negative. Since elements $i\,{\in}\,\textit{LU}$ can only be assigned to elements $j\,{\in}\,\textit{SV}$ ($x_{i,j}\,{=}\,1$), we have $x_{i,j}\,{=}\,0$ for all $j\,{\not\in}\,\textit{SV}$, and by consequence the CSC is preserved.
\item When $i\,{\not\in}\,\textit{LU}$ and $j\,{\in}\,\textit{SV}$, $\overline{c}_{i,j}$ is updated by $c_{i,j}-u_i-(v_j+\delta)=\overline{c}_{i,j}-\delta$ which is non-negative by definition. Since $\delta$ is the minimum of $((\mathcal{U}\,{\setminus}\,\textit{LU})\,{\cup}\,\{n\,{+}\,1\})\times\textit{SV}$, the updating produces at least one pair $(i,j)$ such that $\overline{c}_{i,j}\,{=}\,0$ (a new edge added to $G^0$). Also, $\delta\,{>}\,0$ induces that $\overline{c}_{i,j}\,{\not=}\,0$ and $x_{i,j}\,{=}\,0$ before the updating (by complementary slackness). By consequence the CSC is also preserved in this case.
\end{enumerate}
So transformed costs do not become negative and the CSC is preserved on $\mathcal{U}_\epsilon\,{\times}\,\mathcal{V}_\epsilon$. Moreover, this shows that the graph $G^0$ is augmented by at least one new edge, ensuring the dual updating to produce an enlarged tree.
\end{proof}
\noindent
This shows that the dual update adds at least one edge to $G^0$, guaranteeing to complete the tree at each iteration (lines 30, 31 and 5). This also implies that Eq.~\ref{eq-compslack} is satisfied at each iteration. Since each node of $\mathcal{U}$ can be reached, and since each node of the tree $T^0$ is inserted only once, an augmenting path is found if there is an unassigned node of $\mathcal{U}$. Since nodes $\epsilon$ are unconstrained, an augmenting path is also found whenever they are encountered. Even when all elements of $\mathcal{U}$ are already assigned, it is possible to find an augmenting path, ending by a node $\epsilon$. Since the length $\gamma_{\overline{\mathbf{C}}}\,{=}\,0$, by Proposition~\ref{prop-subprob} the corresponding matrix $\mathbf{X}^\prime$ solves the sub-problem of the LSAPE.

The growing of the tree depends on the encoding of the set $\textit{LU}\,{\setminus}\,\textit{SU}$ and the selection of the next node of $\mathcal{U}$ (line 30). We propose to use a FIFO strategy, leading to a breadth-first like growth of the tree. It can be efficiently encoded, together with the sets $\textit{LU}$, $U\,{\setminus}\,\textit{LU}$ and $\textit{SU}$, by a permutation of $\mathcal{U}$ (Fig.~\ref{fig-array} and Appendix~\ref{app-st}), since $\mathcal{U}\,{=}\,\textit{SU}\,{\cup}\,(\textit{LU}\,{\setminus}\,\textit{SU})\,{\cup}\,(\mathcal{U}\,{\setminus}\,\textit{LU})$. 
\begin{example}
Consider the partial $\epsilon$-assignment given by $\rho\,{=}\,(3,0,1,4,0)$ and $\varrho=(3,6,1,4,0,0)$, and the following cost matrix with $\mathbf{u}\,{=}\,\mathbf{0}_5$ and $\mathbf{v}\,{=}\,\mathbf{0}_6$. We want to assign $k\,{=}\,5\,{\in}\,\mathcal{V}$, which is the root of the tree $T$. So in the first iteration of Algorithm~\ref{algo-augment}, the element $j\,{=}\,5\,{\in}\,\mathcal{V}$ is added to \textit{SV} and all its neighbors in $\mathcal{U}$ are added to the tree as a child of $j$ ($\pi$ and $\textbf{pred}$). Then, since there is one element $i\,{=}\,4\,{\in}\,\mathcal{U}$ having a zero transformed cost $\overline{c}_{i,j}\,{=}\,0$, we have ${LU}\,{\setminus}\,{SU}\,{=}\,{LU}\,{=}\,\{4\}$. Since there is one element $j^\prime\,{=}\,4\,{\in}\,\mathcal{V}$ such that $\varrho_{j^\prime}\,{=}\,i$, it is selected as the next $j$ for the next iteration and $i\,{=}\,4$ is then added to \textit{SU}. The resulting tree $T$ is given below, circles represent elements of $\mathcal{V}$ and squares represent elements of $\mathcal{U}$. Black edges are used to represent $T^0$.
\begin{equation*}\begin{small}
\begin{blockarray}{lllllll|l}
	&  &   &     &            & S  &  &\\
	\begin{block}{l\Left{}{(\;}llllll|l<{\;})}
		&8 & 2 & \mathbf{0^*} & 1 & 2 & 1&3\\
		&2 & 4 & 4 & 6 & 5 & 7& 1\\
		&\mathbf{0^*} & 4 & 3 & 1 & 8 &5& 4\\
		S/L&1 & 3 & 2 & \mathbf{0^*} & 0 &0& 1\\
		&2 & 0 & 1 &3 &4 &5 & 3\\\cline{1-8}
		&4 & \mathbf{0}^* & 1 & 7 & 8 &4&0\\
	\end{block}
	\end{blockarray}
		\,,~~\pi=~~ \begin{blockarray}{c}
		~\\
		\begin{block}{\Left{}{(\;\;}c<{\;\;})}
			2\\5\\8\\0\\4\\
		\end{block}
		~
	\end{blockarray} ~,~\textbf{pred}=~~ \begin{blockarray}{c}
		~\\
		\begin{block}{\Left{}{(\;\;}c<{\;\;})}
			5\\5\\5\\5\\5\\
		\end{block}
		~
	\end{blockarray}\,~~~
	\vcenter{\hbox{\includegraphics[scale=0.7]{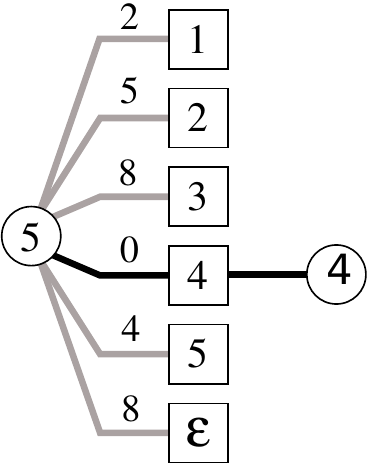}}}
\end{small}\end{equation*}
In the second iteration, while there is no neighbor of $j\,{=}\,4$ in $\mathcal{U}\,{\setminus}\,{LU}$ with a zero reduced cost, the tree is updated if elements of $\mathcal{U}\,{\setminus}\,{LU}$ are connected to $j$ by a strictly lower cost.
\begin{equation*}\begin{small}
\begin{blockarray}{lllllll|l}
	&  &   &     &         S   & S  &  &\\
	\begin{block}{l\Left{}{(\;}llllll|l<{\;})}
		&8 & 2 & \mathbf{0^*} & 1 & 2 & 1&3\\
		&2 & 4 & 4 & 6 & 5 & 7& 1\\
 		&\mathbf{0^*} & 4 & 3 & 1 & 8 &5& 4\\
		S/L& 1 & 3 & 2 & \mathbf{0^*} & 0 &0& 1\\
		&2 & 0 & 1 &3 &4 &5 & 3\\\cline{1-8}
		&4 & \mathbf{0}^* & 1 & 7 & 8 &4&0\\
	\end{block}
	\end{blockarray}
		\,,~~\pi=~\begin{blockarray}{c}
		~\\
		\begin{block}{\Left{}{(\;\;}c<{\;\;})}
			1\\5\\1\\0\\3\\
		\end{block}
		~
	\end{blockarray}\,,~\textbf{pred}=~\begin{blockarray}{c}
		~\\
		\begin{block}{\Left{}{(\;\;}c<{\;\;})}
			4\\5\\4\\5\\4\\
		\end{block}
		~
	\end{blockarray}\,,~~~
	\vcenter{\hbox{\includegraphics[scale=0.7]{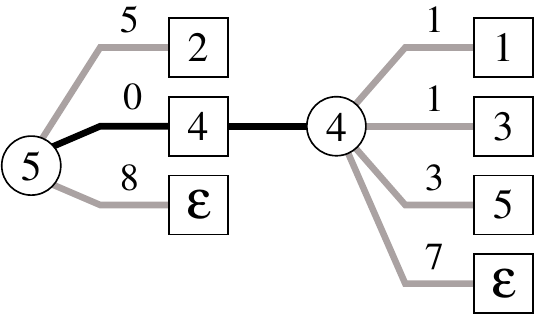}}}
\end{small}\end{equation*}
So a dual updating is performed since ${LU}\,{\setminus}\,{SU}\,{=}\,\emptyset$. The minimum cost is $1$ and so the dual variables become $\mathbf{u}\,{=}\,(0,0,0,-1,0)$ and $\mathbf{v}\,{=}\,(0,0,0,1,1,0)$.
\begin{equation*}\begin{small}
	\begin{blockarray}{llll|ll|ll}
	&  &   &     &         S   & S  &  &\\
	\begin{block}{l\Left{}{(\;}lll|ll|ll<{\;})}
		&8 & 2 & \mathbf{0^*} & \underline{1} & 2 & 1&3\\
		&2 & 4 & 4 & 6 & 5 & 7& 1\\
 		&\mathbf{0^*} & 4 & 3 & \underline{1} & 8 &5& 4\\\cline{1-8}
		S/L& 1 & 3 & 2 & \mathbf{0^*} & 0 &0& 1\\\cline{1-8}
		&2 & 0 & 1 &3 &4 &5 & 3\\
		&4 & \mathbf{0}^* & 1 & 7 & 8 &4&0\\
	\end{block}
	\end{blockarray}~,~~\left\{\begin{aligned}
		&\delta_s\,{=}\,1\\
		&\delta_\epsilon\,{=}\,7\\
		&\delta\,{=}\,\min\{\delta_s,\delta_\epsilon\}\,{=}\,1
	    \end{aligned}\right.~,
	    ~\pi=~~\begin{blockarray}{c}
		~\\
		\begin{block}{\Left{}{(\;\;}c<{\;\;})}
			0\\4\\0\\0\\2\\
		\end{block}
		~
	\end{blockarray}\,,~\textbf{pred}=~~ \begin{blockarray}{c}
		~\\
		\begin{block}{\Left{}{(\;\;}c<{\;\;})}
			4\\5\\4\\5\\4\\
		\end{block}
		~
	\end{blockarray}
\end{small}\end{equation*}
At the end of the iteration, we have the following edit cost matrix and tree, with ${LU}\,{\setminus}\,{SU}\,{=}\,\{1,3\}$. By considering a FIFO strategy, the element $i\,{=}\,1$ is added to \textit{SU} and its corresponding assigned element $j\,{=}\,3\,{\in}\,\mathcal{V}$ is selected for the next iteration. 
\begin{equation*}\begin{small}
	\begin{blockarray}{rlll|ll|ll}
	&  &   &     &         S   & S  &  &\\
	\begin{block}{r\Left{}{(\;}lll|ll|ll<{\;})}
		L&8 & 2 & \mathbf{0^*} & 0 & 1 & 1&3\\
		&2 & 4 & 4 & 5 & 4 & 7& 1\\
 		L&\mathbf{0^*} & 4 & 3 & 0 & 7 &5& 4\\\cline{1-8}
		S/L& 2 & 4 & 3 & \mathbf{0^*} & 0 &1& 2\\\cline{1-8}
		&2 & 0 & 1 &2 &3 &5 & 3\\
		&4 & \mathbf{0}^* & 1 & 6 & 7 &4&0\\
	\end{block}
	\end{blockarray}\,,~~~
	\vcenter{\hbox{\includegraphics[scale=0.7]{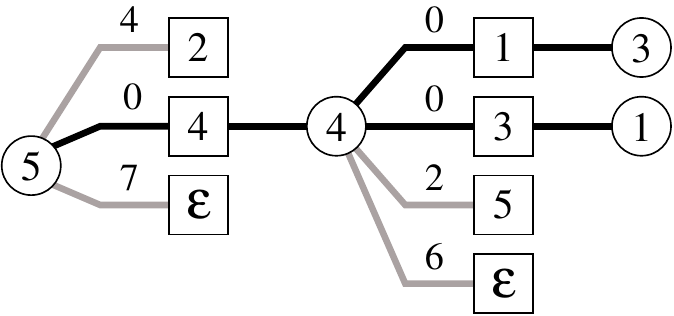}}}~~i\,{=}\,1,~j\,{=}\,3
\end{small}\end{equation*}
In the third iteration, $j\,{=}\,3\,{\in}\,\mathcal{V}$ is added to \textit{SV}, and the tree is updated. Since ${LU}\,{\setminus}\,{SU}\,{=}\,\{3\}$, $i\,{=}\,3$ is added to \textit{SU} and its corresponding element $j\,{=}\,1$ is selected for the next iteration.
\begin{equation*}\begin{small}
	\begin{blockarray}{rllllll|l}
	&  &   &   S  &         S   & S  &  &\\
	\begin{block}{r\Left{}{(\;}llllll|l<{\;})}
		S/L&8 & 2 & \mathbf{0^*} & 0 & 1 & 1&3\\
		&2 & 4 & 4 & 5 & 4 & 7& 1\\
 		L&\mathbf{0^*} & 4 & 3 & 0 & 7 &5& 4\\
		S/L& 2 & 4 & 3 & \mathbf{0^*} & 0 &1& 2\\
		&2 & 0 & 1 &2 &3 &5 & 3\\\cline{1-8}
		&4 & \mathbf{0}^* & 1 & 6 & 7 &4&0\\
	\end{block}
	\end{blockarray}
		\,,~~\pi=~~\begin{blockarray}{c}
		~\\
		\begin{block}{\Left{}{(\;\;}c<{\;\;})}
			0\\4\\0\\0\\1\\
		\end{block}
		~
	\end{blockarray}\,,~~\textbf{pred}=~~\begin{blockarray}{c}
		~\\
		\begin{block}{\Left{}{(\;\;}c<{\;\;})}
			4\\5\\4\\5\\3\\
		\end{block}
		~
	\end{blockarray}
\end{small}\end{equation*}
\begin{equation*}
	\vcenter{\hbox{\includegraphics[scale=0.7]{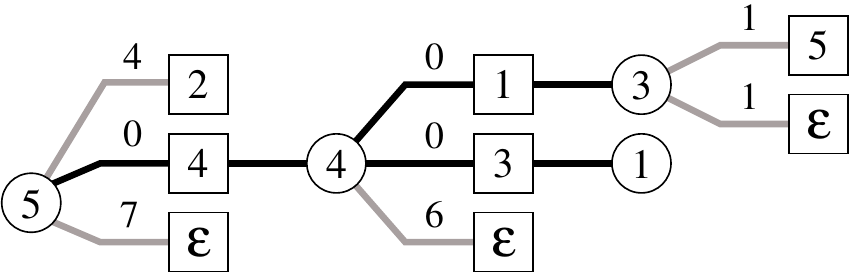}}}
\end{equation*}
In the fourth iteration, $j\,{=}\,1\,{\in}\,\mathcal{V}$ is added to \textit{SV}.
\begin{equation*}\begin{small}
	\begin{blockarray}{rllllll|l}
	& S &   &   S  &         S   & S  &  &\\
	\begin{block}{r\Left{}{(\;}llllll|l<{\;})}
		S/L&8 & 2 & \mathbf{0^*} & 0 & 1 & 1&3\\
		&2 & 4 & 4 & 5 & 4 & 7& 1\\
 		S/L&\mathbf{0^*} & 4 & 3 & 0 & 7 &5& 4\\
		S/L& 2 & 4 & 3 & \mathbf{0^*} & 0 &1& 2\\
		&2 & 0 & 1 &2 &3 &5 & 3\\\cline{1-8}
		&4 & \mathbf{0}^* & 1 & 6 & 7 &4&0\\
	\end{block}
	\end{blockarray}
		\,,~~\pi=~~\begin{blockarray}{c}
		~\\
		\begin{block}{\Left{}{(\;\;}c<{\;\;})}
			0\\2\\0\\0\\1\\
		\end{block}
		~
	\end{blockarray}\,,~~\textbf{pred}=~~\begin{blockarray}{c}
		~\\
		\begin{block}{\Left{}{(\;\;}c<{\;\;})}
			4\\1\\4\\5\\3\\
		\end{block}
		~
	\end{blockarray}
\end{small}\end{equation*}
\begin{equation*}
	\hbox{\includegraphics[scale=0.7]{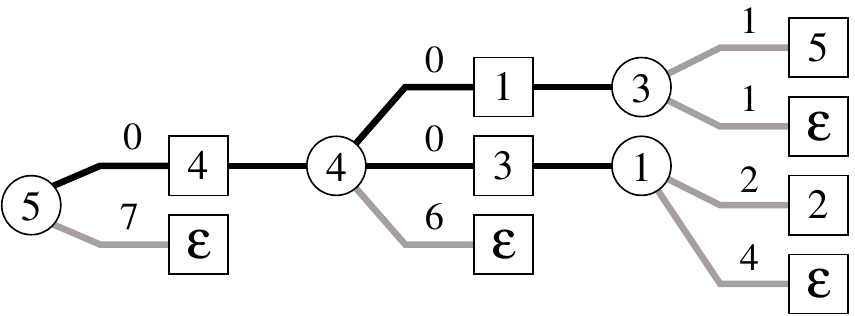}}
\end{equation*}
Since ${LU}\,{\setminus}\,{SU}\,{=}\,\emptyset$, dual variables are updated. The minimum is $\delta\,{=}\,1$ and the dual variables are updated as before: $\mathbf{u}=(-1,0,-1,-2,0)$ and $\mathbf{v}=(1,0,1,2,2,0)$. According to
\begin{equation*}\begin{small}
	\begin{blockarray}{rl|l|lll|ll}
	& S &   &   S  &         S   & S  &  &\\
	\begin{block}{r\Left{}{(\;}l|l|lll|ll<{\;})}
		S/L&8 & 2 & \mathbf{0^*} & 0 & \underline{1} & 1&3\\\cline{1-8}
		&2 & 4 & 4 & 5 & 4 & 7& 1\\\cline{1-8}
 		S/L&\mathbf{0^*} & 4 & 3 & 0 & 7 &5& 4\\
		S/L& 2 & 4 & 3 & \mathbf{0^*} & 0 &1& 2\\\cline{1-8}
		&2 & 0 & \underline{1} &2 &3 &5 & 3\\
		&4 & \mathbf{0}^* & \underline{1} & 6 & 7 &4&0\\
	\end{block}
	\end{blockarray}
		\,,~~\pi=~~\begin{blockarray}{c}
		~\\
		\begin{block}{\Left{}{(\;\;}c<{\;\;})}
			0\\2\\0\\0\\1\\
		\end{block}
		~
	\end{blockarray}\,,~~\textbf{pred}=~~\begin{blockarray}{c}
		~\\
		\begin{block}{\Left{}{(\;\;}c<{\;\;})}
			4\\1\\4\\5\\3\\
		\end{block}
		~
	\end{blockarray}
\end{small}\end{equation*}
the updating of the transformed costs is given by
\begin{equation*}\begin{small}
	\begin{blockarray}{rllllll|l}
	& S &   &   S  &         S   & S  &  &\\
	\begin{block}{r\Left{}{(\;}llllll|l<{\;})}
		S/L&8 & 2 & \mathbf{0^*} & 0 & \underline{1} & 1&3\\
		&1 & 4 & 3 & 4 & 3 & 7& 1\\
 		S/L&\mathbf{0^*} & 4 & 3 & 0 & 7 &5& 4\\
		S/L& 2 & 4 & 3 & \mathbf{0^*} & 0 &1& 2\\
		&1 & 0 & 0 &1 &2 &5 & 3\\\cline{1-8}
		&3 & \mathbf{0}^* & 0 & 5 & 6 &4&0\\
	\end{block}
	\end{blockarray}
		\,,~~\pi=~~\begin{blockarray}{c}
		~\\
		\begin{block}{\Left{}{(\;\;}c<{\;\;})}
			0\\1\\0\\0\\0\\
		\end{block}
		~
	\end{blockarray}\,,~~\textbf{pred}=~~\begin{blockarray}{c}
		~\\
		\begin{block}{\Left{}{(\;\;}c<{\;\;})}
			4\\1\\4\\5\\3\\
		\end{block}
		~
	\end{blockarray}
\end{small}\end{equation*}
Remark that there is two minimal paths. But, according to Algorithm~\ref{algo-augment}, $\epsilon\,{\in}\,\mathcal{U}_\epsilon$ is selected as a sink, and an augmenting path can then be extracted as described in the following section.
\begin{equation*}\begin{small}
	\vcenter{\hbox{\includegraphics[scale=0.7]{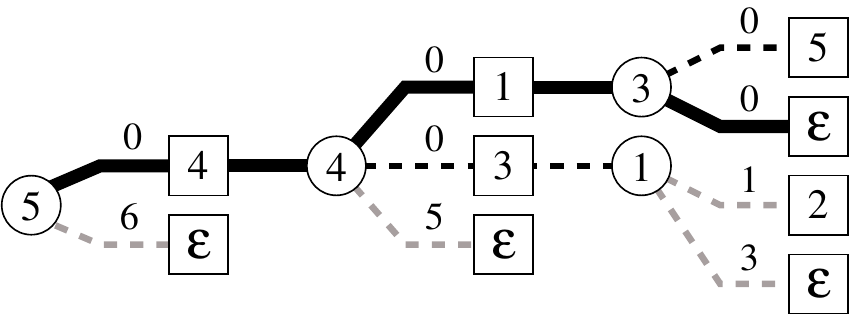}}}~~~~~\begin{blockarray}{rllllll|l}
	& S &   &   S  &         S   & S  &  &\\
	\begin{block}{r\Left{}{(\;}llllll|l<{\;})}
		S/L&8 & 3 & 0 & \mathbf{0^*} & 1 & 2&4\\
		&1 & 4 & 3 & 4 & 3 & 7& 1\\
 		S/L&\mathbf{0^*} & 5 & 3 & 0 & 7 &6& 5\\
		S/L& 2 & 5 & 3 & 0 & \mathbf{0^*} &2& 3\\
		&1 & 0 & 0 &1 &2 &5 & 3\\\cline{1-8}
		&3 & \mathbf{0}^* & \mathbf{0^*} & 5 & 6 &4&0\\
	\end{block}
	\end{blockarray}
\end{small}\end{equation*}
\end{example}
\begin{proposition}\label{prop-oaug}
	Algorithm~\ref{algo-augment} computes an augmenting path in $O((2n+\min\{n,m\})\min\{n,m\})$ time complexity.
\end{proposition}
\begin{proof}
At each iteration of the main loop, one element of $\mathcal{V}$ is added to ${SV}$ (initially empty). For the first iteration this element is provided ($j\,{=}\,k$ with $k$ being the root of the tree), but for all other iterations it is computed from the last element $i\,{\in}\,\mathcal{U}$ added to ${SU}$ (also initially empty). By assuming that a sink is always found by the algorithm, we have $|{SU}|=|{SV}|\,{-}\,1$ since the main loop ends before augmenting ${SU}$. Moreover, since elements are added only once to \textit{SU} and \textit{SV}, we have $|{SU}|\leq|\mathcal{U}|\,{-}\,1$ and $|{SV}|\leq|\mathcal{V}|$. Then we have $|{SV}|=|{SU}|+1\leq|\mathcal{U}|$, and by consequence $|{SV}|\leq\min\{|\mathcal{U}|,|\mathcal{V}|\}$ and $|{SU}|\leq\min\{|\mathcal{U}|,|\mathcal{V}|\}-1$. So the main loop has $O(\min\{n,m\})$ iterations. For each iteration:
\begin{enumerate}
\item[$\bullet$]$\mathcal{U}\,{\setminus}\,\textit{LU}$ is traversed 3 times (lines 9, 17 and 24) and has at most $n$ elements,
\item[$\bullet$]\textit{LU} is traversed 1 time (line 22) and has at most $n$ elements, 
\item[$\bullet$]\textit{SV} is traversed 2 times (lines 19 and 21) and has at most $\min\{n,m\}$ elements. 
\end{enumerate}
So each iteration of the main loop runs in $O(2n+\min\{n,m\})$ time complexity. Indeed, the encoding of the sets ${LU}$, ${LU}\,{\setminus}\,{SU}$ and $\mathcal{U}\,{\setminus}\,{LU}$ with a permutation allows to traverse these sets in $O(n)$ time complexity.
\end{proof}
\noindent
Observe that the time complexity depends on $\min\{n,m\}$. If $n\,{\leq}\,m$ then Algorithm~\ref{algo-augment} runs in $O(n^2)$, else in $O((2n\,{+}\,m)m)$. 

Algorithm~\ref{algo-augment} assigns an unassigned node $k\,{\in}\,\mathcal{V}$. An unassigned node $k\,{\in}\,\mathcal{U}$ can also be assigned by simply swapping the role for $\rho$ and $\varrho$, $\mathbf{u}$ and $\mathbf{v}$, and by considering $\mathbf{C}^T$.
\subsection{Main algorithm}\label{sec-mainalg}
\begin{algorithm}[!t]
\Begin(HungarianLSAPE${(\mathbf{C},\rho,\varrho,\mathbf{u},\mathbf{v})}$){
	$U\leftarrow\{1,\ldots,n\}$,~\,$V\leftarrow\{1,\ldots,m\}$\\
	{\color{dark-gray}// \textit{assign unassigned elements of $V$}}\\
	$(\rho,\varrho,\mathbf{u},\mathbf{v})\leftarrow$assignCols${(\mathbf{C},\rho,\varrho,\mathbf{u},\mathbf{v},U,V)}$\\
	{\color{dark-gray}// \textit{assign unassigned elements of $U$}}\\
	$(\varrho,\rho,\mathbf{v},\mathbf{u})\leftarrow$assignCols${(\mathbf{C}^T,\varrho,\rho,\mathbf{v},\mathbf{u},V,U)}$\\
\Return{$((\rho,\varrho),(\mathbf{u},\mathbf{v}))$}
}
\BlankLine
\Begin(assignCols${(\mathbf{C},\rho,\varrho,\mathbf{u},\mathbf{v},U,V)}$){
\ForEach{$k\,{\in}\,V\,|\,\varrho_k=0$}{
	{\color{dark-gray}// \textit{find an augmenting path rooted in $k$}}\\
	$(i,j,\mathbf{u},\mathbf{v},\mathbf{pred})\leftarrow\text{Augment}(k,\mathbf{C},\rho,\varrho,\mathbf{u},\mathbf{v},U)$\\
	{\color{dark-gray}// \textit{update partial primal solution}}\\
	\textbf{if}~$i=n+1$~\textbf{then}~{$r\leftarrow\varrho_j$,~\,$\varrho_j\leftarrow i$,~\,$i\leftarrow r$}\\
	\textbf{else}~$j\leftarrow 0$\\
	\While{$j\not=k$}{
		$j\leftarrow{pred}_i,~\,\rho_i\leftarrow j$\\
		$r\leftarrow\varrho_j$,~\,$\varrho_j\leftarrow i$,~\,$i\leftarrow r$
	}}
\Return{$(\rho,\varrho,\mathbf{u},\mathbf{v})$}
}
\caption{Compute an $\epsilon$-assignment.\label{algo-hungarian}}
\end{algorithm}
Let $(\rho,\varrho)$ and $(\mathbf{u},\mathbf{v})$ be a partial $\epsilon$-assignment 
and its associated dual variables satisfying Eq.~\ref{eq-compslack}, obtained for instance with Algorithm~\ref{algo-preproc}.

The partial $\epsilon$-assignment is completed by
Algorithm~\ref{algo-hungarian}. To this, each unassigned element of
$\mathcal{V}$ is assigned to an element of $\mathcal{U}_\epsilon$ by
computing an augmenting path (Algorithm~\ref{algo-augment}) and then
by swapping unassigned and assigned edges along the path (2nd step of
assignCols and Proposition~\ref{prop-subprob}), which is realized by a
backtrack from the sink to the root of the tree according to the
predecessor vector and the previous partial $\epsilon$-assignment in
$O(\min\{n,m\})$ time complexity. Since $m$ elements of $\mathcal{V}$
can be unassigned, the first call to assigCols is executed in
$O(m\min\{n,m\}(2n\,{+}\,\min\{n,m\})=O(\min\{n,m\}^2(m\,{+}\,2\max\{n,m\}))$ time complexity.

The second step of Algorithm~\ref{algo-hungarian} consists in
augmenting similarly the remaining unassigned elements of
$\mathcal{U}$ (line 6). This case occurs when $m\,{<}\,n$, or
when more than $\max\{n,m\}\,{-}\,\min\{n,m\}$ elements are assigned
to $\epsilon$ in the first step. In the worst-case, all elements of
$\mathcal{V}$ have been assigned to $\epsilon$ in the first step and
no element of $\mathcal{U}$ have yet been assigned. So the second step
executes in $O(\min\{n,m\}^2(n+2\max\{n,m\}))$ time complexity.
\begin{proposition}\label{prop-comphung}
	Algorithm~\ref{algo-hungarian} solves the LSAPE and its dual problem in $O(\min\{n,m\}^2\max\{n,m\})$ time complexity and in $O(nm)$ space complexity.
\end{proposition}
\begin{proof}
  At each iteration of Algorithm \ref{algo-augment}, the complementary
  slackness condition is satisfied. So Algorithm \ref{algo-hungarian}
  computes a maximal set
  $S\,{=}\,\{(i,\rho_i)\,|\,i\,{\in}\,\mathcal{V}\}\,{\cup}\,\{(\varrho_j,j)\,|\,j\,{\in}\,\mathcal{V},\,\varrho_j\,{=}\,n\,{+}\,1\}$
  of $\epsilon$-independent elements satisfying
  $c_{i,j}\,{=}\,u_i\,{+}\,v_j$ for all $(i,j)\,{\in}\,S$ and
  $c_{i,j}\,{\geq}\,u_i\,{+}\,v_j$ for all $(i,j)\,{\not\in}\,S$. Then
  by Corollary \ref{cor-compslack} the $\epsilon$-assignment is
  optimal. From the above discussion on complexities, we have
  $$
  	\begin{array}{l}
  	\min\{n,m\}^2(m+2\max\{n,m\})+\min\{n,m\}^2(n+2\max\{n,m\})\\
  	=\min\{n,m\}^2(n+m+4\max\{n,m\})\\
  	=\min\{n,m\}^2(\min\{n,m\}+5\max\{n,m\})\\
  	\leq 6\min\{n,m\}^2\max\{n,m\}
  	\end{array}
  $$
  which completes the proof.
\end{proof}
\noindent
Assume that $n\,{\leq}\,m$, the LSAPE is thus solved in $O(n^2m)$ time complexity by Algorithm~\ref{algo-hungarian}. Recall that the LSAPE is equivalent the sLSAPE (Problem~\ref{def-squaredlsap}). The complexities of the proposed algorithm (Proposition~\ref{prop-comphung}) are lower than the ones obtained for solving the sLSAPE with the Hungarian algorithm, \textit{i.e.} $O((n\,{+}\,m)^3)$ in time and $O((n\,{+}\,m)^2)$ in space. Observe that more $n$, $m$, and $|m\,{-}\,n|$ are important, more the improvement is.

\begin{figure*}[!t]\label{fig-xplsap}\centering
	\begin{tabular}{cc}
	\includegraphics[width=0.45\linewidth]{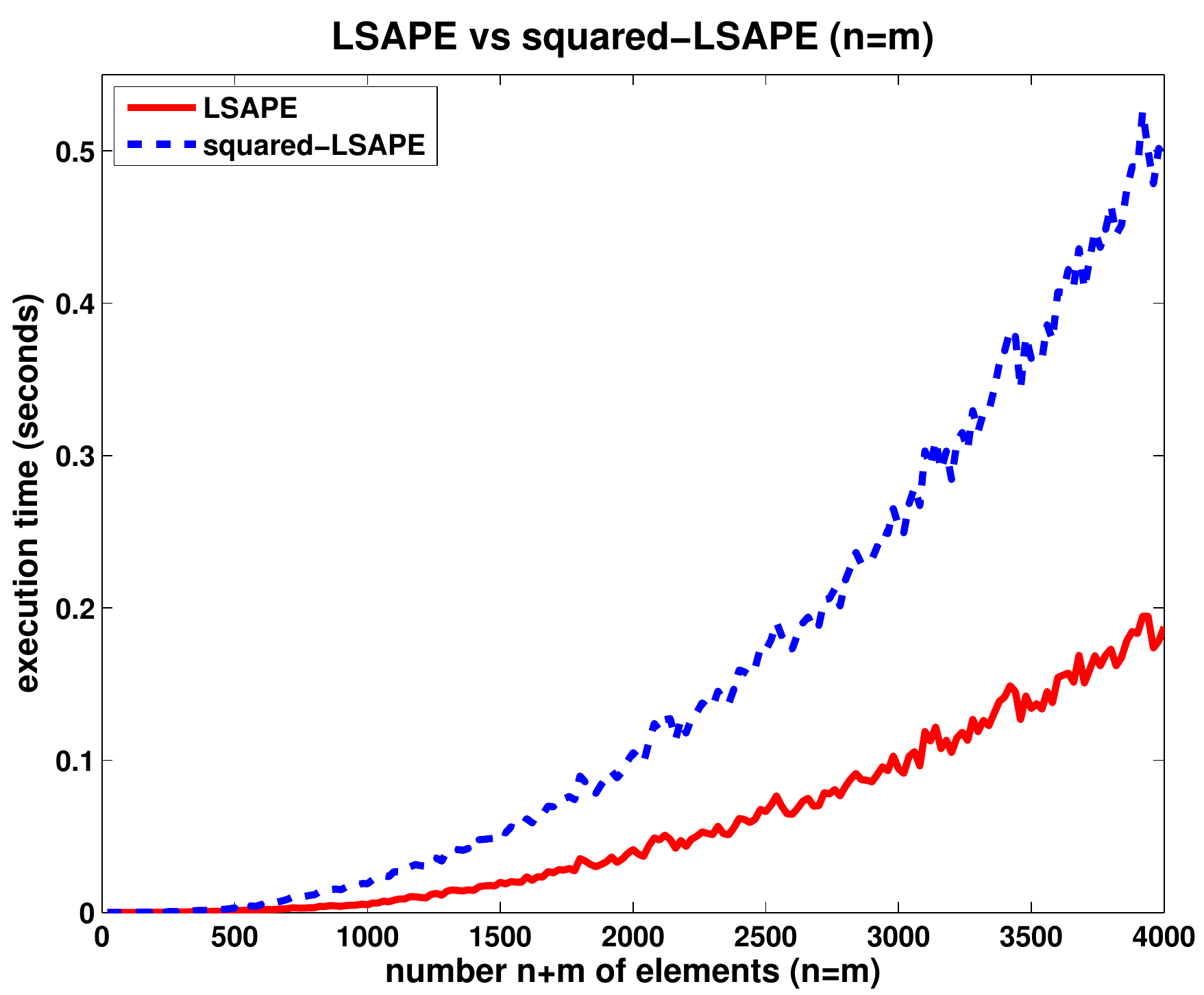}&\includegraphics[width=0.45\linewidth]{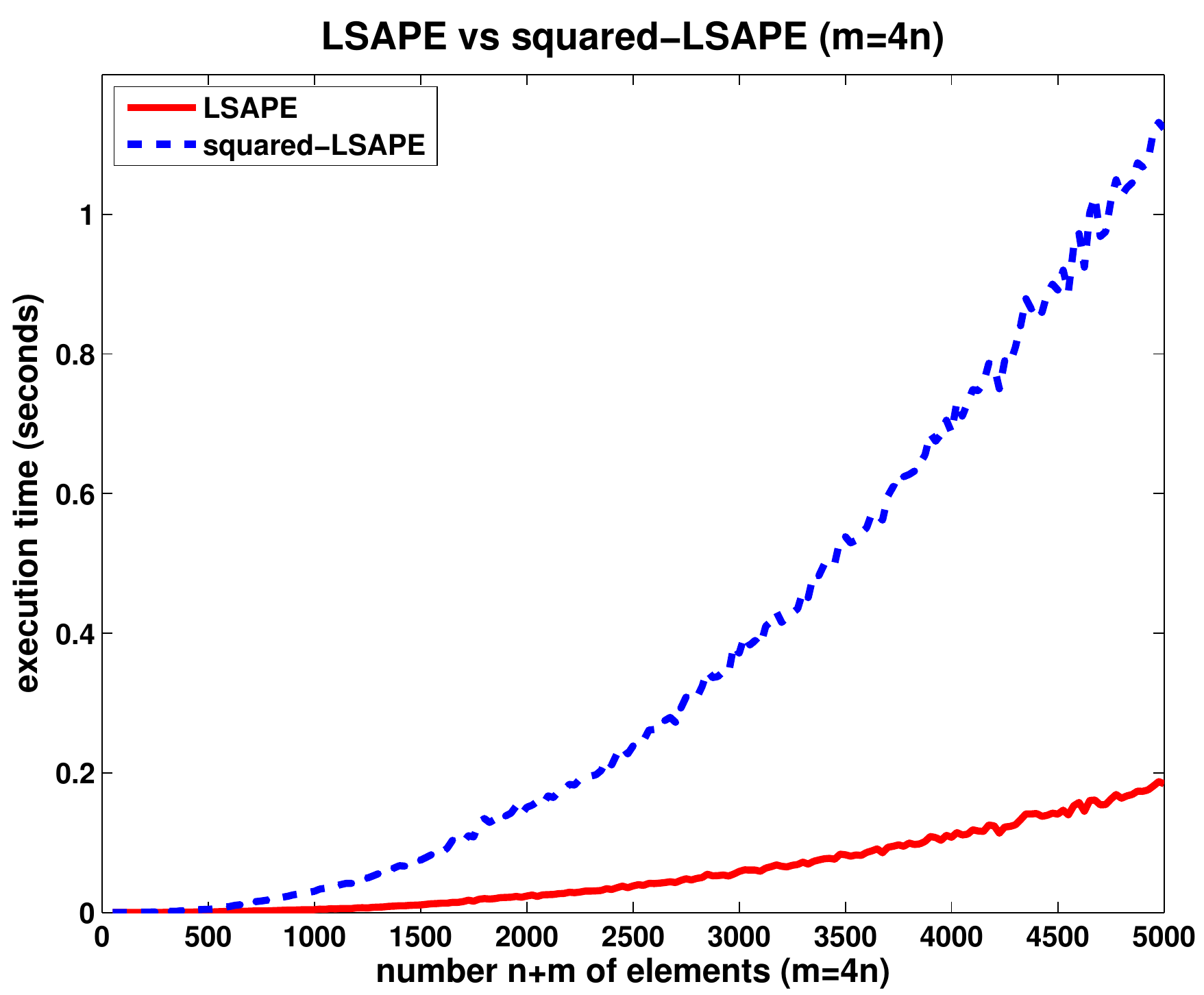}
	\end{tabular}
	\includegraphics[width=0.45\linewidth]{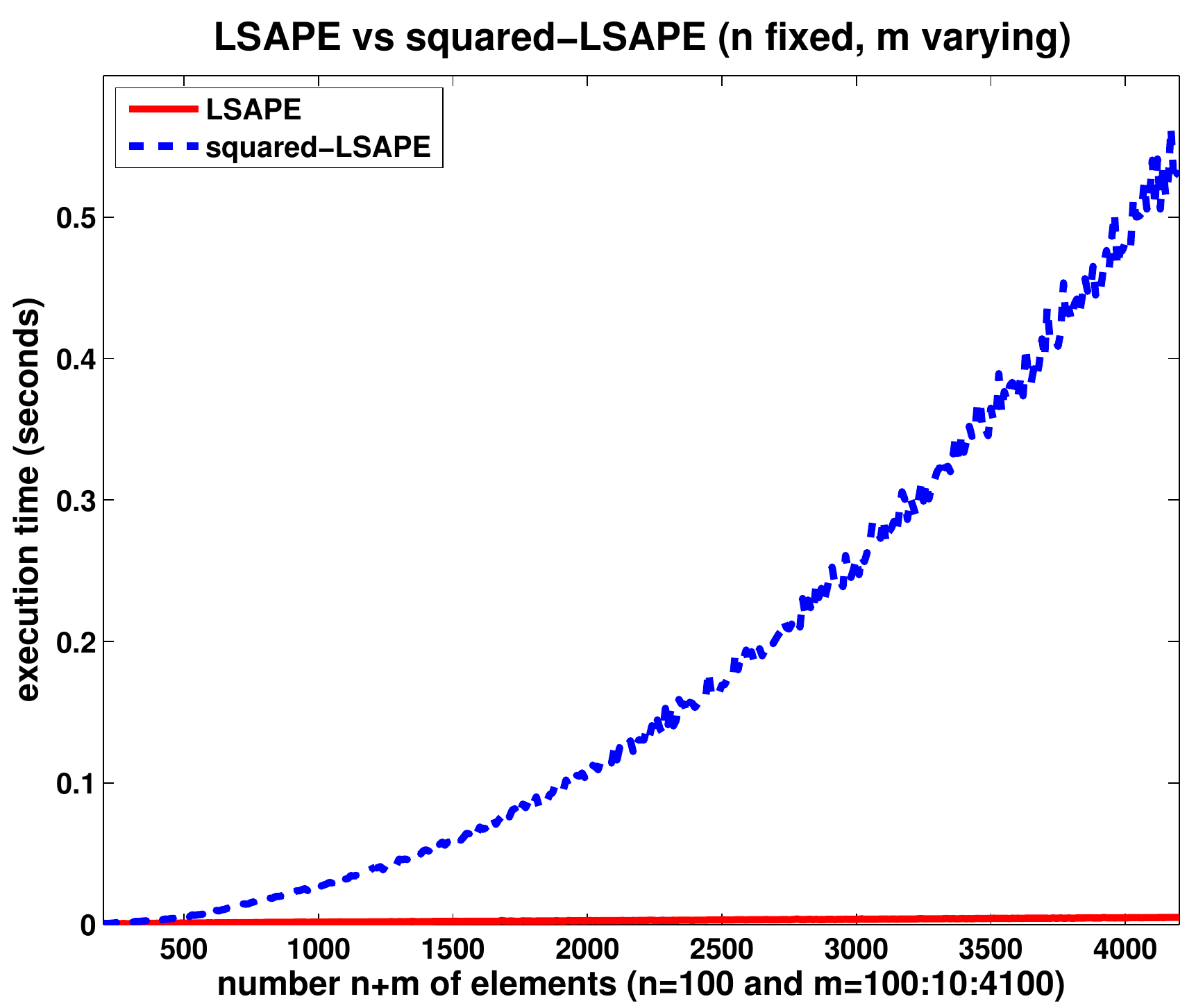}
	\caption{Execution time (in seconds) of the LSAPE (red) vs. squared-LSAPE (blue dotted) for (uniform) random edit cost matrix. For each size $n\,{+}\,m$ ($x$-abscissa), a new edit cost matrix is generated and the time is averaged over several executions on this matrix. First row: both $n$ and $m$ are varying. Second row: $n$ is fixed and $m$ is varying.}
\end{figure*}
\begin{figure*}[!t]\label{fig-xplsap2}\centering
	\begin{tabular}{c}
		\includegraphics[width=0.4\linewidth]{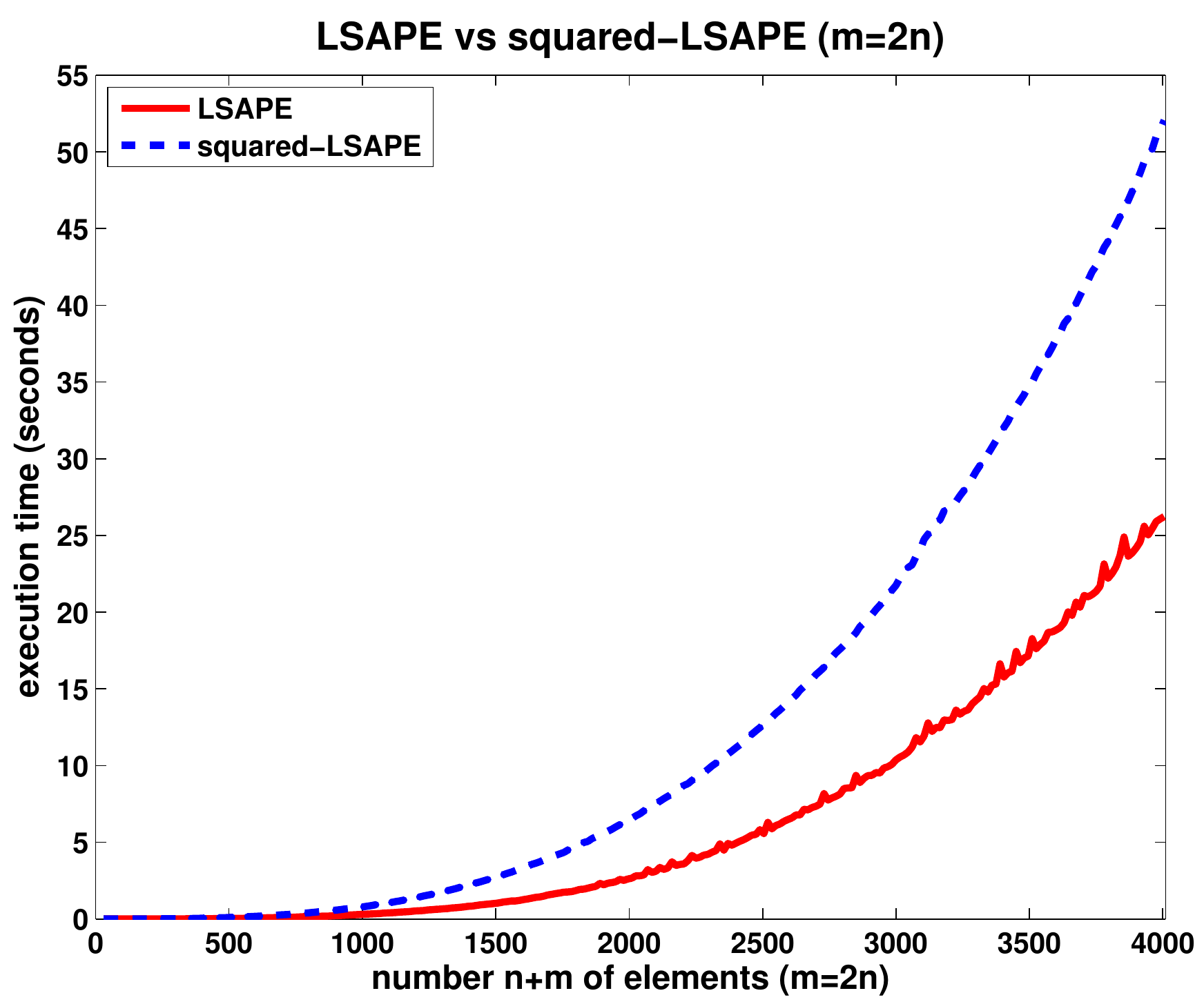}~\includegraphics[width=0.4\linewidth]{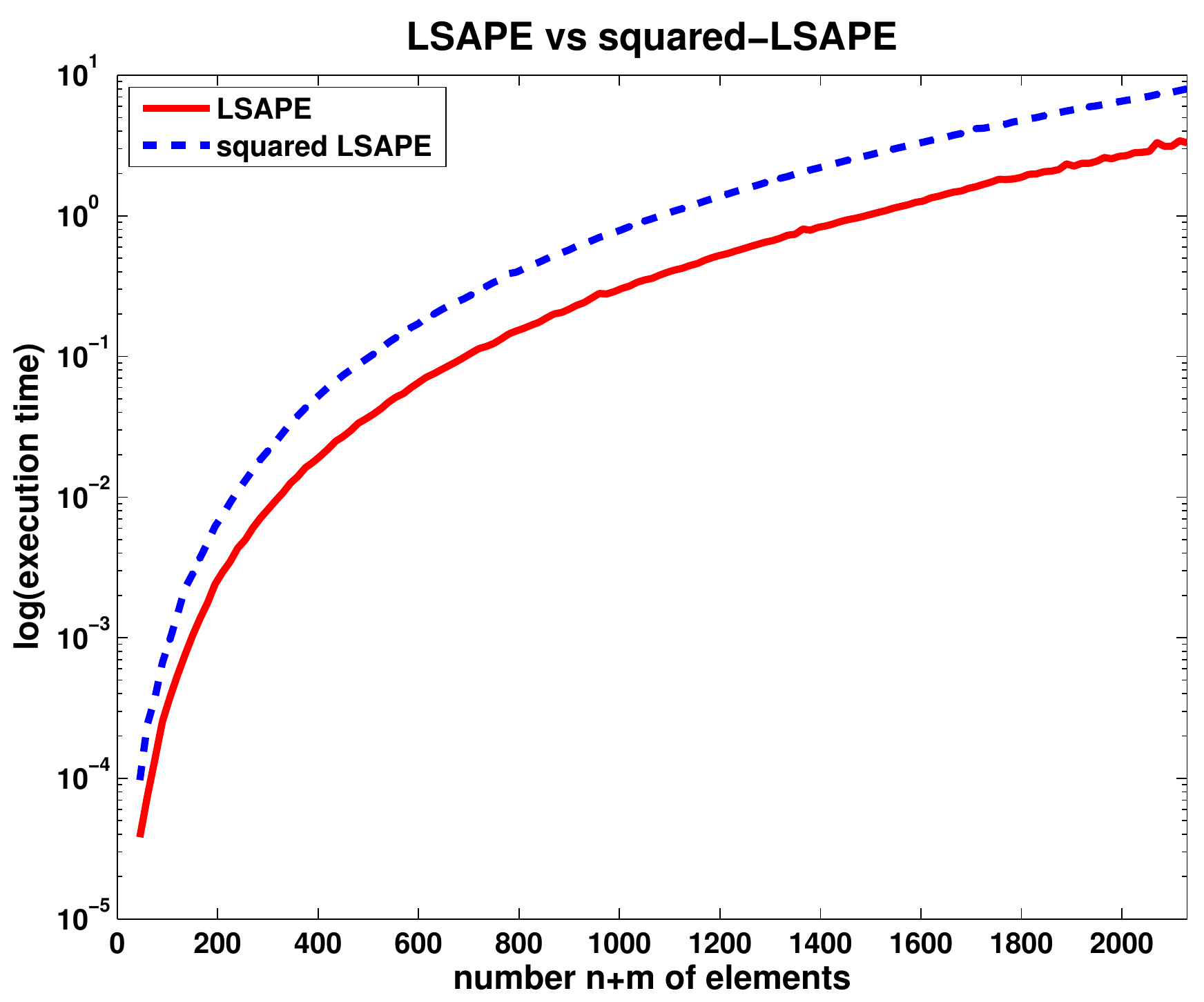}\\
		\includegraphics[width=0.4\linewidth]{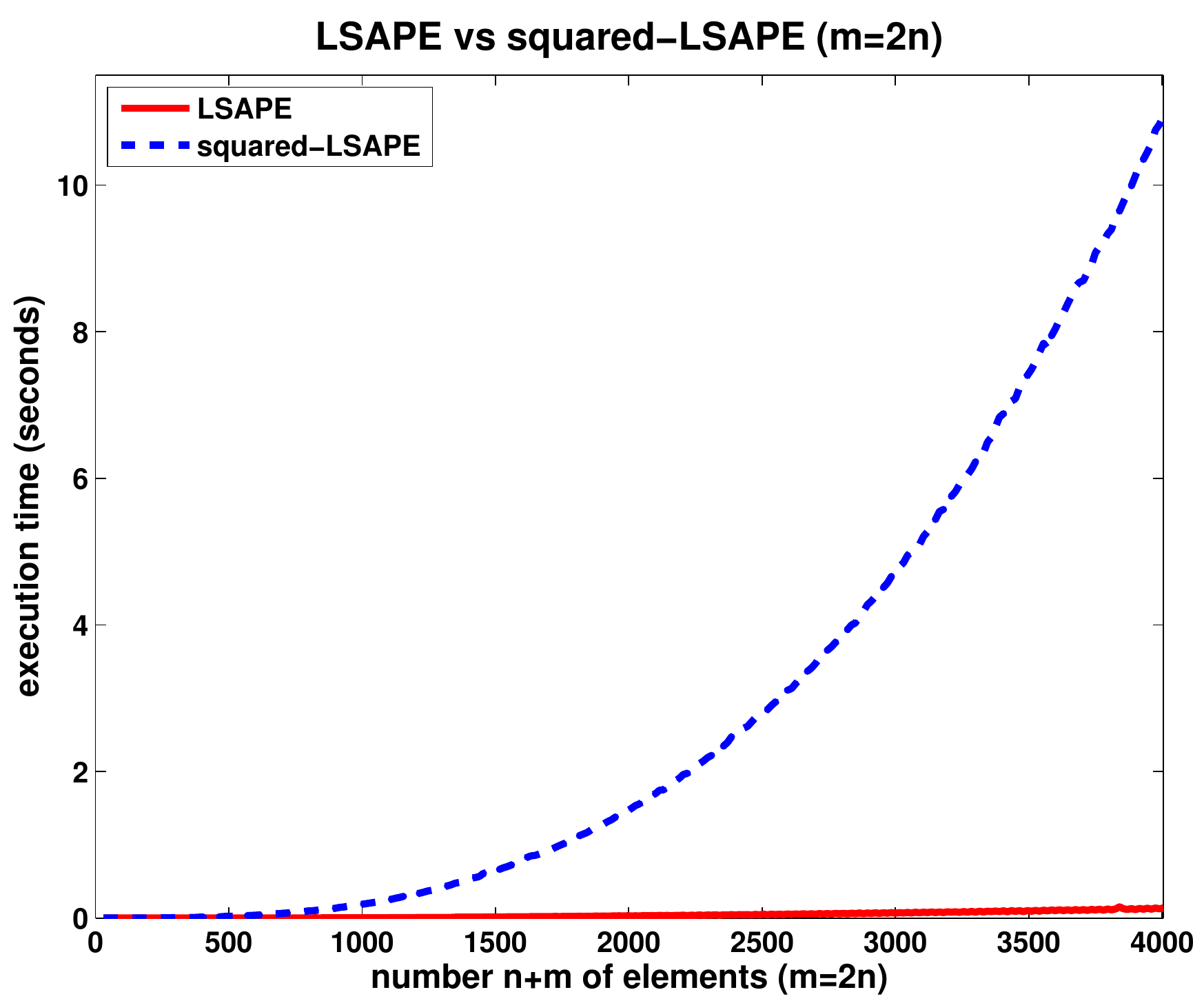}~\includegraphics[width=0.4\linewidth]{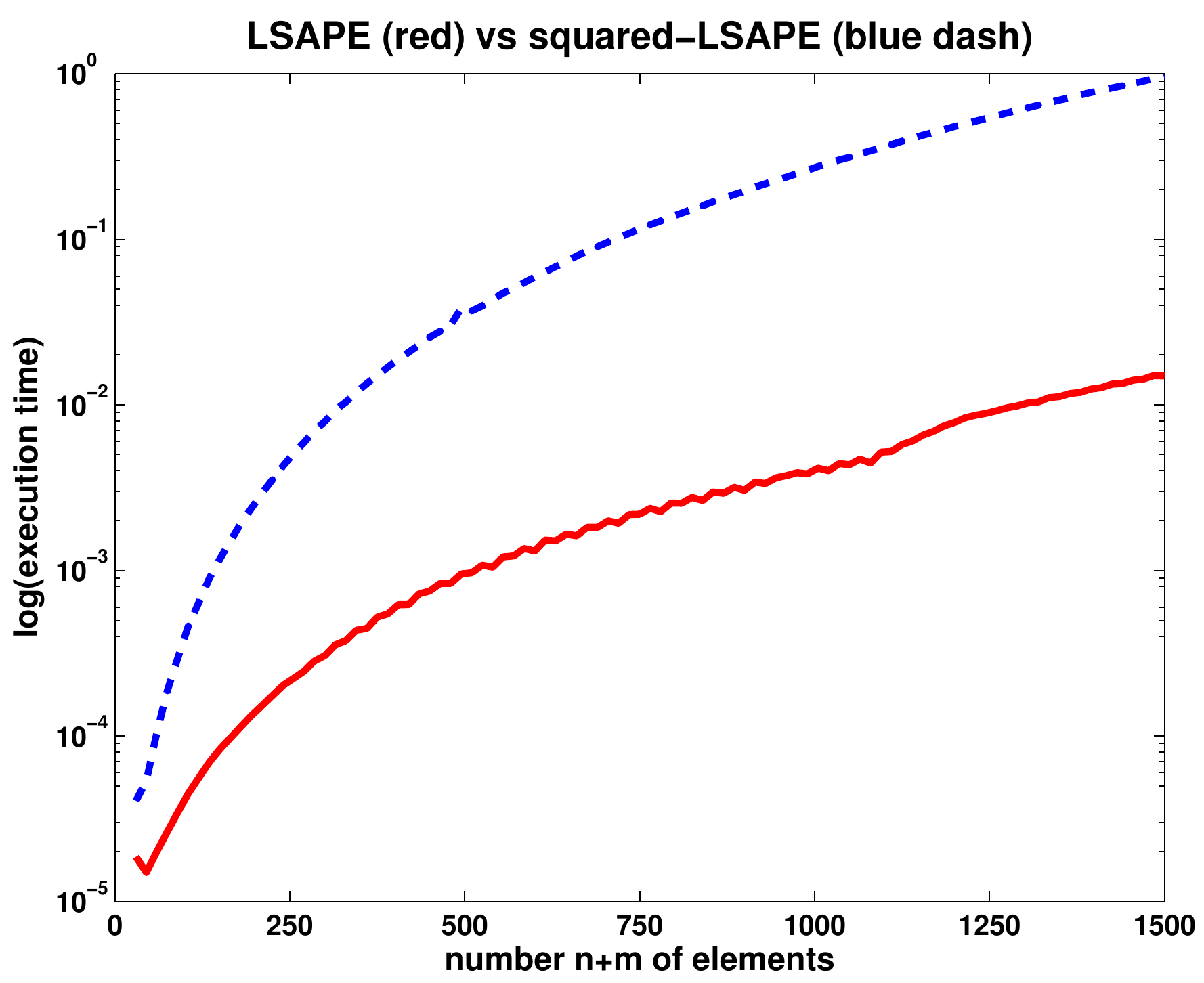}
	\end{tabular}
	\caption{Left column: execution time (in seconds) of the LSAPE (red) vs. squared-LSAPE (blue dotted). Right column: associated log plot. First row: worst-case obtained for $c_{i,j}\,{=}\,ij$. Second row: a variant showing an important improvement (see text).}
\end{figure*}
This is confirmed in practice. We have compared the execution time of the LSAPE and the sLSAPE with three types of synthetic edit cost matrices:
\begin{enumerate}
\item random matrices (Fig.~\ref{fig-xplsap}),
\item matrices of general term  $c_{i,j}\,{=}\,i*j$ (Fig.~\ref{fig-xplsap2}, 1st row),
\item matrices of general term $c_{n-i+1,m-j+1}\,{=}\,i*j$ for 
  $i\,{\in}\,\{1,\ldots,n\}$ et $j\,{\in}\,\{1,\ldots,m\}$. Last line
  and column are respectively defined as copies of the $n-th$ line and
  the $m-th$ column (Fig.~\ref{fig-xplsap2}, 2nd row).
\end{enumerate}
The sLSAPE is solved by an algorithm similar to Algorithm~\ref{algo-hungarian} (see \cite{bur09}) wherein only one set needs to be augmented and all instructions related to $\epsilon$ are not considered. While the LSAPE and the sLSAPE provide the same optimal value $A_\epsilon$, as expected the proposed algorithm runs faster. The difference in execution time between both algorithms increases according to the difference $|m\,{-}\,n|$. This is easily explained by the growth of the number of assignments between $\epsilon$-elements in the sLSAPE.


\section*{Conclusion}
We have presented a general framework to transform a set into another by means of edit operations
(substitutions, removals and insertions). The problem of finding a minimal transformation
is formalized as a linear sum assignment problem where elements can be removed and inserted.
This problem was previously solved by augmenting the given sets such that it can be formalized
as a classical linear sum assignment problem, and thus solved by well-known methods such as
the Hungarian algorithm. Based on the proposed model, the problem is solved by an adaptation
of the Hungarian algorithm that has lower time and memory complexities. Based on the same
model, other algorithms can be adapted similarly.

\appendix
\section{Managing the sets in Algorithm~\ref{algo-augment}}\label{app-st}
Recall that the sets $\mathcal{U}\,{\setminus}\,\textit{LU}$, \textit{LU}, \textit{SU} and $\textit{LU}\,{\setminus}\,\textit{SU}$ are represented by a permutation ${P\mathcal{U}}$ of $\mathcal{U}$ (Fig.~\ref{fig-array}). A link to the beginning of each set in the permutation is also saved.
\begin{example}
  Consider a set $\mathcal{U}\,{=}\,\{1,\ldots,8\}$. At the beginning
  of Algorithm~\ref{algo-augment}, we have
  $P\mathcal{U}\,{=}\,\mathcal{U}$ (Fig.~\ref{fig-array} left),
  ${LU}\,{=}\,{SU}\,{=}\,\emptyset$, and so
  ${LU}\,{\setminus}\,{SU}\,{=}\,\emptyset$. Assume that
  $\pi_4\,{=}\,\pi_6\,{=}\,0$ in the first step, and so $i\,{=}\,4,6$
  are iteratively added to ${LU}$ by swapping each of them with the
  top of $\mathcal{U}\,{\setminus}\,{LU}$ (initially equal to the top
  of $P\mathcal{U}$), and by incrementing the link to the top. Then we
  have ${LU}\,{=}\,\{4,6\}$, the link to the top of
  $\mathcal{U}\,{\setminus}\,{LU}$ is the third element of the array,
  and the link to the top of ${LU}\,{\setminus}\,{SU}$ is the first
  element of the array (Fig.~\ref{fig-array} middle). There is no dual
  update since ${LU}\,{\setminus}\,{SU}\,{\not=}\,\emptyset$, so the
  first iteration of the main loop of Algorithm~\ref{algo-augment}
  ends by selecting the top of ${LU}\,{\setminus}\,{SU}$,
  \textit{i.e.} $i\,{=}\,6$. Then we have ${SU}\,{=}\,\{6\}$,
  ${LU}\,{\setminus}\,{SU}\,{=}\,\{4\}$, and the link to the top of
  ${LU}\,{\setminus}\,{SU}$ is incremented in consequence to become
  the second element of the array (Fig.~\ref{fig-array} right).

\begin{figure}[!t]\centering
	\includegraphics[scale=0.5]{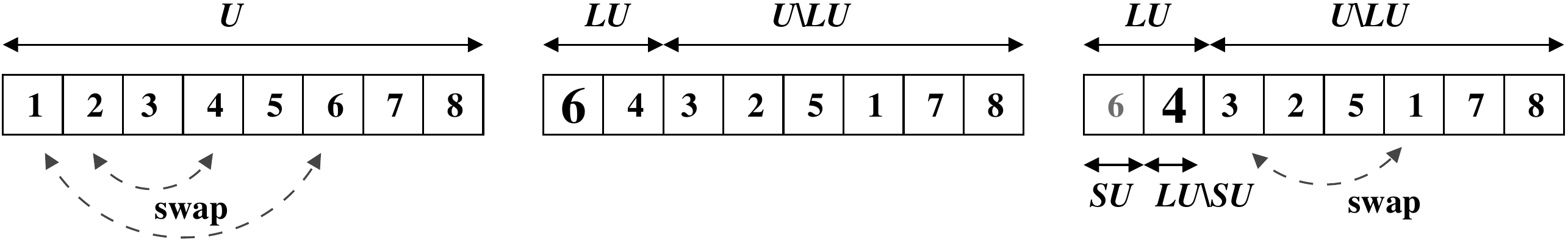}
	\caption{Encoding of the sets involved in Algorithm~\ref{algo-augment} as a permutation $P\mathcal{U}$ of $\mathcal{U}$.}\label{fig-array}
\end{figure}
Now assume that in the second iteration of the main loop, in the first step, we have $\pi_1\,{=}\,0$. The element $i\,{=}\,1$ is added to ${LU}$ by swapping it with the top of $\mathcal{U}\,{\setminus}\,{LU}$, and the process is going on as before until a sink is found.
\end{example}

\footnotesize
\bibliographystyle{plain}
\bibliography{assign}

\end{document}